\documentclass[12pt]{article}

\usepackage{url}
\usepackage{bbm}
\usepackage{mathtools}
\usepackage{amssymb}
\usepackage{amsthm}
\usepackage{empheq}
\usepackage{latexsym}
\usepackage{enumitem}
\usepackage{eurosym}
\usepackage{dsfont}
\usepackage{appendix}
\usepackage{color} 
\usepackage[unicode]{hyperref}
\usepackage{frcursive}
\usepackage[utf8]{inputenc}
\usepackage[T1]{fontenc}
\usepackage{geometry}
\usepackage{multirow}
\usepackage{todonotes}
\usepackage{lmodern}
\usepackage{anyfontsize}
\usepackage{pgfplots}
\usepackage{stmaryrd}
\usepackage{float}
\usepackage{bookmark}

\graphicspath{ {images/} }

\pgfplotsset{compat=newest}

\definecolor{red}{rgb}{0.7,0.15,0.15}
\definecolor{green}{rgb}{0,0.5,0}
\definecolor{blue}{rgb}{0,0,0.7}
\hypersetup{colorlinks, linkcolor={red},citecolor={green}, urlcolor={blue}}
      
\makeatletter \@addtoreset{equation}{section}

\newtheorem{theorem}{Theorem}
\newtheorem{theorem2}{Theorem}[section]

\newtheorem{example}[theorem2]{Example}

\newtheorem{lemma}[theorem2]{Lemma}
\newtheorem{proposition}[theorem2]{Proposition}

\newtheorem{definition}[theorem2]{Definition}
\newtheorem{remark}[theorem2]{Remark}

\DeclareUnicodeCharacter{014D}{\=o}
\usepackage{setspace}

\allowdisplaybreaks

\title{Adaptive trading strategies across liquidity pools\footnote{This work benefits from the financial support of the Chaires Analytics and Models for Regulation, Financial Risk and Finance and Sustainable Development. Bastien Baldacci gratefully acknowledge the financial support of the ERC Grant 679836 Staqamof. The authors would like to thank Joffrey Derchu (Ecole Polytechnique), Mathieu Rosenbaum (Ecole Polytechnique) and Olivier Guéant (Université Paris-1 Panthéon-Sorbonne) for numerous fruitful discussions. In particular, Mathieu Rosenbaum deserves warm thanks for his careful reading of the paper and his many suggestions to improve its quality.}}
\author{Bastien {\sc Baldacci}\footnote{\'Ecole Polytechnique, CMAP, 91128, Palaiseau, France, bastien.baldacci@polytechnique.edu.} \and Iuliia {\sc Manziuk}\footnote{\'Ecole Polytechnique, CMAP, 91128, Palaiseau, France, iuliia.manziuk@polytechnique.edu.} }

\begin{document}

\maketitle 
\begin{abstract}
In this article, we provide a flexible framework for optimal trading in an asset listed on different venues. We take into account the dependencies between the imbalance and spread of the venues, and allow for partial execution of limit orders at different limits as well as market orders. We present a Bayesian update of the model parameters to take into account possibly changing market conditions and propose extensions to include short/long trading signals, market impact or hidden liquidity. To solve the stochastic control problem of the trader we apply the finite difference method and also develop a deep reinforcement learning algorithm allowing to consider more complex settings. \\

\noindent{\bf Keywords:} cross-platform trading, optimal trading, Bayesian learning, adaptive trading strategies, deep reinforcement learning, stochastic control. 

\end{abstract}

\section{Introduction}\label{Section Introduction}

A vast majority of quantitative trading strategies are based on cross-platform arbitrage. These strategies involve cross-listed stocks, that are assets traded on two or more liquidity venues. In~\cite{jain2017hidden}, the authors investigate the prices of cross-listed stocks in different venues and provide evidence of price deviations for the majority of the $600$ cross-listed stocks they studied. In \cite{alsayed2012arbitrage}, the authors highlight mispricing that can exist between a domestic stock and its ADR (American Deposit Receipt) counterpart. The study conducted in \cite{werner1996uk} for US-UK cross-listed stocks shows that markets for cross-listed securities are among the most heavily arbitraged. In particular, higher potential of arbitrage can be exploited for cross-listed stocks from emerging markets, see \cite{rabinovitch2003returns}. \\

Usually, the trader builds an execution curve targeting, for example, an Implementation Shortfall or volume-weighted average price (VWAP). Then, he buys or sells shares of the asset following the execution curve by sending limit and market orders to the different venues. But how to find the best splitting of orders between the venues? The trader splits his orders depending mainly on the imbalance and spread of the different venues, which of course depend on each other. For example, a higher imbalance on the ask side of one venue can indicate a potential imminent price change and a lower probability to have an ask limit order executed, so it may be more profitable to send the order to another venue. The problem of optimal trading across liquidity pools has been treated, for instance, in \cite{almgren2008dynamic,cont2017optimal,laruelle2011optimal,laruelle2013optimal}. In \cite{almgren2008dynamic}, the authors develop a dynamic estimate of the hidden liquidity present on several venues and use this information to make order splitting decisions (Smart Order Routing). The paper \cite{cont2017optimal} solve a general order placement problem and provide explicit solution for the optimal split between limit and market orders on different venues. Finally, in \cite{laruelle2011optimal,laruelle2013optimal}, the authors build a stochastic algorithm to find the optimal splitting between liquidity pools, including dark venues.  \\

Building a good model for optimal trading cross-listed assets requires to take into account the cross-dependence between the imbalance and spread of each venue, as well as the probability and the proportion of execution of limit orders. However, the quality of the model will mainly rely on the estimation of the market parameters. If one assumes constant parameters over the trading period, he believes in the quality of his parameters' estimation. In this case, his strategy is not robust to changes in price dynamics or the platforms' behavior. For example, the trading period may occur when another market participant is executing a buy (or sell) metaorder on one or several venues. This participant will consume the vast majority of the liquidity available on the sell (or buy) side of those venues. If the trader does not adjust his market parameters, the algorithm will keep sending limit orders on the platform where the metaorder is being split, without an execution opportunity. That is why it is essential to update model parameter estimations with new information obtained from observing the market dynamics. Here, we treat the updates in a Bayesian manner. \\

In this paper, we formulate the problem of a trader dealing in a stock, listed on several venues, by placing limit and market orders. Trader's activity can be formulated as a stochastic control problem. The controls are the splitting of volumes between limit and market orders on each venue, and the limits chosen by the trader. The optima are obtained from a classical Hamilton-Jacobi-Bellman (HJB) quasi-variational inequality, which, for a parsimonious model, can be easily solved by grid methods.\\ 

Then we propose a Bayesian update of each market parameter, decoupled from the control problem. One of the advantages of this method is the simplicity of the formulae for each parameter's posterior estimate. In particular, we do not need to use Markov chain Monte-Carlo. This method's choice comes from the fact that updating the market parameters continuously in the control problem increases the number of state variables drastically, leading to high computation time. A continuous Bayesian update would require first to compute the conditional expectation of the value function given the market parameters and then to integrate it over their posterior distribution. This last integration brings multiple non-linearities in the equation, making this fully Bayesian control problem hard to solve numerically. \\

The proposed Bayesian procedure is easier to apply in practice: we divide the trading period into time slices of about several seconds up to a few minutes long, assuming that market conditions do not vary drastically throughout the slice. For each slice, we keep track of all the market events. Specifically, on each venue, we count the number of executed limit orders, the executed proportions (for example, $50\%$ or $100\%$ of the order volume) given the couple spread-imbalance on each venue at the time of the execution. We also keep track of the price dynamics. At the end of each slice, we update our view on the market parameters and recompute the optimal trading strategy for the next slice. This application of the Bayesian updates on slices of execution is time-inconsistent. However, we see it as a first step toward a more integrated Bayesian learning framework for cross-listed trading. By using finite difference schemes or deep reinforcement learning methods (which could also be mixed) for high-dimensional PDE resolution, we can compute in a couple of minutes the optimal trading strategy on a slice given new market conditions. \\

This paper aims at giving a useful and applicable model for practitioners who work on cross-trading strategies. For a quantitative firm, the control model is flexible enough to reproduce the main stylized facts about the market and to design trading strategies taking into account real signals. Moreover, the procedure for Bayesian updates of market parameters in the control problem enables to reevaluate the optimal strategy when the market conditions may differ from the prior empirical estimation of the trader.\\

The article has the following structure: in Section \ref{section_framework}, we describe the framework for cross-platform trading and formulate the trader's optimization problem. In Section \ref{section_HJB_Equation}, we derive the Hamilton-Jacobi-Bellman quasi-variational inequality (HJBQVI) associated with the trader's optimal trading problem. We introduce a change of variable to reduce the dimensionality of the problem and prove the existence and uniqueness of the viscosity solution of the initial HJBQVI in \ref{Section Proof HJBQVI}. In Section \ref{section bayesian model}, we first define the conjugate Bayesian update of all market parameters. Then, we describe the update procedure in practice and its link to the control problem of the trader. Section \ref{sec_extensions} is dedicated to some extensions of the model and their impact on the dimensionality of the resulting HJBQVI. We devote Section \ref{section_results} to numerical results, for the sake of clarity of interpretations considered in the case of limit orders only. Finally in Appendix \ref{sec_Bayesian_OTC}, we present an application of the Bayesian update of the market parameters to the problem of an OTC market maker. 

\section{Optimal trading on several liquidity pools}\label{section_framework}

The model presented in this section is a generalization of the classic optimal trading framework, developed notably in \cite{avellaneda2008high,gueant2013dealing,guilbaud2013optimal,stoikov2009option} and in the reference books \cite{cartea2015algorithmic,gueant2016financial}, to the case of several liquidity venues.   
\subsection{Framework}
We consider a trader acting on $N$ liquidity platforms operating with limit order books over time interval $[0,T]$. He trades continuously on each venue by sending limit and market orders. For $n\in \{1,\dots,N\}$, the $n$-th venue is characterized by the following continuous-time Markov chains:
\begin{itemize}
  \item the bid-ask spread process $(\psi_t^n)_{t \in [0,T]}$ taking values in the state space $\overline{\psi}^n=\{\delta^n,\dots,J\delta^n\}$,
  \item the imbalance process $(I^n_t)_{t \in [0,T]}$ taking values in the state space $\overline{I}^n=\{I^n_1,\dots,I^n_K\}$,
\end{itemize}
where $J, K \in \mathbb{N}$ denote the number of possible spreads and imbalances respectively and $\delta^n$ stands for the tick size of the $n$-th venue. We define the sets $\Psi=\{\Psi_1,\dots,\Psi_{\# \Psi}\}, \mathcal{I}=\{\mathcal{I}_1,\dots,\mathcal{I}_{\# \mathcal{I}}\}$ of disjoint intervals, representing different market regimes of interest in terms of spreads and imbalances.

\begin{example}
Assume for all $n\in \{1,\dots,N\}$ that $\delta^n=\delta$. The set $\Psi=\big\{\delta, \{2\delta, 3\delta\}, \{4\delta, 5\delta\} \big\}$ denotes three spread regimes: low (one tick), medium (two or three ticks), and high (four or five ticks). 
\end{example}
\begin{example}
Assume for all $n\in \{1,\dots,N\}$ and $k\in \{1,\dots,K\}$ that $I^n_k=I_k$. In this case the set $\mathcal{I}=\big\{[-1,-0.66], (-0.66, -0.33], (-0.33, 0.33], (0.33, 0.66], (0.66, 1]\big\}$ denotes five regimes of imbalance: low ($-33\%$ to $33\%$), medium on the ask (resp. bid) from $33\%$ to $66\%$ (resp. from $-66\%$ to $-33\%$) and high on the ask (resp. bid) from $66\%$ to $100\%$ (resp. from $-100\%$ to $-66\%$).
\end{example}
Whenever the spread and the imbalance of each venue enter the state $\mathbf{k}=(\mathbf{k}^\psi,\mathbf{k}^I)\in \mathcal{K}$ where $\mathcal{K}=\prod_{n=1}^N \overline{\psi}^n \times \prod_{n=1}^N \overline{I}^n$, they remain in this state for a time exponentially distributed with mean $\frac{1}{\nu_{\mathbf{k}}}$. We define a transition matrix $\mathbf{P}=(p_{\mathbf{k}\mathbf{k'}})$, $(\mathbf{k},\mathbf{k'})\in\mathcal{K}$, and corresponding intensity vectors $\nu=(\nu_{\mathbf{k}})_{\mathbf{k}}^{\mathbf{T}}$. We assume that $p_{\mathbf{k}\mathbf{k}}=0$, meaning that we cannot come to the same state twice in a row. The infinitesimal generator of the processes can be written as
\begin{align*}
  & r_{\mathbf{k}\mathbf{k'}} = \nu_{\mathbf{k}}p_{\mathbf{k}\mathbf{k'}} \quad \text{if } \mathbf{k}\neq \mathbf{k'} \\
  & r_{\mathbf{k}\mathbf{k}}=-\sum_{\mathbf{k'}\neq \mathbf{k}} r_{\mathbf{k}\mathbf{k'}}=-\nu_{\mathbf{k}}, \text{ otherwise}.
\end{align*}
\begin{remark}
\label{transition}
This general formulation allows us a full coupling between the spread and imbalance of all venues. If one wants a more parsimonious model, the following simplifications could be made. When the spread (imbalance) of the $n$-th venue enters the state $k$, it remains there for an exponentially distributed time with mean $\frac{1}{\nu_k^{n,\psi}}$ ($\frac{1}{\nu_k^{n, I}}$ for the imbalance). Therefore, we define a transition matrix $\mathbf{P}^{n,\psi}=(p^{n,\psi}_{kk'})$, $n\in \{1,\dots,N\}, (k,k')\in \overline{\psi}^n$ such that $p^{n, \psi}_{kk}=0$, and corresponding intensity vectors $\nu^{n, \psi} = (\nu^{n, \psi}_1, \dots, \nu^{n, \psi}_K)^{\mathbf{T}}$. Similarly we define a transition matrix $\mathbf{P}^{n, I}$ for the imbalance. Then, the infinitesimal generator of the processes can be written as
\begin{align*}
  & r^{n, \psi}_{kk'} = \nu_k^{n, \psi}p_{kk'}^{n, \psi} \quad \text{if } k\neq k' \\
  & r^{n, \psi}_{kk}=-\sum_{k'\neq k} r^{n, \psi}_{kk'}=-\nu_k^{n, \psi} \text{ otherwise}.
\end{align*}
This framework will be used in Section \ref{section_results}, where we present the numerical results. 
\end{remark}
In what follows, the trader designs his strategy on the ask side of the market (optimal liquidation problem). The extension to trading on both sides of the market is straightforward and does not cause an increase in the problem's dimensionality. \\

The number of, possibly partially, filled ask orders in the venue $n$ is modeled by a Cox process denoted by $N^{n}, n\in\{1, \dots, N\}$ with intensities $\lambda^{n}\big(\psi_t, I_t, p_t^{n},\ell_t\big)$ where $p_t^{n}\in Q_\psi^n$ represent the limit at which the trader sends a limit order of size $\ell_t^n$, and
\begin{align*}
 & Q_\psi^n=\{0,1\} \text{ if } \psi^n=\delta^n,\text{ and } \{-1,0,1\} \text{ otherwise},\\
 & \mathcal{A}= \Big\{(\ell_t)_{t\in [0,T]}, \mathcal{F}-\text{predictable, s.t for all }t\in [0,T], 0\leq \sum_{n=1}^N \ell_t^n \leq q_t \Big\},
\end{align*}
where $(q_t)_{t\in[0,T]}$ is defined in Equation \eqref{control_problem_trader}. Practically for $n\in \{1,\dots,N\}$, when the spread is equal to the tick size, the trader can post at the first best limit ($p^n=0$) or the second best limit (if $p^n=1$). When the spread is equal to two ticks or more, the trader can either create a new best limit ($p^n=-1$) or post at the best or the second best limit as previously. The arrival intensity of a buy market order at time $t$ on the venue $n\in \{1,\dots,N\}$ at the limit $p\in Q_\psi^n$, given a couple $(\psi_t,I_t)=\mathbf{m}$ of spread and imbalance on each venue, is equal to $\lambda^{n,\mathbf{m},p}>0$. When the trader posts limit orders of volume~$\ell_t^n$ on the $n$-th venue for $n\in \{1,\dots,N\}$, the probability that it is executed is equal to $f^{\lambda}(\ell_t)$, where $f^\lambda(\cdot)\in [0,1]$ is a continuously differentiable function, decreasing with respect to each of its coordinate. Therefore, the arrival intensity of an ask market order filling the buy limit order of the trader on the $n$-th venue at the limit $p_t^n$, given spread and imbalance $(\psi_t,I_t)$ is a multi-regime function defined by
\begin{align*}
  & \lambda^{n}(\psi_t,I_t,p_t^{n},\ell_t)=f^{\lambda}(\ell_t)\sum_{\mathbf{m}\in\mathcal{M},p\in Q_\psi^n} \lambda^{n,\mathbf{m},p}\mathbf{1}_{\{(\psi_t,I_t)\in \mathbf{m}, p_t^{n}=p\}},
\end{align*}
where $\mathcal{M}=\Psi^N \times \mathcal{I}^N$. Moreover, we allow for partial execution, the fact of which we represent by random variables $\epsilon^{n}_t \in [0,1]$. The proportion of executed volume for limit orders in each venue depends on the spread and the imbalance in all $N$ venues, as well as the volume and the limit of the order chosen by the trader. We assume a categorical distribution with $R>0$ different execution proportions $\omega^r, r\in\{1, \ldots, R\}$ for each venue with $\mathbb{P}(\epsilon_t^{n} = \omega^r) = \rho^{n,r}(\psi_t, I_t, p_t^n, \ell_t)$, where
\begin{align*}
  \rho^{n,r}(\psi_t,I_t,p_t^n,\ell_t) = f^\rho (\ell_t)\sum_{\mathbf{m}\in \mathcal{M},p\in Q_\psi^n} \rho^{n,\mathbf{m},p,r}\mathbf{1}_{\{(\psi_t,I_t)\in \mathbf{m},p_t^n=p\}}, 
\end{align*}
where $f^\rho(\cdot)$ is a continuously differentiable function, decreasing with respect to each of its coordinate. 

\begin{remark}
The estimation of this kind of parameters for executed proportions can be quite intricate in practice. To simplify, one can assume that $\rho^{n,r}(\psi_t, I_t, p_t^n, \ell_t) = \rho^{n,r}\in[0,1]$. In practice, this means that there are different execution proportion probabilities inherent by each venue, depending on its toxicity.
\end{remark}
Finally, we allow for the execution of market orders (denoted by a point process $(J_t^n)_{t\in[0,T]}$) on each venue of size $(m_t^n)_{t\in[0,T]}\in [0,\overline{m}]$ where $\overline{m}>0$ and $J_t^n = J_{t^-}^n + 1$. We assume that market orders are always fully executed. \\

The cash process of the trader at time $t \in [0, T]$ is
\begin{align*}
  dX_t=\sum_{n=1}^N \Big(\ell_t^{n}\big(S_t + \frac{\psi^n_t}{2} + p_t^{n}\delta^n\big)\epsilon^{n}_t dN_t^{n} + m_t^{n}\big(S_t-\frac{\psi_t^n}{2}\big)dJ_t^{n}\Big),
\end{align*}
where
\begin{align*}
  dS_t=\mu dt + \sigma dW_t, \quad (\mu,\sigma)\in \mathbb{R}\times \mathbb{R}^+, 
\end{align*}
is the dynamics of the mid-price process. The inventory process of the trader at time $t \in [0, T]$ is defined by
\begin{align}\label{control_problem_trader}
  q_t=q_0-\sum_{n=1}^N \int_0^t \Big(\ell_u^{n}\epsilon_u^{n}dN_u^{n}+\int_0^t m_u^{n}dJ_u^{n}\Big).
\end{align}
We also assume that the trader has a pre-computed trading curve $q^\star$ that he wants to follow (Almgren-Chriss trading curve or VWAP strategy, for example). Then the trader's optimization problem is
\begin{align}
\label{pbmtrader}
  \sup_{(p, \ell, m) \in Q_{\psi}\times \mathcal{A}\times [0,\overline{m}]^N}\mathbb{E}\Big[X_T + q_T S_T - \int_0^T g(q_t-q_t^{\star})dt\Big],
\end{align}
where the function $g$ penalizes deviation from the pre-computed optimal trading curve.

\subsection{The Hamilton-Jacobi-Bellman quasi-variational inequality}
\label{section_HJB_Equation}
The HJBQVI associated with the optimization problem of the trader \eqref{pbmtrader} is the following:
\begin{align}\label{HJBQVI}
\begin{split}
  0= & \min \Bigg\{-\partial_t u(t,x,q,S,\psi,I) + g(q-q_t^{\star}) - \mu\partial_S u - \frac{1}{2}\sigma^2 \partial_{SS}u \\
  & - \sum_{\mathbf{k}\in \mathcal{K}} r_{(\psi,I),(\mathbf{k}^\psi,\mathbf{k}^I)}\big(u(t,x,q,S,\mathbf{k}^\psi,\mathbf{k}^I)-u(t,x,q,S,\psi,I)\big)\\
  & - \sup_{p\in Q_\psi, \ell\in \mathcal{A}} \sum_{n=1}^N \lambda^{n}(\psi,I,p^{n},\ell)\mathbb{E}\Big[u\big(t,x+\epsilon^{n}\ell^{n}(S+\frac{\psi^{n}}{2}+p^{n}\delta^n),q-\ell^{n}\epsilon^{n},S,\psi,I\big) \\
  & - u(t,x,q,S,\psi,I)\Big]; \quad \!\!\sum_{n=1}^N u(t,x,q,S,\psi,I)-\!\!\sup_{m^n\in [0,\overline{m}]}\!\!u\big(t,x+m^n(S-\frac{\psi^n}{2}),q-m^n,S,\psi,I\big)
  \Bigg\},
\end{split}
\end{align}
with terminal condition
\begin{align*}
  u(t,x,q,S,\psi,I)=x+qS,
\end{align*}
where $\psi=(\psi^1,\dots,\psi^N),I=(I^1,\dots,I^N)$. The expectation in \eqref{HJBQVI} is taken over the variables~$\epsilon^n, n\in\{1,\dots,N\}$. We prove the following theorem in \ref{Section Proof HJBQVI}:
\begin{theorem}
  There exists a unique viscosity solution to the HJBQVI \eqref{HJBQVI}, which coincides with the value function of the control problem of the trader \eqref{control_problem_trader}.
\end{theorem}
The proof of existence and uniqueness of the viscosity solution mainly relies on adaptations of the theory of the second order viscosity solution with jumps, see \cite{barles2008second}, for example. \\

The value function has to be linear with respect to the cash process and the mark-to-market value of the trader's inventory due to the form of the terminal condition. Therefore we use the following ansatz: 
\begin{align*}
  u(t,x,q,S,\psi,I)=x+qS + v(t,q,\psi,I).
\end{align*}
The HJBQVI then becomes a system of ODEs with $2N+1$ state variables:
\begin{align}
\label{vfansatz}
\begin{split}
  0= & \min \Bigg\{-\partial_t v(t,q,\psi,I) + g(q-q_t^{\star}) - \mu q \\
  & - \sum_{\mathbf{k}\in \mathcal{K}} r_{(\psi,I),(\mathbf{k}^\psi,\mathbf{k}^I)}\big(v(t,q,S,\mathbf{k}^\psi,\mathbf{k}^I)-v(t,q,S,\psi,I)\big) \\
  & - \sup_{p\in Q_\psi, \ell\in \mathcal{A}} \sum_{n=1}^N \lambda^{n}(\psi,I,p^{n},\ell)\mathbb{E}\Big[\epsilon^{n}\ell^{n}(\frac{\psi^{n}}{2}+p^{n}\delta^n)+v\big(t,q-\ell^{n}\epsilon^{n},\psi,I\big) \\
  & - v(t,q,\psi,I)\Big] ; \quad \sum_{n=1}^N v(t,q,\psi,I)-\sup_{m^n\in [0,\overline{m}]}-m^n\frac{\psi^n}{2} + v\big(t,q-m^n,\psi,I\big)
  \Bigg\}, 
\end{split}
\end{align}
with terminal condition $v(T,q,\psi,I)=0$. \\

Conditionally on the market parameters such as the transition matrix of both the spread and the imbalance processes, the drift and volatility of the underlying asset and the execution proportion probabilities, solving Equation \eqref{vfansatz} is done using simple finite difference schemes and the optimal splitting of volumes as well as the optimal limits can be computed in advance. \\

If one want to incorporate directly Bayesian learning of the parameters in the control problem, the result would be a very high number of state variables, which makes the problem intractable in practice. For example, if we want to update continuously the value of the processes $\lambda^n$ for $n\in \{1,\dots,N\}$ we need to add the counting processes $(N_t^n)_{t\in [0,T]}$ to the state variables, which increases the dimension of the HJBQVI \eqref{vfansatz} by $N$. What we propose in the following section is a practical way to update the market parameters according to trader's observations in a Bayesian way. This method, which is performed separately from the optimization procedure, allows to update, at the end of a slice, the trading strategy according to changing market conditions.

\section{Adaptive trading strategies with Bayesian update}\label{section bayesian model}
The framework presented in the above section allows to choose generic parametric forms for the state variables prior distributions (transition matrix of spreads and imbalances, intensities of orders' arrival on each venue) suitable to the use of conjugate Bayesian updates.
\subsection{Bayesian update of the model parameters}

In this section, we present the conjugate Bayesian update of the market parameters and how to choose the prior distributions. 

\subsubsection{Update of the intensities} \label{section_learn_intensities}
Let us recall the form of the intensities for counterpart market orders' arrival:
\begin{align*}
  & \lambda^{n}(\psi_t,I_t,p_t^{n},\ell_t)=f^\lambda(\ell_t)\sum_{\mathbf{m}\in\mathcal{M},p\in Q_\psi^n} \lambda^{n,\mathbf{m},p}\mathbf{1}_{\{(\psi_t,I_t)\in \mathbf{m}, p^n_t=p\}}.
\end{align*}
In the vast majority of optimal liquidation models, the probability of execution $\lambda^{n,\mathbf{m},p}$ is estimated empirically. We propose to put a prior law $\Gamma(\alpha^{n,\mathbf{m},p},\beta^{n,\mathbf{m},p})$ on the arrival rate, and to update a prior belief at the end of each slice of execution. The parameters $\alpha^{n,\mathbf{m},p},\beta^{n,\mathbf{m},p}$ are chosen by the trader according to his vision of the market before he starts to trade. Up to time $t\in [0,T]$ the trader observes the processes
\begin{align*}
  N_t^{n,\mathbf{m},p}=\int_0^t \mathbf{1}_{\{(\psi_s,I_s)\in \mathbf{m}, p_s^{n}=p\} }dN_s^{n},
\end{align*}
which represent the number of executed orders on each venue for every spread-imbalance zone $\mathbf{m}$. The posterior distribution of $\lambda^{n,\mathbf{m},p}$ for $n \in \{1, \ldots, N\}$ is then given by
\begin{align*}
  \lambda^{n,\mathbf{m}, p} | N_t^{n,\mathbf{m},p} \sim \Gamma\big(\alpha^{n, \mathbf{m}, p} + N_t^{n,\mathbf{m},p}, \beta^{n, \mathbf{m}, p} + \int_0^t f^\lambda(\ell_s)ds\big),
\end{align*}
and at time $t$, our best estimate of the filling ratio becomes 
\begin{align*}
\lambda^{n,\mathbf{m},p}(t,N_t^{n,\mathbf{m},p},\ell_t) = \mathbb{E}\Big[\lambda^{n,\mathbf{m},p} | N_t^{n,\mathbf{m},p} \Big]= \frac{\alpha^{n,\mathbf{m},p} + N_t^{n,\mathbf{m},p}}{\beta^{n,\mathbf{m},p} + \int_0^t f^\lambda(\ell_s)ds}.  
\end{align*}
The posterior estimate of the intensity $\lambda^{n}(\psi_t,I_t,p_t^{n},\ell_t)$ becomes 
\begin{align*}
  \hat{\lambda}^{n}(\psi_t,I_t,p_t^{n},\ell_t)=f^\lambda(\ell_t)\sum_{\mathbf{m}\in\mathcal{M},p\in Q_\psi^n,} \frac{\alpha^{n,\mathbf{m},p} + N_t^{n,\mathbf{m},p}}{\beta^{n,\mathbf{m},p} +\int_0^t f^\lambda(\ell_s)ds}\mathbf{1}_{\{(\psi_t,I_t)\in \mathbf{m}, p_t^{n}=p\}}.
\end{align*}

As the convergence of the prior parameters toward the true market specification follows from the central limit theorem, the convergence rate equals to $\frac{1}{\sqrt{o^\mathbf{m}}}$ where $o^{\mathbf{m}}$ is the number of observations of filled limit orders on the spread-imbalance zone $\mathbf{m}$. If we consider even a quite parsimonious model, for example two venues, two regimes of spread and three regimes of imbalance, we have $\#\mathcal{M}=36$ different zones. This means that we need a sufficiently large amount of observations (large number of executed orders) to get an accurate approximation of the market behavior. \\

If the trader anticipates that the number of observations he will have is not adequate to obtain a suitable approximation of the ``true'' market parameters (in the case of a mid to low frequency strategy with only a few number of trades throughout the day), he might choose at the beginning the couples $(\alpha^{n,\mathbf{m},p},\beta^{n,\mathbf{m},p})$ such that $\frac{\alpha^{n,\mathbf{m},p}}{\beta^{n,\mathbf{m},p}} >> \frac{N_t^{n,\mathbf{m},p}}{\int_0^t f^\lambda(\ell_s)ds}$. That way, his prior will not be sensitive to a small number of observations, and with sufficient number of observations the prior will have less influence and the estimation will be less biased.

\subsubsection{Update of the executed proportion}

We propose to use the Dirichlet prior distribution on the executed proportion parameters so that $\rho^{n, \mathbf{m}, p} \sim \text{Dirichlet}(\alpha^{\epsilon, n, \mathbf{m}, p})$ where $\alpha^{\epsilon, n, \mathbf{m}, p} = (\alpha^{\epsilon, n, \mathbf{m}, p, 1}, \dots, \alpha^{\epsilon, n, \mathbf{m}, p, R})$ for all $(n, \mathbf{m}, p, r)\in \{1,\dots,N\}\times \mathcal{M}\times Q_\psi\times \{1,\dots,R\}$. Given observations of $\epsilon_t^{n}$, the executed proportion parameters have Dirichlet posterior distribution 
\begin{align*}
  \rho^{n, \mathbf{m}, p} \sim \text{Dirichlet} (\alpha^{\epsilon, n, \mathbf{m}, p} + c_t^{n, \mathbf{m}, p}),
\end{align*}
where $c_t^{n, \mathbf{m}, p} = (c_t^{n, \mathbf{m}, p, 1}, \dots, c_t^{n, \mathbf{m}, p, R})$ and $c_t^{n, \mathbf{m}, p, r } = \sum_{s\leq t} \mathbf{1}_{\{\epsilon_s^{n} = \omega^r, (\psi_s, I_s)\in \mathbf{m}, p_s^n = p,N_s^n-N_{s^-}^n=1\}}$ is the number of observations before time $t$ in zone \textbf{m} for a limit $p$ in the venue $n$. Therefore, the $\epsilon_t^{i}$ have the following posterior distribution:
\begin{align*}
  \hat{\rho}^{n,r}(\psi_t, I_t, p_t^n, \ell_t) = f^{\rho}(\ell_t)\sum_{\mathbf{m}\in \mathcal{M}, p\in Q_\psi}\frac{\alpha^{\epsilon, n, \mathbf{m}, p, r}+c_t^{n, \mathbf{m}, p, r}}{\sum_{r = 1}^R (\alpha^{\epsilon, n, \mathbf{m}, p, r}+c_t^{n, \mathbf{m}, p, r})}\mathbf{1}_{\{(\psi_t, I_t)\in \mathbf{m}, p^n = p\}}.
\end{align*}
This Bayesian update is linked to the filling of limit orders of the trader: the proportion executed is updated only if the limit order is (partially) executed. If one chooses the parametrization independent of the spread-imbalance zones and the order volume, that is execution proportion depends only on the venue, the speed of convergence is much faster as the same amount of gathered information is used to update a much smaller number of parameters. Using this more parsimonious parametrization the trader can rely on the observations more than on his prior. 

\subsubsection{Update of the characteristics of the venues}

We observe the states of the Markov chains $\psi_{t_d}, I_{t_d}, d\in \{0,\dots,D\}$ and the times $t_d$ of the $D > 0$ transitions. The likelihood function for the spread and the imbalance processes is
\begin{align*}
  \mathcal{L}(\mathbf{P},\nu|\psi_{t\leq t_D}, I_{t\leq t_D}) &= \prod_{d=1}^D \nu_{t_{d-1}} \exp\big(- \nu_{t_{d-1}} (t_{d} - t_{d-1})\big)p_{({\psi_t}_{d-1}, {I_t}_{d-1})({\psi_t}_{d}, {I_t}_{d})} \\
  & \propto \prod_{\mathbf{k}\in\mathcal{K}} (\nu_{\mathbf{k}})^{n_{\mathbf{k}\cdot}}\exp(-\nu_{\mathbf{k}} T_{\mathbf{k}})\prod_{\mathbf{k}'\in \mathcal{K}} (p_{\mathbf{k}\mathbf{k}'})^{n_{\mathbf{k}\mathbf{k}'}},
\end{align*}
where $n_{\mathbf{k}\mathbf{k}'}$ is the number of observed transitions from state $\mathbf{k}$ to $\mathbf{k}'$ for $(\mathbf{k},\mathbf{k}')\in \mathcal{K}$, $T_{\mathbf{k}}$ is the total time spent in state $\mathbf{k}$, and $n_{\mathbf{k}\cdot}=\sum_{\mathbf{k}'\in \mathcal{K}} n_{\mathbf{k}\mathbf{k}'}$ is the total number of transitions out of state $\mathbf{k}$. \\

Given independent prior distributions for $\mathbf{P},\nu$, the posterior distributions will also be independent. We can carry out Bayesian inference separately on the probability matrix and the intensity vectors of the Markov chains. We assume the following priors: 
\begin{align*}
  & \nu_{\mathbf{k}} \sim \Gamma(a_{\mathbf{k}},b_{\mathbf{k}}), \\
  & \mathbf{p}_{\mathbf{k}} = (p_{\mathbf{k}\mathbf{k}'})_{\mathbf{k}'\in \mathcal{K}} \sim \text{Dirichlet}(\alpha_{\mathbf{k}}), \text{ where } \alpha_{\mathbf{k}}=(\alpha_{\mathbf{k}\mathbf{k}'})_{\mathbf{k}'\in \mathcal{K}}.
\end{align*}
Given these conjugate priors, our best estimators of $\nu_{\mathbf{k}},\mathbf{p}_{\mathbf{k}}$ are 
\begin{align*}
  & \hat{\nu}_{\mathbf{k}} = \frac{a_{\mathbf{k}} + n_{\mathbf{k} \cdot}-1}{b_{\mathbf{k}} + T_{\mathbf{k}} },\\
  & \hat{p}_{\mathbf{k}\mathbf{k}'}=\frac{\alpha_{\mathbf{k}\mathbf{k}'}+n_{\mathbf{k}\mathbf{k}'}}{\sum_{\mathbf{l}\neq k}(\alpha_{\mathbf{k}\mathbf{l}}+n_{\mathbf{k}\mathbf{l}})}.
\end{align*}
Then the posterior transition matrix is 
\begin{align*}
  & \hat{r}_{\mathbf{k}\mathbf{k}'} = \hat{\nu}_{\mathbf{k}}\hat{p}_{\mathbf{k}\mathbf{k}'},\quad \mathbf{k}\neq \mathbf{k}', \\
  & \hat{r}_{\mathbf{k}\mathbf{k}}=-\hat{\nu}_{\mathbf{k}}.
\end{align*}
This update aims at finding the ``true'' behavior of the imbalance and spread processes of each venue. This is of particular importance if an event (for instance, an announcement or news) happens in the market. More specifically, if one event occurs in a particular platform (if a metaorder is executed in one specific platform, for example), this helps to discriminate one venue from the others and to redirect the orders to the less toxic liquidity platforms. Given the large number of observations (transitions from one state of imbalance or spread to another occur fast), the trader does not necessarily need to be confident about his prior distributions. 

\begin{remark}
If one wants to use a more parsimonious model as in Remark \ref{transition}, the same methodology applies. In particular for $k \in \overline{\psi}^n$, we assume the following prior: 
\begin{align*}
  & \nu^{n,\psi}_k \sim \Gamma(a^{n, \psi}_k,b^{n, \psi}_k), \\
  & \mathbf{p}^{n, \psi}_k = (p^{n, \psi}_{kk'})_{k'\in \overline{\psi}^n} \sim \text{Dirichlet}(\alpha_k^{n,\psi}), \text{ where } \alpha_k^{n,\psi}=(\alpha_{kk'}^{n,\psi})_{k'\in \overline{\psi}^n}.
\end{align*}
Given these conjugates priors, our best estimators of $\nu^{n, \psi}_k, \mathbf{p}^{n, \psi}_k$ are 
\begin{align*}
  & \hat{\nu}^{n, \psi}_k = \frac{a^{n, \psi}_k + n^{n, \psi}_{k \cdot}-1}{b^{n, \psi}_k + T^{n, \psi}_k },\\
  & \hat{p}^{n, \psi}_{kk'}=\frac{\alpha_{kk'}+n^{n, \psi}_{kk'}}{\sum_{l\neq k}(\alpha_{kl}+n^{n, \psi}_{kl})}.
\end{align*}
The posterior transition matrix is given by
\begin{align*}
  & \hat{r}^{n, \psi}_{kk'} = \hat{\nu}_k^{n, \psi}\hat{p}_{kk'}^{n, \psi},\quad k\neq k', \\
  & \hat{r}^{n, \psi}_{kk}=-\hat{\nu}_k^{n, \psi}.
\end{align*}
Similar formulae apply for $\nu^{n, I}_k, \mathbf{p}^{n, I}_k$.

\end{remark}

\subsubsection{Update of the mid-price}

We recall that the price process has the following dynamics:
\begin{align*}
  dS_t=\mu dt + \sigma dW_t, 
\end{align*}
so that $(S_t-S_0|\mu,\sigma) \sim \mathcal{N}(\mu t,\sigma^2 t)$. We assume that the couple $(\mu,\sigma^2)$ follows a Normal-Inverse-Gamma prior distribution $NIG(\mu_0,\nu,\alpha^s,\beta^s)$, where $(\mu_0,\nu,\alpha^s,\beta^s)\in \mathbb{R}\times\mathbb{R}^3_+$. Therefore the posterior distribution has the following form:
\begin{align*}
  (\mu, \sigma^2 | S_t-S_0) \sim NIG\Big(\frac{(S_t-S_0)+\mu_0 \nu}{\nu + t},\nu+t,\alpha^s + \frac{t}{2}, \beta^s + \frac{t\nu}{\nu+t}\frac{(\frac{S_t-S_0}{t}-\mu_0)^2}{2}\Big).
\end{align*}
Given our observations of the stock price up to time $t$, the best approximation of the drift and volatility are given by
\begin{align*}
  \mu(t,S_t)=\mathbb{E}[\mu|S_t-S_0]=\frac{(S_t-S_0) + \mu_0 \nu}{\nu + t}, \quad \sigma^2(t,S_t)=\mathbb{E}[\sigma^2|S_t-S_0]= \frac{\beta^s + \frac{t\nu}{\nu+t}\frac{(\frac{S_t-S_0}{t}-\mu_0)^2}{2}}{\alpha^s + \frac{t}{2}-1}.
\end{align*}
The volatility $\sigma$ does not appear explicitly in the HJBQVI \eqref{HJBQVI}. However, it is taken into account when the trader computes his trading curve $q^\star$. \\

In the case where the trader is confident with his estimation of $\sigma$, one can use a Normal prior distribution on $\mu$ such that $\mu\sim \mathcal{N}(\mu_0,\nu^2)$. Then, the best approximation of the drift is given by
\begin{align}\label{update_mu_only}
  \mu(t,S_t)=\mathbb{E}[\mu|S_t-S_0]=\frac{ \mu_0 \sigma^2 + \nu^2 (S_t-S_0)}{\sigma^2 + \nu^2 t}.
\end{align}
If the trader firmly believes in the a priori parameter estimation, he can set $\nu$ close to $0$ so that he mostly relies on his prior. On the contrary, if he sets $\nu$ high enough, his estimation comes mostly from market information. Given the large amount of data coming from the market (each time step corresponding to one new observation), convergence to the real value of the drift is fast. 

\begin{remark}

One can argue about the use of a frequentist estimator of the model parameters, which would actually lead to quite similar formulae. However the original problem, that is continuous update of market parameters in the control problem, is of Bayesian nature. Moreover, in our approach, the formulae for posterior distribution of market parameters are as explicit as in the frequentist approach. 

\end{remark}

\subsection{Algorithm description}
\label{section_quasi_stationary}

We now present the use of the Bayesian updates in order to obtain adaptive trading strategies in practice. We emphasize that the procedure is decoupled from the optimization problem \eqref{vfansatz}, so that we do not perform Bayesian optimization but rather a Bayesian update of the parameters of an optimization problem. \\

Number of time steps is an important parameter of the optimization problem because its choice is a trade-off between computation time and computation precision. To address this problem, we use the trading algorithm with fixed market parameters over a short period of time (a couple of seconds up to a few minutes), which we call a slice. Let us consider $\mathcal{V}>0$ slices $\mathcal{T}_v=[T_v,T_{v+1}], v=0,\dots,\mathcal{V}-1,$ such that $T_0=0,T_{\mathcal{V}}=T$. We define for each slice $v\in \mathcal{V}$ a set of market parameters
\begin{align*}
  \mathbf{\theta}^m_v = (r,\rho^n,\lambda^{n,\mathbf{m},p},\mu,\sigma)_{\left\{n\in\{1,\dots,N\}, \mathbf{m}\in\mathcal{M},p\in Q_\psi \right\}}.
\end{align*}
At each time slice $v \in \{0, \mathcal{V}-1\}$ starting from $v=0$ we perform the following algorithm: 
\begin{enumerate}
  \item Take the best estimation of market parameters $\theta^m_v$ from the prior distribution for the current slice $v$. 
  \item Compute the optimal trading strategy on $\mathcal{T}_v$ using the set of parameters $\theta^m_v$.
  \item Observe market events during the current slice (executions, changes of the state).
  \item At $T_{v + 1}$, update the parameters $\theta^m_{v + 1}$ following the Bayes rules described in Section \ref{section bayesian model}.
\end{enumerate}

To summarize, we use the output of the control model (the optimal volumes and limits in each venue) over a slice of execution and then run the model again with the updated market parameters. This method, which is clearly time inconsistent, is common practice when one applies optimal control with online parameter estimation, see for example \cite{baldacci2019algorithmic}. We now present some possible extensions of the presented model.

\section{Model extensions}\label{sec_extensions}
In this section we describe different potential model extensions and their impact on the problem's dimensionality. 
\subsection{Extension 1: Incorporation of signals in the price process}

\subsubsection{Short-term price signals}

The two main sources of signals at the microstructural level are the imbalance and the bid-ask spread. Therefore, one can assume a parametric dependence $f^{\text{short}}(\psi_t, I_t)$ of the price process on these two sources, such that the price process becomes
\begin{align*}
  dS_t = \big(\mu + f^{\text{short}}(\psi_t, I_t)\big)dt + \sigma dW_t. 
\end{align*}
In a modified stochastic control problem the term $\mu q$ in the HJBQVI is replaced by $(\mu + f^{\text{short}}(\psi, I))$, which causes no increase in the dimensionality of the state process. 

\subsubsection{Mid/Long term and path-dependent price signals}

When trading on longer time horizon, one can incorporate mid- or long-term signals such as Bollinger bands, moving average or cointegration ratio. For example, consider a signal taking into account the moving average and the maximum of the price process $S_t$, that is 
\begin{align*}
  \overline{S}_t = \frac{1}{t}\int_0^t S_t dt, \quad S^{\star}_t =\max_{s\leq t} S_s.
\end{align*}
The triplet $(\overline{S}_t,S_t^\star,S_t)$ is Markovian. Therefore, we can add a long term signal $f^{\text{long}}(S_t,\overline{S}_t,S^{\star}_t)$ into the asset's drift:
\begin{align*}
  dS_t = \big(\mu + f^{\text{long}}(S_t,\overline{S}_t,S^{\star}_t)\big)dt + \sigma dW_t. 
\end{align*}
The HJBQVI then becomes:
\begin{align*}
  0= & \min \Bigg\{-\partial_t u(t,q,S,\overline{S},S^\star,\psi,I) + g(q-q_t^{\star}) - \big(\mu+f^{\text{long}}(S,\overline{S},S^{\star})\big)\partial_S u - \frac{S-\overline{S}}{t}\partial_{\overline{S}}u - \frac{1}{2}\sigma^2 \partial_{SS}u \\
  & - \sum_{\mathbf{k}\in \mathcal{K}} r_{(\psi,I),(\mathbf{k}^\psi,\mathbf{k}^I)}\big(u(t,q,S,\overline{S},S^\star,\mathbf{k}^\psi,\mathbf{k}^I)-u(t,q,S,\overline{S},S^\star,\psi,I)\big) \\
  & - \sup_{p\in Q_\psi, \ell\in \mathcal{A}}\sum_{n=1}^N\lambda^{n}(\psi,I,p^{n},\ell)\mathbb{E}\Big[\epsilon^{n}\ell^{n}(S+\frac{\psi^{n}}{2}+p^{n}\delta^n)+u\big(t,q-\ell^{n}\epsilon^{n},S,\overline{S},S^\star,\psi,I\big) \\
  & - u(t,q,S,\overline{S},S^\star,\psi,I)\Big]; \quad \!\!\!\!\!\sum_{n=1}^N u(t,q,S,\overline{S},S^\star,\psi,I)-\!\!\!\!\!\sup_{m^n\in [0,\overline{m}]}\!\!\!\!\!m^n(S-\frac{\psi^n}{2}) + u\big(t,q-m^n,S,\overline{S},S^\star,\psi,I\big)
  \Bigg\},
\end{align*}
for $S\leq S^\star$, with $\partial_{S}u = 0$ for $S=S^\star$. To obtain this equation we just use a change of variable $v(t,x,q,S,\overline{S},S^\star,\psi,I)= x+ u(t,q,S,\overline{S},S^\star,\psi,I)$, linear with respect to the cash process $X_t$. We end up with a $2N+4$ dimensional HJBQVI, that we can still solve using our deep reinforcement learning algorithm (but unlikely with finite differences). \\

More generally, adding a path-dependent state variable that gives information on the price trend adds one dimension to the HJBQVI (in the example above $(\overline{S}_t, S_t^\star, S_t)$ add one dimension each). 

\subsection{Extension 2: Market impact}

So far we assumed no market impact on the price process. It is common knowledge that cost of market impact can cut down a large proportion of the trading strategy's profit. Therefore, we can use a simple permanent-temporary market impact model, inspired by \cite{almgren2005h}. \\

The impacted mid-price process can be modeled as follows:
\begin{align*}
  dS_t= \big(\mu + h(\ell_t) \big) dt + \sigma dW_t + \sum_{n=1}^N \big( \xi^{n,l}(t,\ell_t^{n}) dN_t^{n} + \xi^{n,m}(t,\ell_t^{n}) dJ_t^{n} \big),
\end{align*}
where the functions $h,\xi^{n,l},\xi^{n,m}$ are the permanent and temporary market impact functions. Following~\cite{gatheral2010no}, we assume linear permanent market impact, that is 
\begin{align*}
   h(\ell_t) = \sum_{n=1}^N \kappa^{n,\text{per}} \ell_t^n, \quad \kappa^{n,\text{per}} >0 \text{ for all } n\in \{1,\dots,N\}.
\end{align*}
For the temporary market impact, we can follow the well-known ``square-root law'' and set
\begin{align*}
  \xi^{n,l}(t,\ell_t^{n}) = \kappa^{n,l} (\ell_t^n)^{\gamma^{n,l}}, \quad \xi^{n,m}(t,\ell_t^{n}) = \kappa^{n,m} (\ell_t^n)^{\gamma^{n,m}}, 
\end{align*}
where $\kappa^{n,l},\kappa^{n,m},\gamma^{n,l},\gamma^{n,m} >0$ and $\gamma^{n,l},\gamma^{n,m} \approx 1/2$. On the other hand, in order to take into account the transient part of the impact, we can set the following form for $S_t$:
\begin{align}\label{price impacted}
  S_t = S_0 + \int_0^t \mu + h(\ell_s) ds + \sigma W_t + \sum_{n=1}^N \int_0^t \xi^{n,l}(t-s)\tilde{\xi}^{n,l}(\ell_s^i)dN_s^n + \xi^{n,m}(t-s)\tilde{\xi}^{n,m}(\ell_s)dJ_s^n, 
\end{align}
where $\xi^{n,l}, \xi^{n,m}$ are decreasing kernels, and $\tilde{\xi}^{n,l}, \tilde{\xi}^{n,m}$ are decreasing functions of the posted volume. It is well known that by taking an exponentially decreasing kernel, Equation \eqref{price impacted} admits a Markovian representation as the couples $\big(N_t^{n},\int_0^t \xi^{n,\{l,m\}}(t-s)dN_s^n\big)_{t\in [0,T]}$ are Markovian. Practically, this will add $2N$ dimension to the HJBQVI. \\

Functions $h, \xi^{n,l}, \xi^{n,m}$ could also be approximated by neural networks. Determination of a cross-impact function between liquidity pools can lead to possible arbitrage detection across liquidity venues.

\subsection{Extension 3: Hidden liquidity}

Hidden liquidity represents a great proportion of the liquidity especially in the US markets, see for example \cite{jain2017hidden}. Therefore, if one wants to design trading tactics for assets cross-listed in a European and an American market, taking into account the hidden part of the liquidity is crucial. \\

Assume that the $n$-th venue is a US liquidity pool. Borrowing the notations of \cite{avellaneda2011forecasting}, we denote by $H^n$ the hidden liquidity of the $n$-th venue at the first limit of the order book. Therefore, the corresponding imbalance process represented by the continuous-time Markov chain $I^n$ can be rewritten as $\frac{N_t^{n, a, m} - N_t^{n, b, m}}{N_t^{n,a, m} + N_t^{n, b, m} + 2H^n},$ where $N_t^{n, b, m}, N_t^{n, a, m}$ are the bid and ask market order flow processes on the $n$-th venue. Empirical estimation of the prior parameters for the transition matrix of $I^n$ have to take into account this additional term in the imbalance processes. Furthermore, incorporating the imbalance process with hidden liquidity into trading signals allows to detect arbitrage opportunities between different venues. This does not increase the dimensionality of Equation \eqref{HJBQVI}.

\section{Numerical results}\label{section_results}

\subsection{Global parameters}

We take the example of a trader acting on a stock cross-listed on $2$ different venues ($N=2$), with the following global parameters:
\begin{itemize}
  \item $\overline{\psi}^n = \{\delta, 2\delta\}$: the processes $(\psi_t^n)_{t\in [0, T]}$ can take two values, which correspond to a low or high spread regimes, and the tick size is $\delta=0.05$.
  \item $\overline{I}^n = \{-0.5, 0, 0.5\}$: the processes $(I_t^n)_{t\in [0, T]}$ can take three values, which correspond to a negative, neutral or positive imbalance regime. 
  \item $R=2, (\omega^1, \omega^2)=(0.5,1)$: the processes $(\epsilon_t^n)_{t\in [0,T]}$ can take two values, which correspond to a total or half-execution of the posted volume $(\ell_t^n)_{t\in [0,T]}$.
  \item $q_0=5\times 10^4$: initial inventory of the trader.
  \item $\mathcal{T}_v = [v, v+\Delta_v]$, where $\Delta_v = 1$ min, which means that each slice lasts one minute, with $\mathcal{V}=10$ slices and $T=10$ min. 
  \item $\Delta_t = 0.1$: we take 10 time steps in each slice, that is the agent takes $10$ trading decisions during each slice. 
\end{itemize}
The pre-computed trading curve is borrowed from an implementation shortfall execution using market orders, that is:
\begin{align*}
 q^\star_t = q_0 \frac{\sinh\Big(\sqrt{\frac{\gamma\sigma^2 V}{2\eta}}(T-t) \Big)}{\sinh\Big(\sqrt{\frac{\gamma\sigma^2 V}{2\eta}}T\Big)}. 
\end{align*}
%\item VWAP target: $q^\star_t = q_0(1-\frac{t}{T}) - q_0\frac{k}{\gamma\sigma^2 T}\sinh\Big(\sqrt{\frac{\gamma\sigma^2 V}{2\eta}}t\Big)\bigg(\sinh\Big(\sqrt{\frac{\gamma\sigma^2 V}{2\eta}}\frac{T}{2}\Big)-\sinh\Big(\sqrt{\frac{\gamma\sigma^2 V}{2\eta}}\frac{t}{2}\Big) \bigg)$

with the following set of parameters
\begin{itemize}
  \item $\eta = 0.1$: coefficient of quadratic costs.
  \item $V=1\times 10^8$: average market volume. 
  \item $\gamma = 1\times 10^{-6}$: risk aversion of the trader using a CARA utility function.
  \item $\sigma = 0.05$: volatility of the asset.
  \item $f^\lambda(\ell_t)=\exp(-\kappa \sum_{n=1}^N \ell_t^n)$ with $\kappa = 2.5 \times 10^{-5}$: sensitivity of the execution with respect to the total volume posted. 
  \item $f^\rho(\ell_t)=1$: no sensitivity of the executed proportion with respect to the total volume posted. 
\end{itemize}

For this numerical experiment for the sake of clarity of interpretations we consider the trader sending only limit orders.

\subsection{Numerical methods}

\subsubsection{Finite differences}

To find optimal strategy for limit orders we consider the following equation:
\begin{align*}
0= & -\partial_t v(t,q,\psi,I) + g(q-q_t^{\star}) - \mu q \\
  & - \sum_{n=1}^N \sum_{j=1}^J r^{n,\psi}_{\psi,j\delta}\big(v(t,q,\psi_{j\delta}^{-n},I)-v(t,q,\psi,I)\big) - \sum_{n=1}^N \sum_{k=1}^K r^{n,I}_{I,I_k}\big(v(t,q,\psi,I_{I_k}^{-n})-v(t,q,\psi,I)\big) \\
  & - \sup_{p\in Q_\psi, \ell\in \mathcal{A}} \sum_{n=1}^N \lambda^{n}(\psi,I,p^{n},\ell)\mathbb{E}\Big[\epsilon^{n}\ell^{n}(\frac{\psi^{n}}{2}+p^{n}\delta^n)+v\big(t,q-\ell^{n}\epsilon^{n},\psi,I\big) - v(t,q,\psi,I)\Big],
\end{align*}
where
\begin{align*}
  \psi_{j\delta}^{-n} = (\psi^1,\dots,\psi^{n-1},j\delta,\psi^{n+1},\dots), \quad I_{I_k}^{-n} = (I^1,\dots,I^{n-1},I_k,I^{n+1},\dots).
\end{align*}

In order to apply the finite difference method we introduce the discretization of time and state space. For inventories we have $\mathfrak{Q} = \{q_1=0 < \ldots < q_{\#\mathfrak{Q}} = q_0\}$. Time discretization in the slice is $\mathfrak{T} = \{t_0=0 < t_1 = t_0 + \Delta_t < \ldots < t_{\#\mathfrak{T}} = \Delta_v\}$. We also discretize the order volumes the trader can send $\mathfrak{L} = \{l_1=0 < \ldots < l_{\#\mathfrak{L}} = q_0\}$. \\

Using the first difference for the value function derivative with respect to time we can rewrite the above equation as $\forall i \in \{0,\ldots, \#\mathfrak{T}-1\}, \forall q \in \mathfrak{Q}, \forall (\psi, I) \in \mathcal{M}$
\begin{align*}
v(t_{i+1}, q, \psi, I) = & v(t_i, q, \psi, I) -\Delta_t\Bigg( g(q - q_t^{\star}) - \mu q \\
  & - \sum_{n=1}^N \sum_{j=1}^J r^{n,\psi}_{\psi,j\delta}\big(v(t,q,\psi_{j\delta}^{-n},I)-v(t,q,\psi,I)\big) - \sum_{n=1}^N \sum_{k=1}^K r^{n,I}_{I,I_k}\big(v(t,q,\psi,I_{I_k}^{-n})-v(t,q,\psi,I)\big) \\
  & -\!\!\!\!\!\!\!\!\!\! \sup_{p\in \{-1, 0, 1\}^N, \ell\in \mathfrak{L}^N} \sum_{n=1}^N \lambda^{n}(\psi,I,p^{n},\ell)\mathbb{E}\Big[\epsilon^{n}\ell^{n}(\frac{\psi^{n}}{2}+p^{n}\delta^n)+v\big(t,q-\ell^{n}\epsilon^{n},\psi,I\big) - v(t,q,\psi,I)\Big]\Bigg),
\end{align*}
with terminal condition $v(T,q,\psi,I)=0$.\\ 

In terms of calculations the most demanding part is obviously the search of the supremum which is needed to be performed on the dimension $3^N \times \#\mathfrak{Q} \times \#\mathfrak{L}^N \times \#\mathcal{M}$ for each time step. From what follows that finite differences can be applied to solve the problem of optimal orders posting for the stock cross-listed in $N=2$ venues with reasonable precision and calculation time. However, if we introduce more venues finite differences are not going to be any more efficient because the complexity is growing exponentially.\\

For our numerical example, we used the discretization with $\#\mathfrak{Q} = 101$ and $\#\mathfrak{L} = 51$ which assures the calculation time (on a simple PC) around 1min for the whole slice. 
\subsubsection{Neural networks}
In this section, we briefly introduce the method using neural networks to solve HJB equations. In this paper, we used a method which can be referred to as Actor-Critic method to approximate optimal controls and corresponding value function for the problem. Applications of this approach have shown to be fruitful, especially when we talk about equations in high dimension, more elaborate description of the method can be found for example in \cite{bachouch2018deep, baldacci2019market, gueant2019deep, hure2018deep}. \\

The core of this approach is to represent the strategy of the trader with a neural network as well as the corresponding value function. Then one needs to formalize the target functions for both neural networks and to perform the gradient descent on the parameters (weights) of these networks. This procedure needs to be done for every time step, and so one ends up with $2\#\mathfrak{T}$ networks.\\

Let us start from the description of the value function approximation. We consider the neural networks taking as an input the spreads and the imbalances in the venues of interest and the inventory of the trader giving as an output the value function at this point. As in the finite difference method we solve our problem backward, starting from $t_{\#\mathfrak{T} - 1}=\Delta_v-\Delta_t$, because the value function at the end of the slice is known from the terminal condition. To calculate the value function at time $t_i, \forall i \in \{0,\ldots, \#\mathfrak{T}-1\}$ we use the minimization of the mean-squared error between values given by the neural network and the target values calculated with the use of the value function approximation for time $t_{i+1}$ and the network for the controls at the current step. Let us assume that we have the controls $\ell^*, p^*$ (obtained via neural networks, for example) for time $t_i$, then the target for the value function can be found as
\begin{align*}
v^{\text{target}}(t_{i-1}, q, \psi, &I) = v[\theta^v_i](t_i, q, \psi, I) +\Delta_t\bigg( g(q - q_t^{\star}) - \mu q \\
& - \sum_{n=1}^N \sum_{j=1}^J r^{n,\psi}_{\psi,j\delta}\big(v[\theta^v_i](t,q,\psi_{j\delta}^{-n},I)-v[\theta^v_i](t,q,\psi,I)\big) \\
&- \sum_{n=1}^N \sum_{k=1}^K r^{n,I}_{I,I_k}\big(v[\theta^v_i](t,q,\psi,I_{I_k}^{-n}) - v[\theta^v_i](t,q,\psi,I)\big) \\
 & - \sum_{n=1}^N \lambda^{n}(\psi,I,p^{*n},\ell^*)\mathbb{E}\Big[\epsilon^{n}\ell^{*n}(\frac{\psi^{n}}{2}+p^{*n}\delta^n)+v[\theta^v_i]\big(t,q-\ell^{*n}\epsilon^{n},\psi,I\big) - v[\theta^v_i](t,q,\psi,I)\Big]\bigg), 
\end{align*}
with $v[\theta^v_{\#\mathfrak{T}}](t_{\#\mathfrak{T}}, q, \psi, I) = 0$ and where $[\theta^v_i]$ stands for the weights of the neural network for the value function at time $t_i$.\\

The trader's inventory is of continuous nature, however, spread and imbalance are categorical, so we need to verify if we should use some special techniques to ensure better fitting in this case.
\vspace{-4mm}
\begin{figure}[H]
  \begin{center}
      \includegraphics[width=0.45\textwidth]{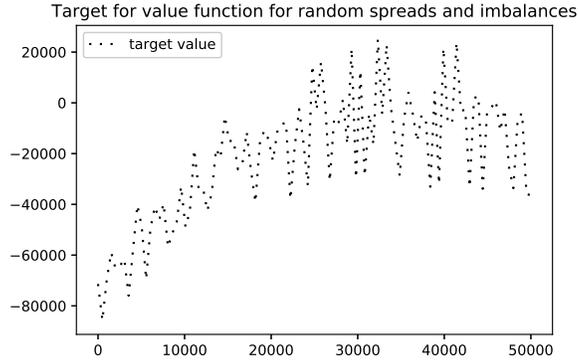}
      \vspace{-5mm}
      \caption{Target value function for increasing inventory and random market states.}
      \label{target_vf}
    \end{center}
\end{figure}
\vspace{-6mm}
Let us see first in Figure \ref{target_vf} the example of the target value function of the trader for $q\in[0, q_0]$ at different spreads and imbalances. We see considerable changes in the value function level depending on the market state which we would like to capture by our approximation.\\

Now, let us compare the fitting of the value function parametrization taking as inputs raw spread and imbalance values with the parametrization working with encoded values of the spread and the imbalance. Here we are going to use the so-called one-hot encoding for categorical variables, which consists in the representation of different values of the variable by a one-hot vector $e_\psi^i \in \{0, 1\}^{\#\Psi}$ for the spread and $e_I^i \in \{0, 1\}^{\#\mathcal{I}}$ for the imbalance. And $e^i$ (both for $e_\psi^i$ and $e_I^i$) are such that that $e^i_j = 0, \forall j\neq i,$ and $e^i_i=1$ otherwise. 
\vspace{-3mm}
\begin{figure}[H]
\begin{minipage}[c]{.49\linewidth}
  \begin{center}
      \includegraphics[width=0.9\textwidth]{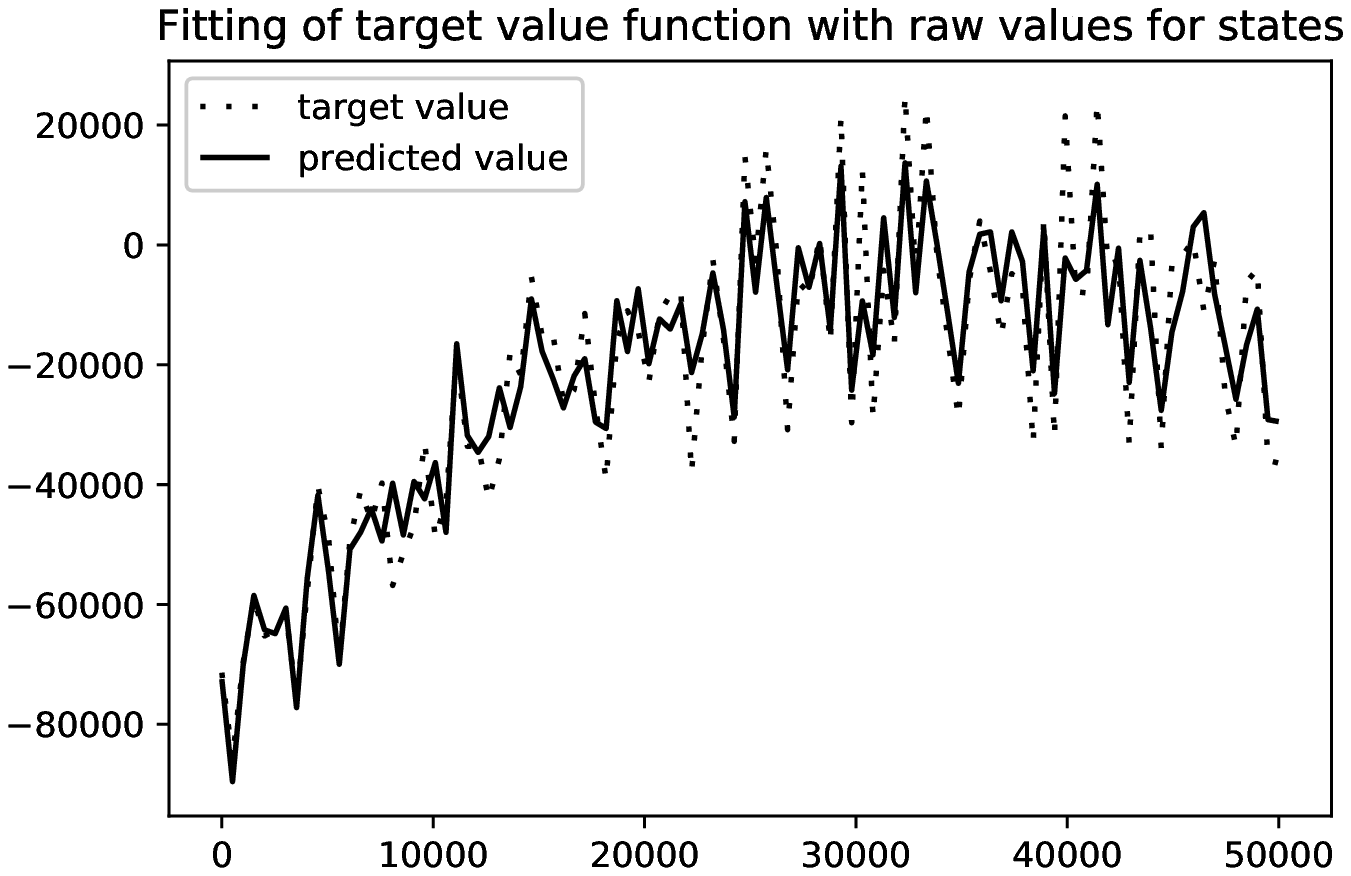}
      \vspace{-3mm}
      \caption{Comparison of the target value with approximation continuous in spread and imbalance.}
      \label{fit_cont}
    \end{center}
\end{minipage} \hfill
\begin{minipage}[c]{.48\linewidth}
   \begin{center}
       \includegraphics[width=0.9\textwidth]{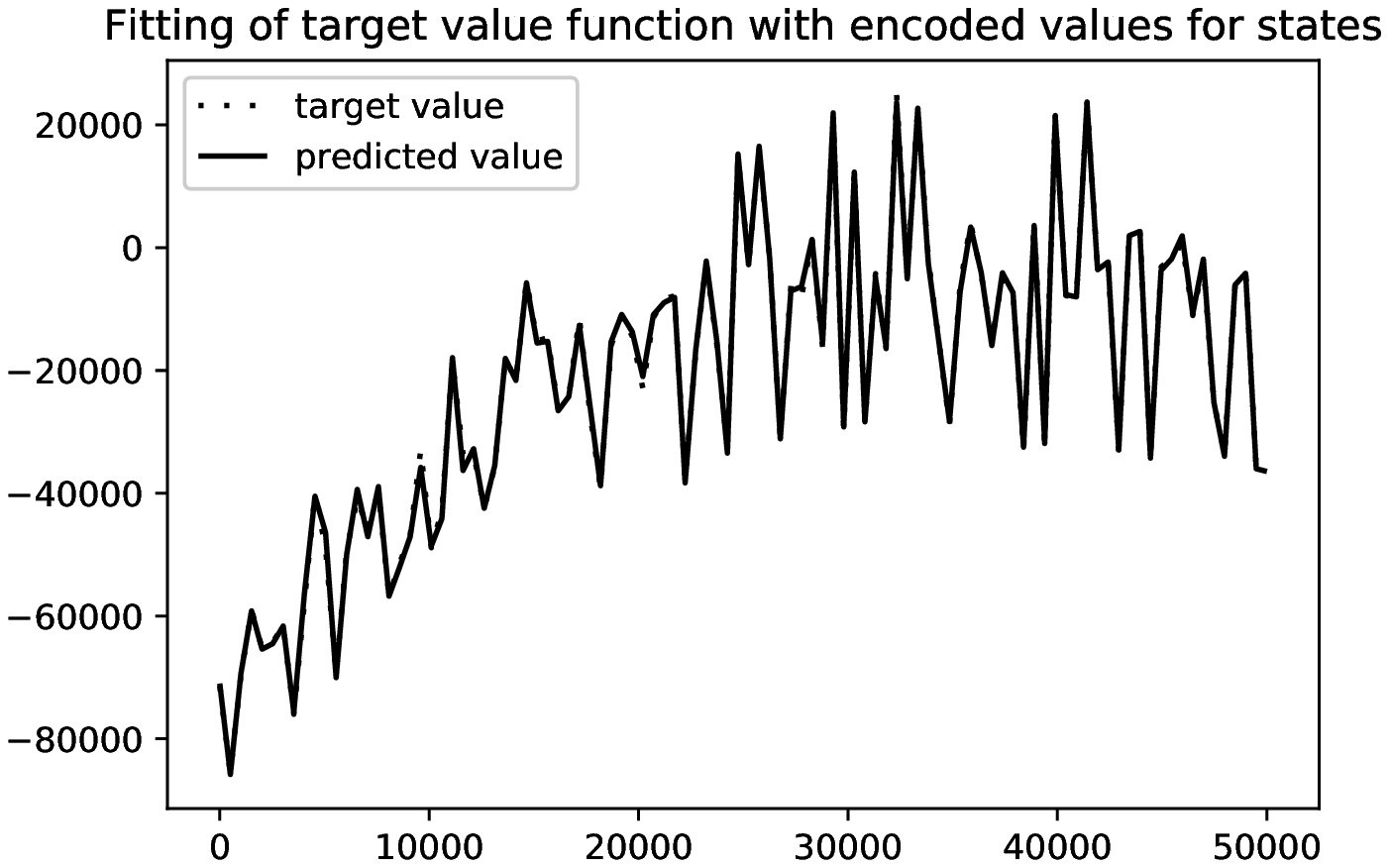}
       \vspace{-3mm}
       \caption{Comparison of the target value with approximation discrete in spread and imbalance.}
       \label{fit_cat}
     \end{center}
  \end{minipage} 
\end{figure}
\vspace{-3mm}
In Figures \ref{fit_cont} and \ref{fit_cat} we see the comparison between values predicted by two parametrizations with target values for the same number of learning epochs. There is a considerable gain in precision when the parametrization takes into account the categorical nature of market states. Therefore we apply it for both value function network approximation and the strategy neural network approximation.\\

Now, let us describe the learning procedure for the strategy. First of all, the inputs of the strategy network are the same as for the value function network, \textit{i.e.} the trader's inventory, spreads and imbalances for both venues. As an output, we need to have volumes of the orders and limits on which the trader needs to send his orders. Volumes to send to each venue are bounded by the current inventory because we do not want the trader to execute more shares than he possesses. Limits should equal $-1$, $0$, or $1$, but as soon as we want to use the tools of automatic differentiation, we need to represent them by differentiable function. The softmax activation function serves well to this purpose, so we represent the limits for each venue by the probabilities to send an order to each precise limit. In practice, the trader can choose the maximum of the three to perform his action.\\

The optimization criterium used for the strategy neural network is the function under supremum from the HJBQVI \eqref{vfansatz}, with limit probabilities taken into account (let us denote them by $\mathbb{P}(p = a)$, for $a \in \{-1, 0, 1\}$) we need to maximize with respect to $\theta^\ell_i, i \in \{0, \ldots, \#\mathfrak{T} - 1\}$:
\begin{align*}
\sum_{n=1}^N \sum_{a \in \{-1, 0, 1\}}\mathbb{P}[\theta^\ell_i](p^{n} = a)\lambda^{n}(\psi,I,a,\ell[\theta^\ell_i])\mathbb{E}\Big[&\epsilon^{n}\ell[\theta^\ell_i]^{n}\Big(\frac{\psi^{n}}{2}+a\delta^n\Big)\\
&+v[\theta^\ell_{i+1}]\big(t_i, q-\ell[\theta^\ell_i]^{n}\epsilon^{n}, \psi, I\big) - v[\theta^\ell_{i+1}](t_i, q, \psi, I)\Big], 
\end{align*}
where $\theta^\ell_i$ stand for the weights of the neural network of controls at time $t_i$. So we want to maximize this function for all possible values of market states and inventories. To avoid the dimensionality trap we need to optimize this function on some subset of possible values, which we are going to draw randomly.\\

When optimizing neural networks approximations, it is important to normalize the data, to have if possible a universal set of hyperparameters. First of all, the inventory entering as an input of the value function neural network and of the strategy neural network is normalized by $q_0$ to always stay in $[0, 1]$. Also, we are going to learn not the target value function itself, but the target value function normalized by $q_0$, which sufficiently reduces the order of values. For strategy network, we are going to learn the proportion of the inventory to be sent and not the volume itself. And finally, we can notice that for high inventories the difference between value functions (which are quadratic in the inventory) in the supremum can become much more important than the profit of the trader coming from the tick (which is not more than linear in inventory). This fact can hinder us from finding optimal values for the limit to which the trader should send his order, especially for small inventories. We normalize the values of the optimized function for different inventories to make small inventories more important by multiplying all values by $\frac{1}{q}$. However, this latter normalization is used when we optimize over the part of the strategy responsible for the limits only, leaving volume updates untouched.\\

To summarize in Figures \ref{model_l} and \ref{model_v}, we presented the structures of the neural networks used to represent the approximators for the strategy and the value function. Another feature worth mentioning here is the separation of market state and inventory inputs for some layers, both for the strategy and the value function. This allows capturing features of the market state independently of the inventory. Also, we separated some layers preceding the outputs of the strategy network to be able to perform the learning process with different learning rates for volumes and limits of limit orders. \\
\vspace{-3mm}
\begin{figure}[H]
\begin{minipage}[c]{.46\linewidth}
  \begin{center}
      \includegraphics[width=\textwidth]{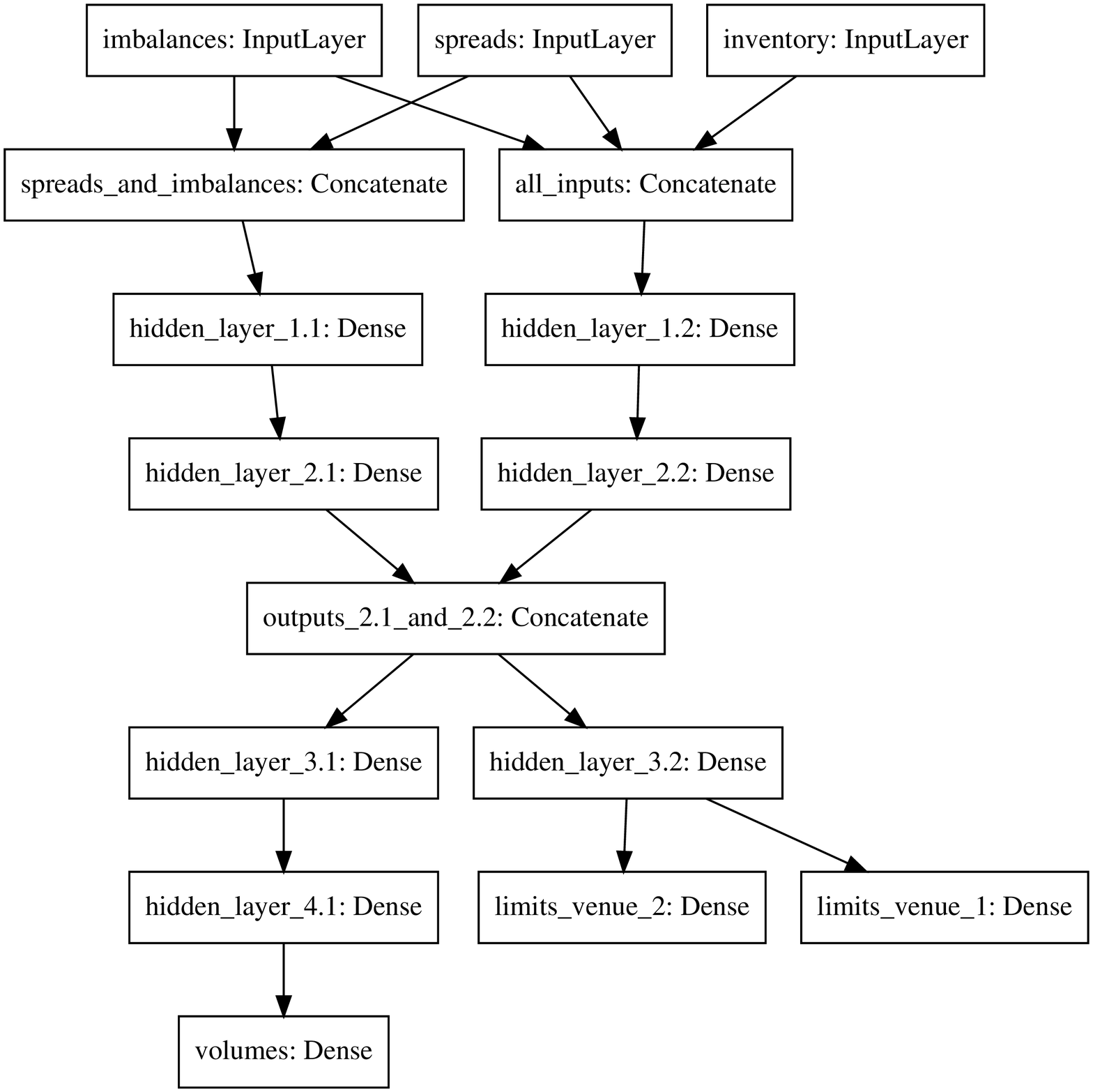}
      \vspace{-3mm}
      \caption{Neural network structure for the trader's strategy.}
      \label{model_l}
    \end{center}
\end{minipage} \hfill
\begin{minipage}[c]{.46\linewidth}
   \begin{center}
       \includegraphics[width=\textwidth]{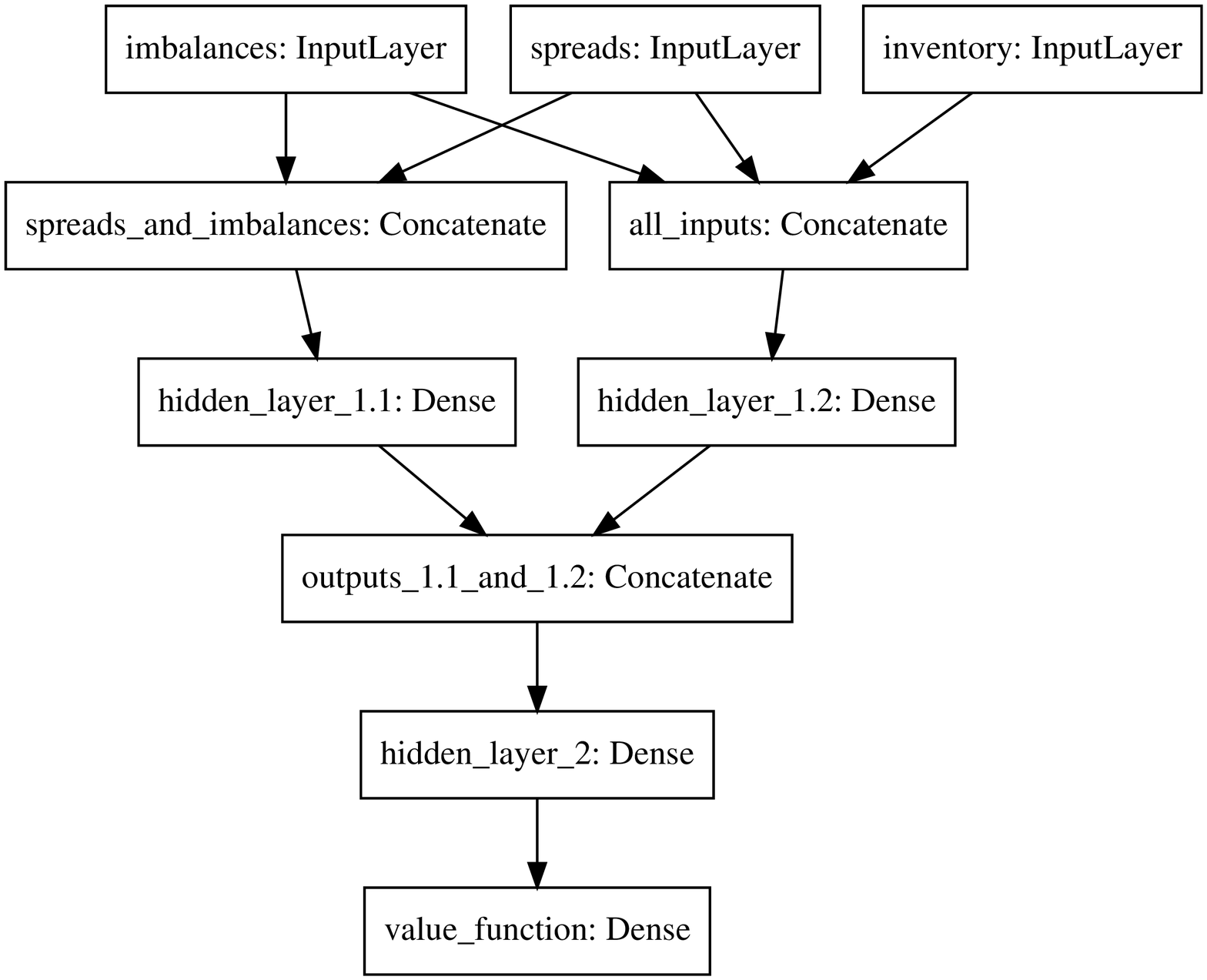}
       \vspace{-3mm}
       \caption{Neural network structure for the value function.}
       \label{model_v}
     \end{center}
  \end{minipage} 
\end{figure}
\vspace{-3mm}
While the finite difference schemes must complete the entire recalculation of values for the whole grid every time the trader wants to adapt his strategy using the updated market parameters, neural networks can be adapted progressively, starting from some pre-trained strategy, for example, the one corresponding to the previous parameters. In practice, a pre-trained model can be reused for different problem settings due to normalization. Therefore a long and elaborate training procedure should be done only once. The resulting model can be ameliorated by small adjustment trainings which take only 1 minute on the simplest instance of the AWS platform (2CPU, no GPU), and have great speed-up potential when performed on more complex infrastructures. 

\subsection{Two identical venues}

We assume that the trader is confident about his estimation of $\sigma$. Therefore he uses Bayesian update only on the drift $\mu$ of the asset. The venues share identical parameters, which will be inferred by the trader through time.

\subsubsection{Value function}

We first plot in Figures \ref{value_funcions_04} and \ref{value_funcions_59} the evolution through time of the value function of the trader in the state $\psi^1=\psi^2=1$ and $I^1=I^2=0$ during a slice of execution, obtained through finite difference method.
\vspace{-3mm}
\begin{figure}[H]
\begin{minipage}[c]{.46\linewidth}
  \begin{center}
      \includegraphics[width=0.9\textwidth]{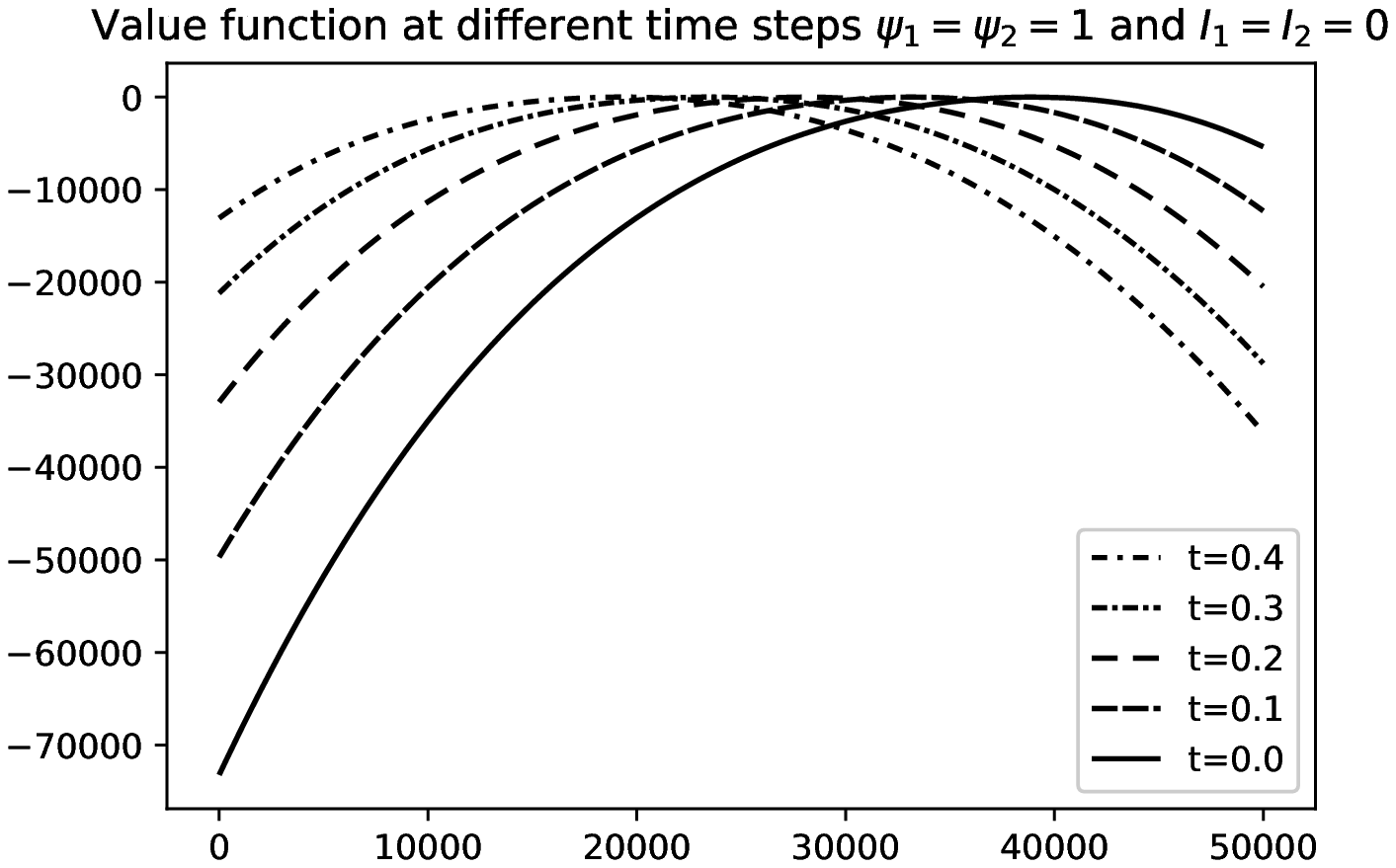}
      \vspace{-3mm}
      \caption{Value function with respect to the inventory between $t=0$ and $t=0.4$.}\label{value_funcions_04}
    \end{center}
\end{minipage} \hfill
\begin{minipage}[c]{.46\linewidth}
   \begin{center}
       \includegraphics[width=0.9\textwidth]{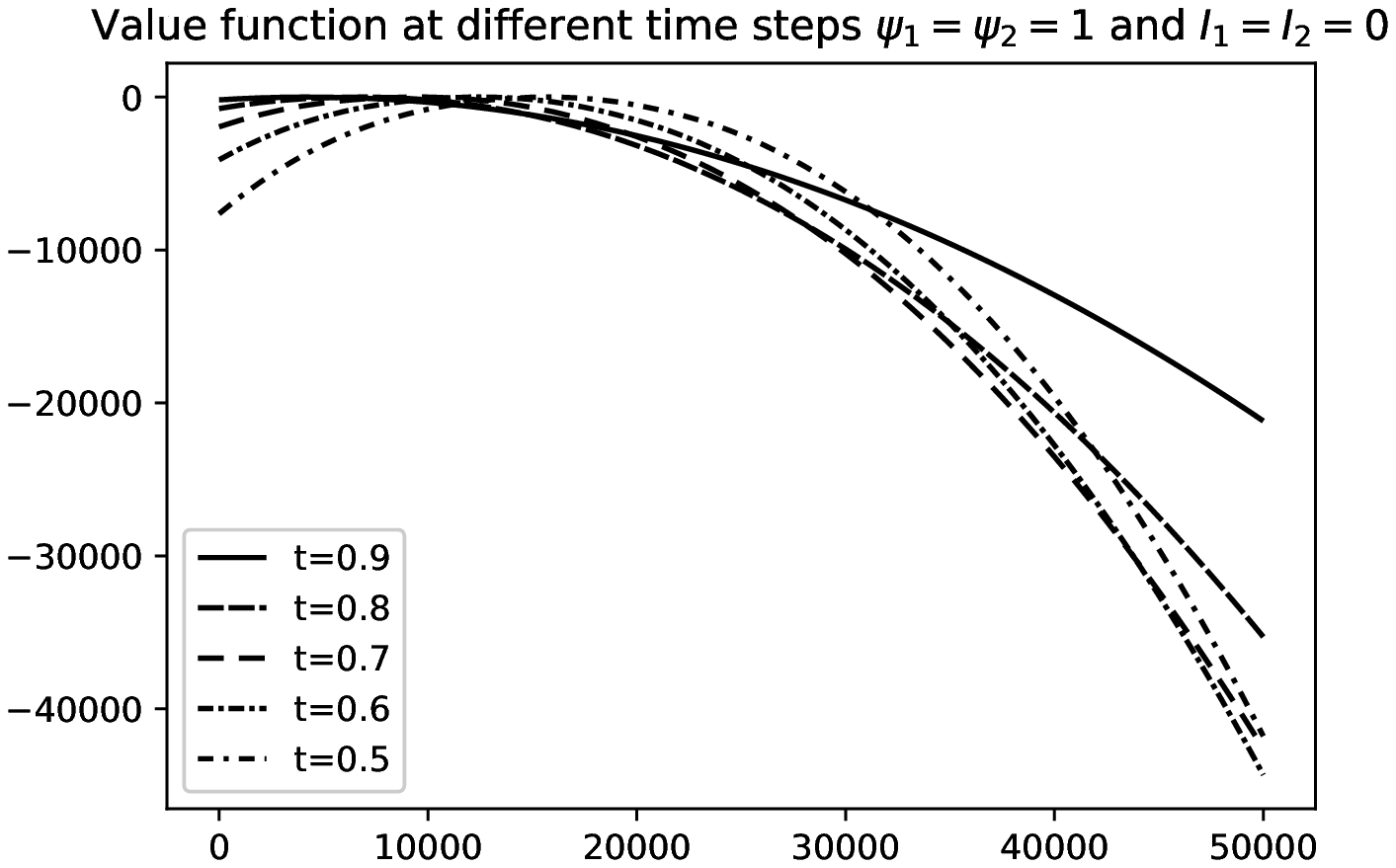}
       \vspace{-3mm}
       \caption{Evolution of the value function $v$ between $t=0.5$ and $t=0.9$.}\label{value_funcions_59}
     \end{center}
  \end{minipage} 
\end{figure}
\vspace{-3mm}
The parabolic form of the value function comes from the term $g(q-q_t^\star)$ in \eqref{HJBQVI}. The maximum value indicates the optimal inventory for the next step in the slice. When $t$ increases, the maximum shifts toward zero, which means that the trader wants to finish the execution at the end of the slice. \\

We plot in Figures \ref{value_funcions_04_nets} and \ref{value_funcions_59_nets} the value function of \eqref{HJBQVI} obtained using neural networks. We can see that the neural networks approximate accurately the value function. 
\vspace{-3mm}
\begin{figure}[H]
\begin{minipage}[c]{.46\linewidth}
  \begin{center}
      \includegraphics[width=0.9\textwidth]{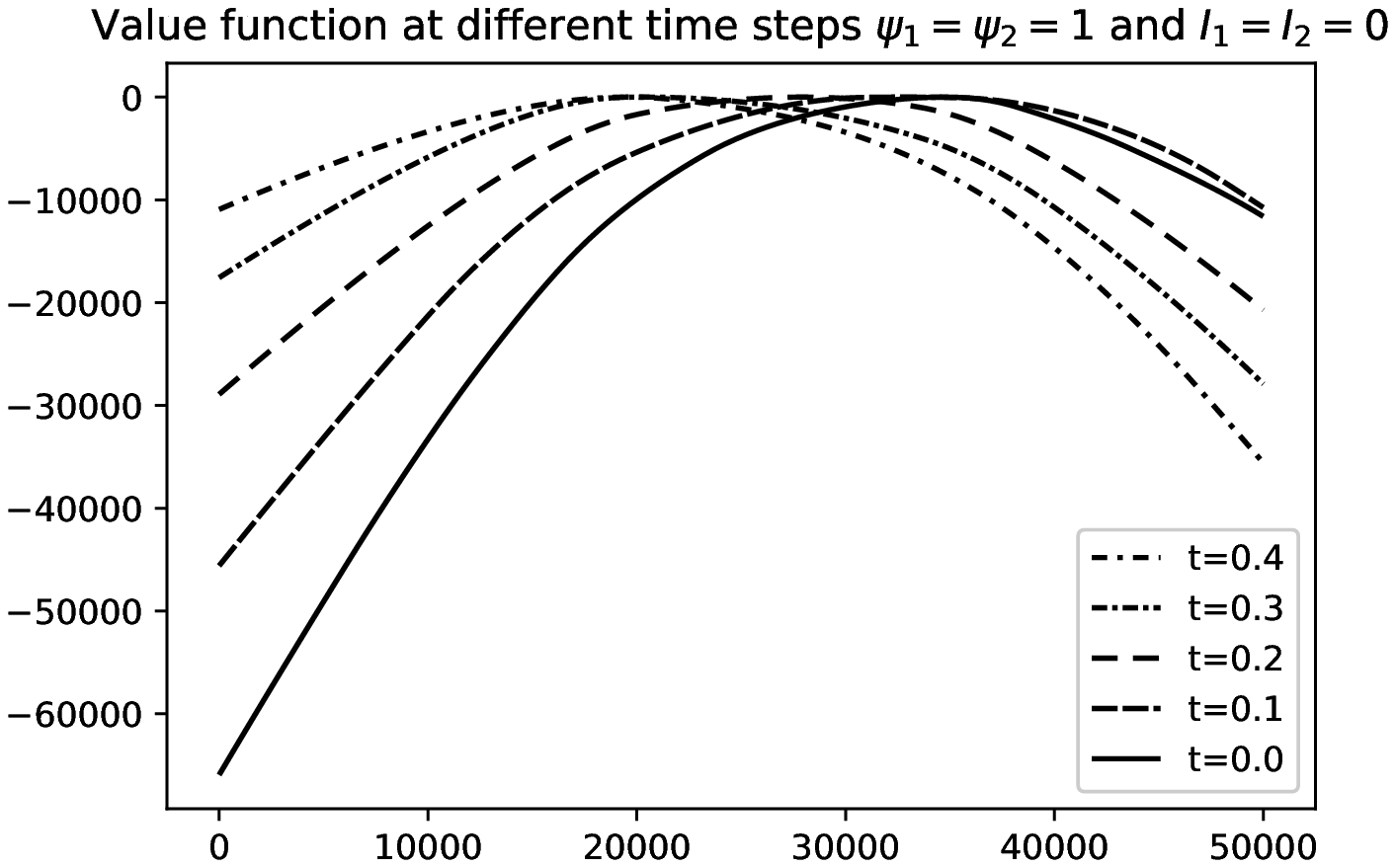}
      \vspace{-3mm}
      \caption{Evolution of the value function $v$ between $t=0$ and $t=0.4$ using neural networks.}\label{value_funcions_04_nets}
    \end{center}
\end{minipage} \hfill
\begin{minipage}[c]{.47\linewidth}
   \begin{center}
       \includegraphics[width=0.88\textwidth]{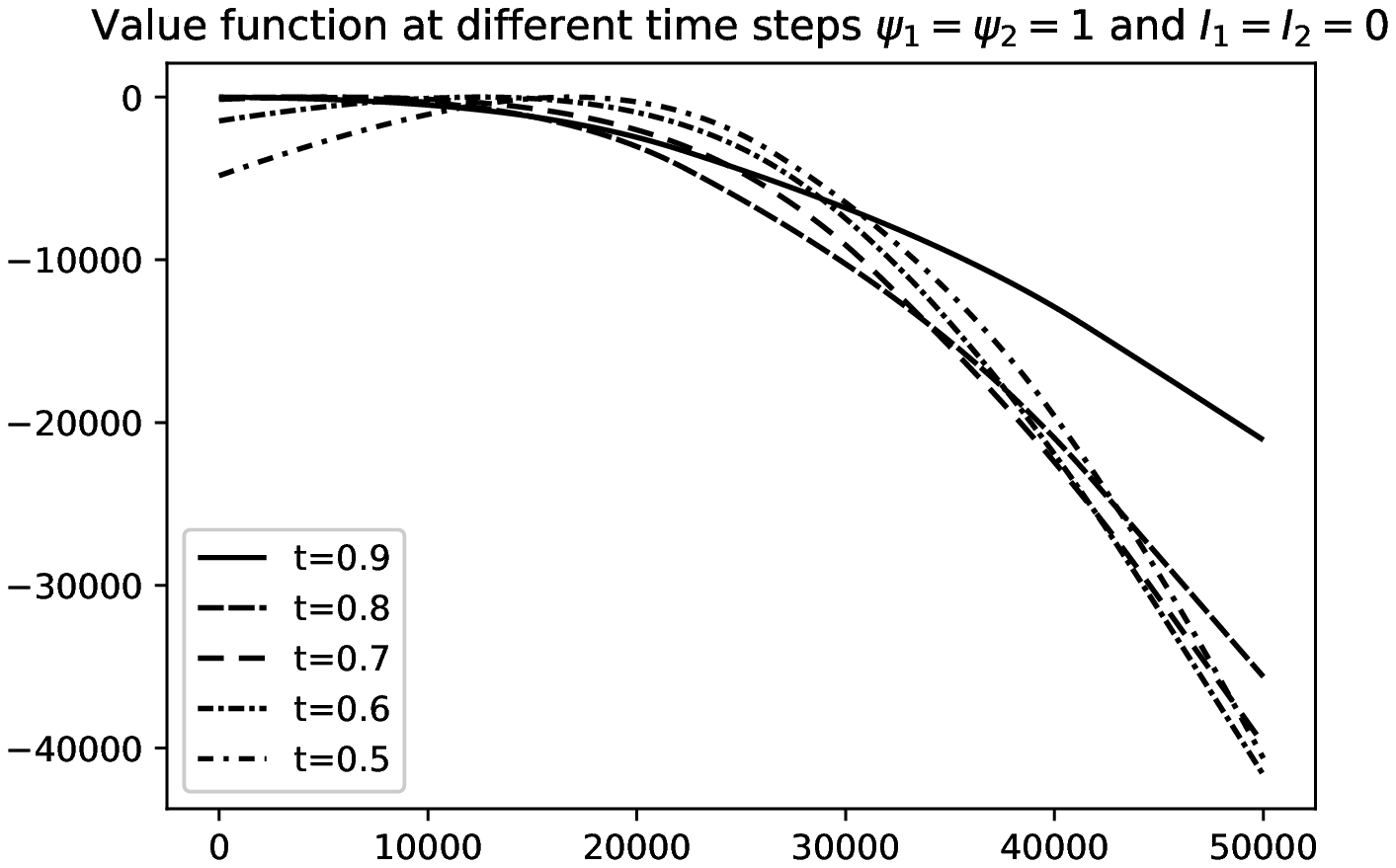}
       \vspace{-3mm}
       \caption{Evolution of the value function $v$ between $t=0.5$ and $t=0.9$ using neural networks. }\label{value_funcions_59_nets}
     \end{center}
  \end{minipage} 
\end{figure}
\vspace{-3mm}

Next, we plot the strategy (in terms of limits and volumes) of the trader in both venues, using finite difference schemes.

\subsubsection{Strategy: limit orders and volumes with finite difference schemes}

In Figures \ref{limits_venue1} and \ref{limits_venue2}, we plot the limits at which the trader posts his limit orders in the two venues, given equal spread and imbalance processes.
\begin{figure}[H]

\begin{minipage}[c]{.46\linewidth}
  \begin{center}
      \includegraphics[width=0.9\textwidth]{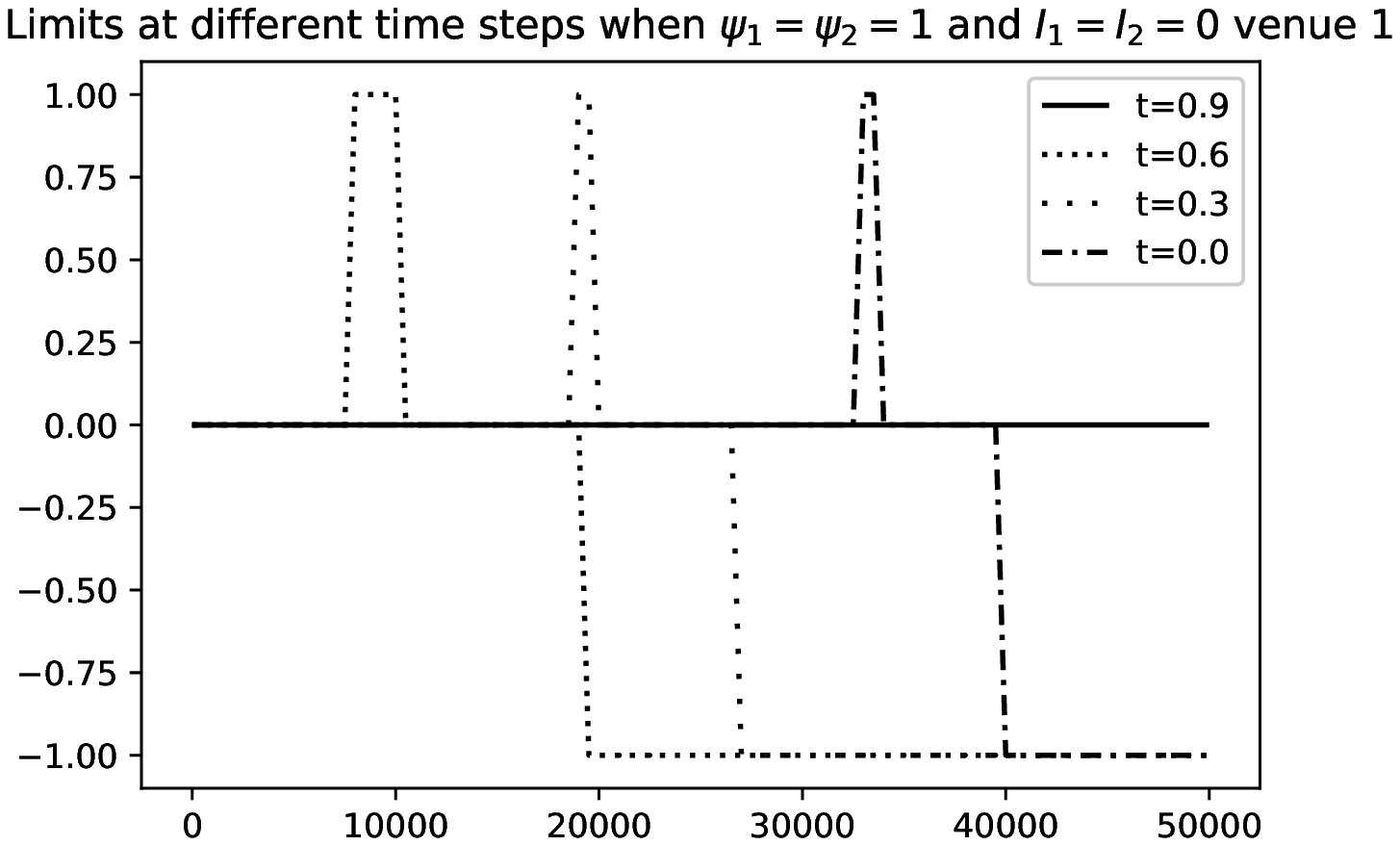}
      \vspace{-3mm}
      \caption{Limit strategy in the first venue, $\psi^1 = \psi^2 = \delta, I^1 = I^2 = 0$.}\label{limits_venue1}
    \end{center}
\end{minipage} \hfill
\begin{minipage}[c]{.46\linewidth}
   \begin{center}
       \includegraphics[width=0.9\textwidth]{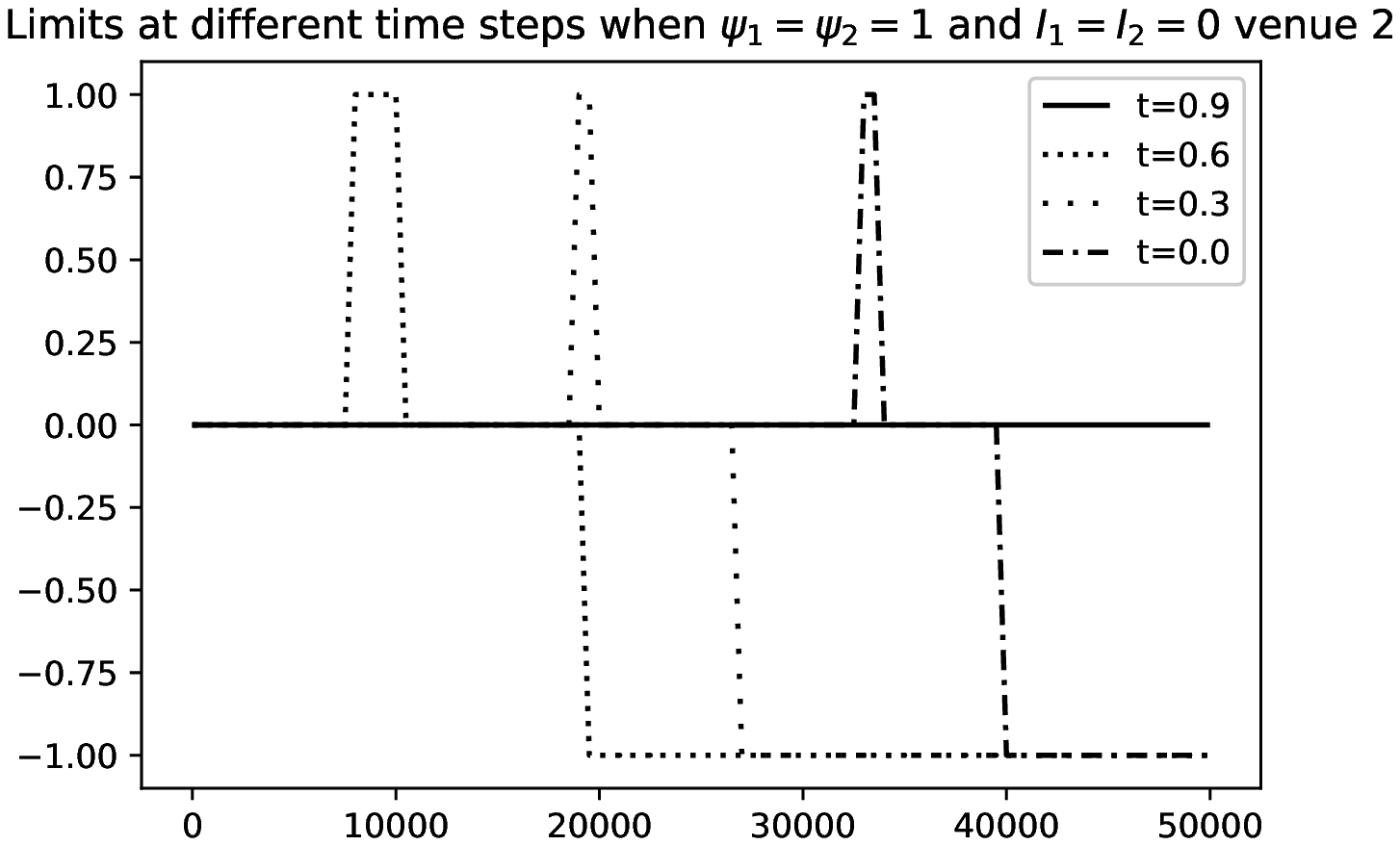}
       \vspace{-3mm}
       \caption{Limit strategy in the second venue, $\psi^1 = \psi^2 = \delta, I^1 = I^2 = 0$.}\label{limits_venue2}
     \end{center}
  \end{minipage} 
\end{figure}

As the trader has the same prior distribution in the two venues, his strategy is the same in both venues. At the beginning of the slice, \textit{i.e.} at $t=0$, the maximum of the value function is near $q=32000$. Therefore, if the trader has a lower inventory, he does not post any orders and wait for the next time step. If he has a higher inventory, he tries to reach $q=32000$ inventory. For $q\in[32000,34000]$, being sufficiently close to the next step optimal inventory, he posts limit orders on the second best limit to collect an additional tick. For $q\in[34000,40000]$, he posts at the first best limit to increase his probability of execution. If he has $q>40000$, he creates a new best limit and accepts to loose one tick in order to be executed faster and reach the optimal inventory at the following time step. We can see in this behavior the trade-off between the possibility of being executed at a more favorable price and the necessity to complete the execution. \\

For the sake of homogeneity (for all $\mathcal{M}$ market states, the trader faces similar trade-off), we considered the same set of controls for the limit where the trader can send his order. For this reason, we can see that even for the spread equal to $\delta$ the trader can submit an order to the limit $p=-1$, which in practice can obviously be treated as $p=0$ due to piecewise monotonous nature of the optimal limit strategy (which is, in fact, monotonous, though it cannot be reflected by finite differences when the optimal volume equals to $0$). \\

When the trader is near the end of the slice, he starts posting limit orders earlier (can be seen if both volumes and limits are considered). For example if $t=0.6$, he begins to trade at the second best limit when $q\in [8000,11000]$, at the first best limit when $q\in [11000,19000]$, and creates a new best limit when $q\in [19000,50000]$. Therefore, if the trader still has a very positive inventory at the end of the slice, he prefers to sacrifice one tick at the first best limit in order to complete his execution at this step. \\

It is important to highlight the fact that, when $t=0.9$, the trader does not rush to liquidate his inventory completely. This comes from the absence of a terminal penalty, often used in optimal liquidation problem to guarantee the complete execution of the inventory. It enables in some sense to ``relax'' the optimal execution framework on a slice, as the part of the inventory that has not been executed during one slice is split between the remaining ones. \\

We plot in Figures \ref{volumes_venue1} and \ref{volumes_venue2} the volumes posted in both venues, for the same spread and imbalance. We see that, at the beginning of the slice, the trader begins to post a nonzero volume only when $q>32000$. Moreover, he posts a higher volume when he is near the end of the slice. 
\vspace{-3mm}
\begin{figure}[H]
\begin{minipage}[c]{.48\linewidth}
  \begin{center}
      \includegraphics[width=0.9\textwidth]{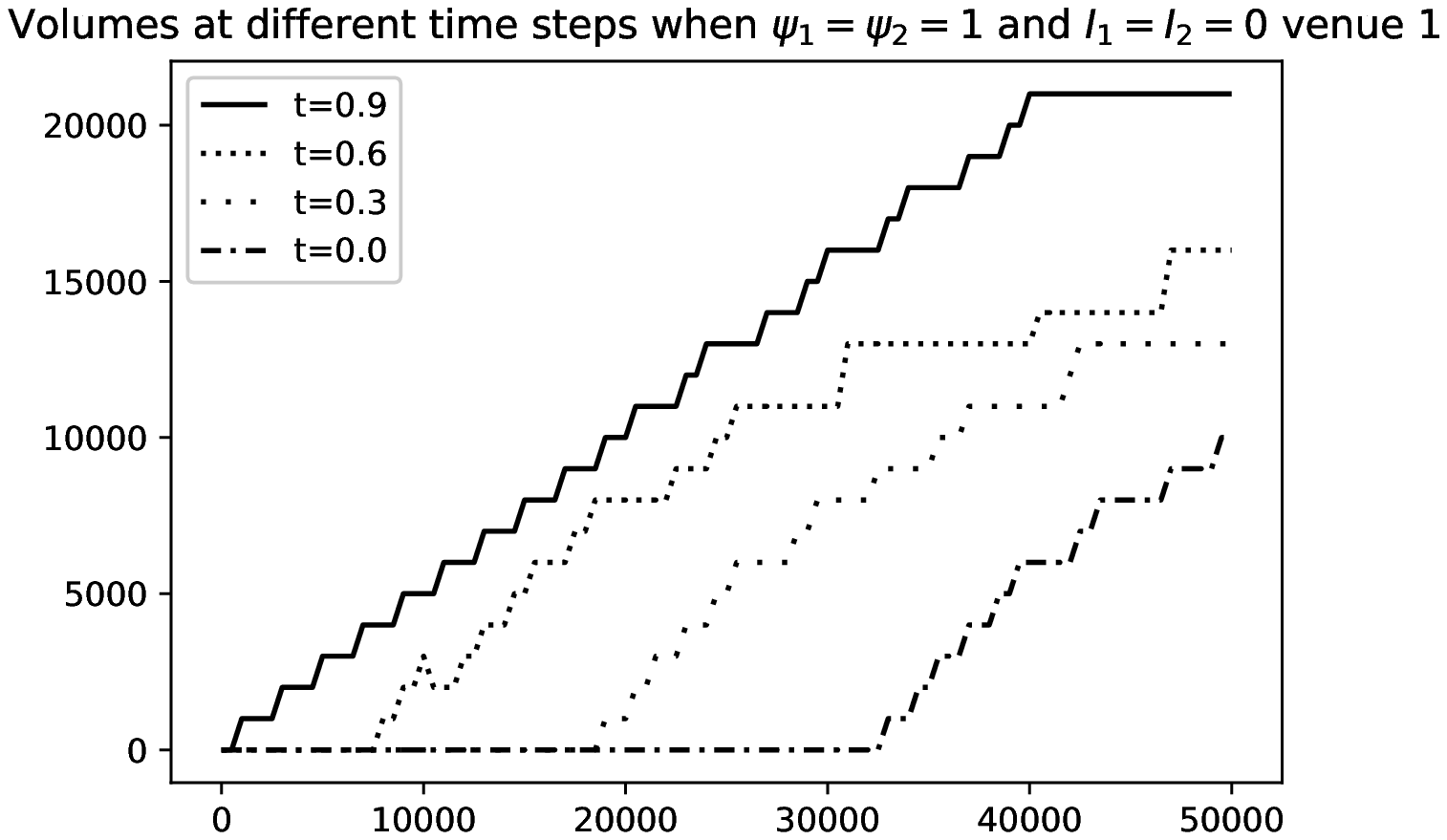}
      \vspace{-3mm}
      \caption{Volume sent to the first venue, \protect\\ $\psi^1 = \psi^2 = \delta, I^1 = I^2 = 0$.}\label{volumes_venue1}
    \end{center}
\end{minipage} \hfill
\begin{minipage}[c]{.48\linewidth}
   \begin{center}
       \includegraphics[width=0.9\textwidth]{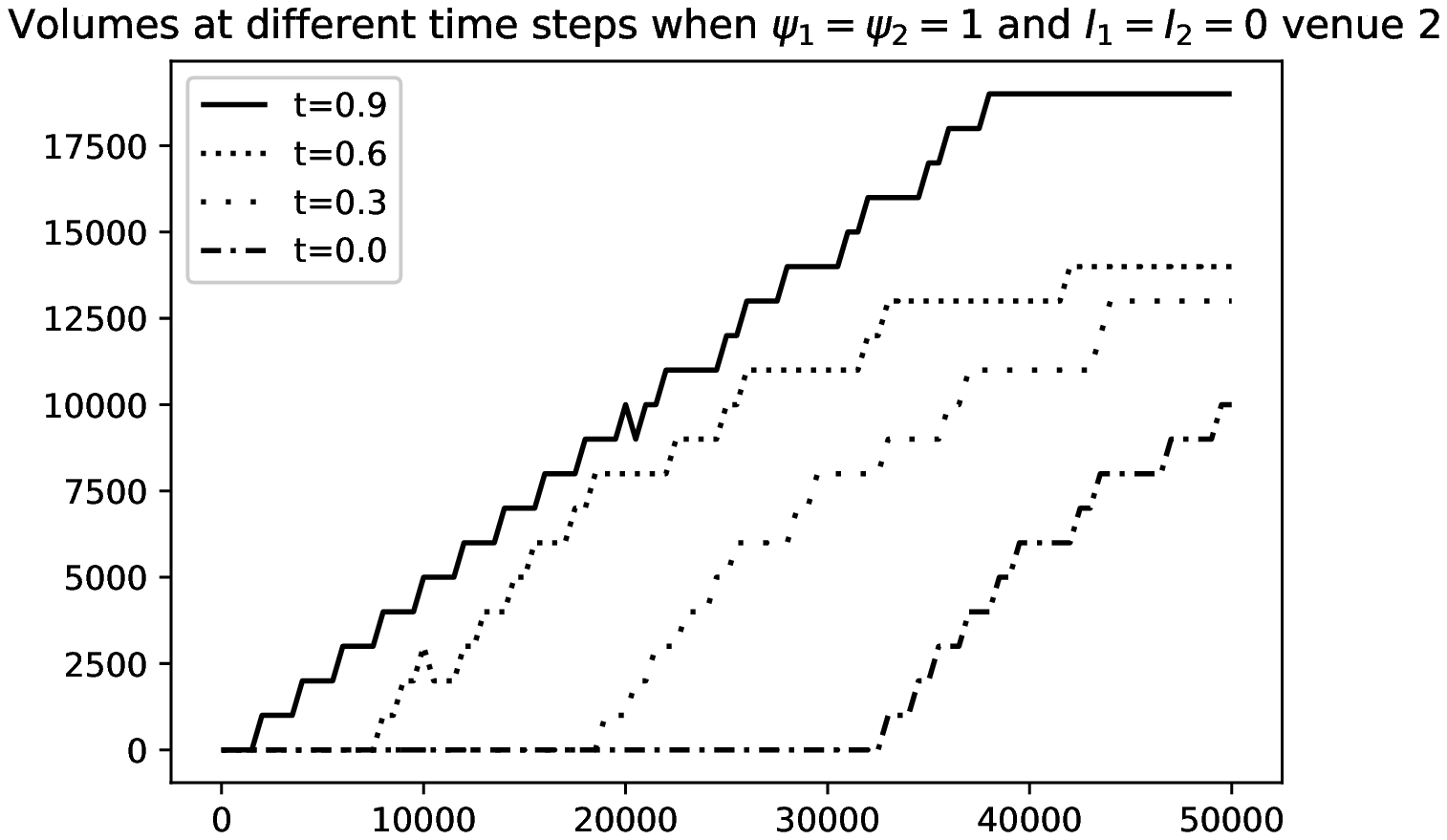}
       \vspace{-3mm}
       \caption{Volume sent to the second venue, $\psi^1 = \psi^2 = \delta, I^1 = I^2 = 0$.}\label{volumes_venue2}
     \end{center}
  \end{minipage} 
\end{figure}
\vspace{-3mm}
s
When the second venue has a higher spread, we plot the strategy of the trader in both venues in Figures \ref{limits_venue_spr} and \ref{volumes_venue_spr}. 

\begin{figure}[H]
\begin{minipage}[c]{.47\linewidth}
  \begin{center}
      \includegraphics[width=0.88\textwidth]{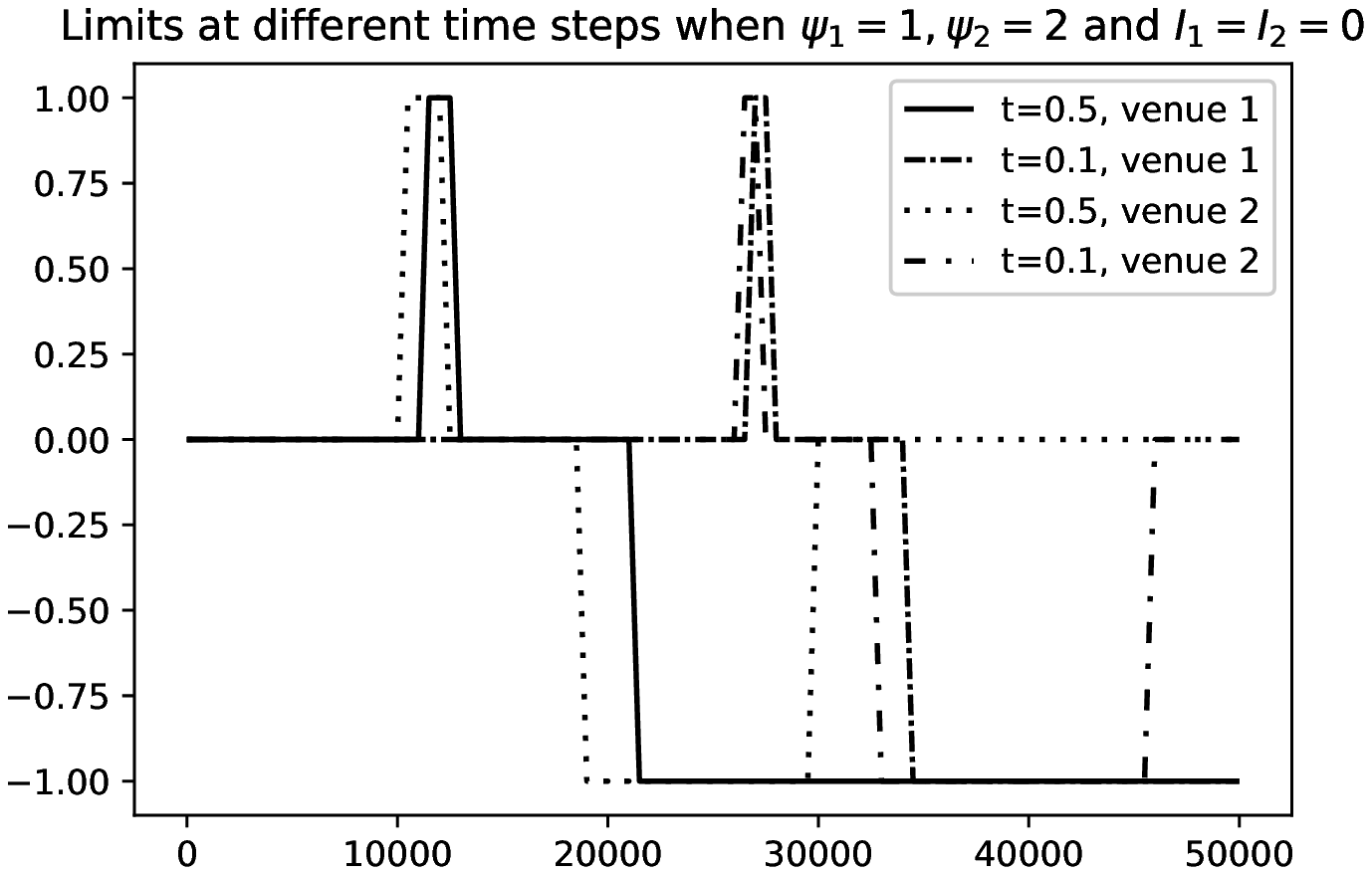}
      \vspace{-3mm}
      \caption{Limit strategy, $\psi^1=\delta, \psi^2 =~2\delta,$ \protect\\ $I^1 = I^2 = 0$.}\label{limits_venue_spr}
    \end{center}
\end{minipage} \hfill
\begin{minipage}[c]{.46\linewidth}
   \begin{center}
       \includegraphics[width=0.9\textwidth]{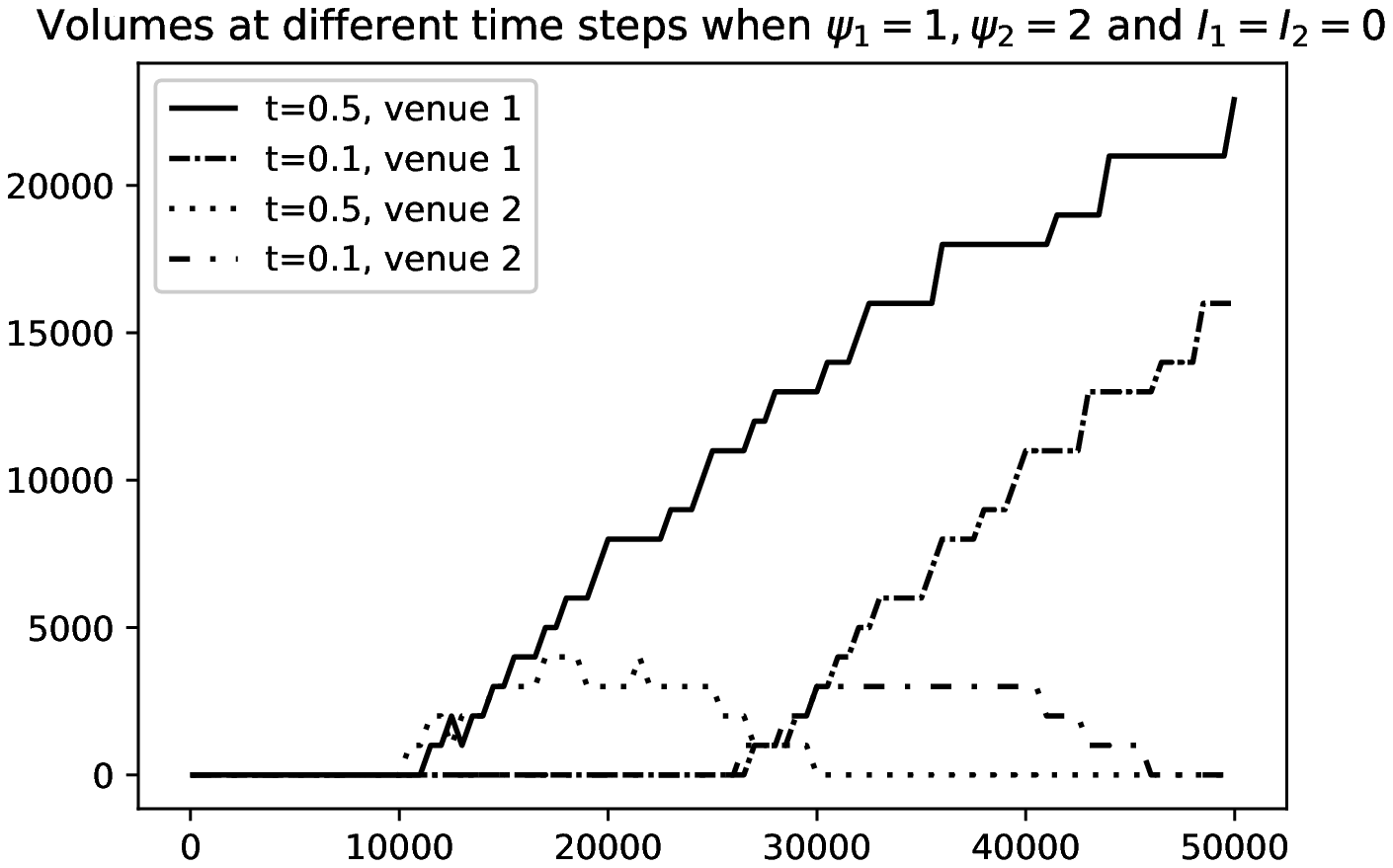}
       \vspace{-3mm}
       \caption{Volume strategy, $\psi^1 = \delta, \psi^2 = 2\delta,$ \protect\\ $I^1 = I^2 = 0$.}\label{volumes_venue_spr}
     \end{center}
  \end{minipage} 
\end{figure}

For $t=0.5$, we see in Figure \ref{limits_venue_spr} that the trader starts to post at the second best limit in the second venue when $q=10000$ and in the first when $q=11000$. For $q\in[18000,21000]$, he creates a new best limit in the second venue to execute his inventory faster but keeps posting at the best limit in the first venue in order to collect a higher spread. Finally, for $q\in[30000,50000]$, he stops posting in the second venue in order to consume more liquidity in the first one where the probability of getting his order filled is higher. Similar interpretations apply for $t=0.1$. \\

In Figure \ref{volumes_venue_spr}, we see that the trader posts a higher volume in the first venue compared to the second one. For $t=0.5$, he starts to trade at $q=10000$ for the second venue and at $q=11000$ for the first one. The volume posted in the first venue increases almost linearly with respect to the inventory. In contrast, the volume posted in the second venue increases until an inventory of $q=22000$, then stays constant until $q=30000$ and decreases to zero afterward. This means that for $q\in[10000,30000]$, the trader prefers to collect the spread from both venues. When $q>30000$, he prefers to stop posting in the second venue, the one with a higher spread, in order to maximize his chances of being executed in the first one. Similar interpretations apply for $t=0.1$. \\

In Figures \ref{limits_venue_imb} and \ref{volumes_venue_imb}, we show the choice of limits and volumes of the trader if the imbalance is more favorable in the second venue. in Figure \ref{limits_venue_imb}, we observe for $t=0.5$ that the trader posts in the first venue at the second best limit for $q\in [12000,15000]$, at the first limit for $q\in [15000,20000]$ and at a new best limit for $q\in [20000,32000]$. At the same time, he posts in the first limit of the second venue when $q\in [12000,22000]$ and at a new best limit for $q\in [22000,50000]$. We see that the trader prefers to post at a higher limit in the second venue because of the higher probability of execution due to a more favorable imbalance. For large inventories, he stops posting in the first venue in order to increase his probability of execution using limit orders in the second venue at a new best limit. Same results hold for $t=0.1$. \\

In Figure \ref{volumes_venue_imb}, we see that the trader posts a majority of his volume in the second venue due to a more favorable imbalance. When his inventory is not too high, he collects the spread from both venues. However, when his inventory is relatively high, he sends all the volume to the first venue in order to increase the probability of filling. 

\begin{figure}[H]

\begin{minipage}[c]{.46\linewidth}
  \begin{center}
      \includegraphics[width=0.9\textwidth]{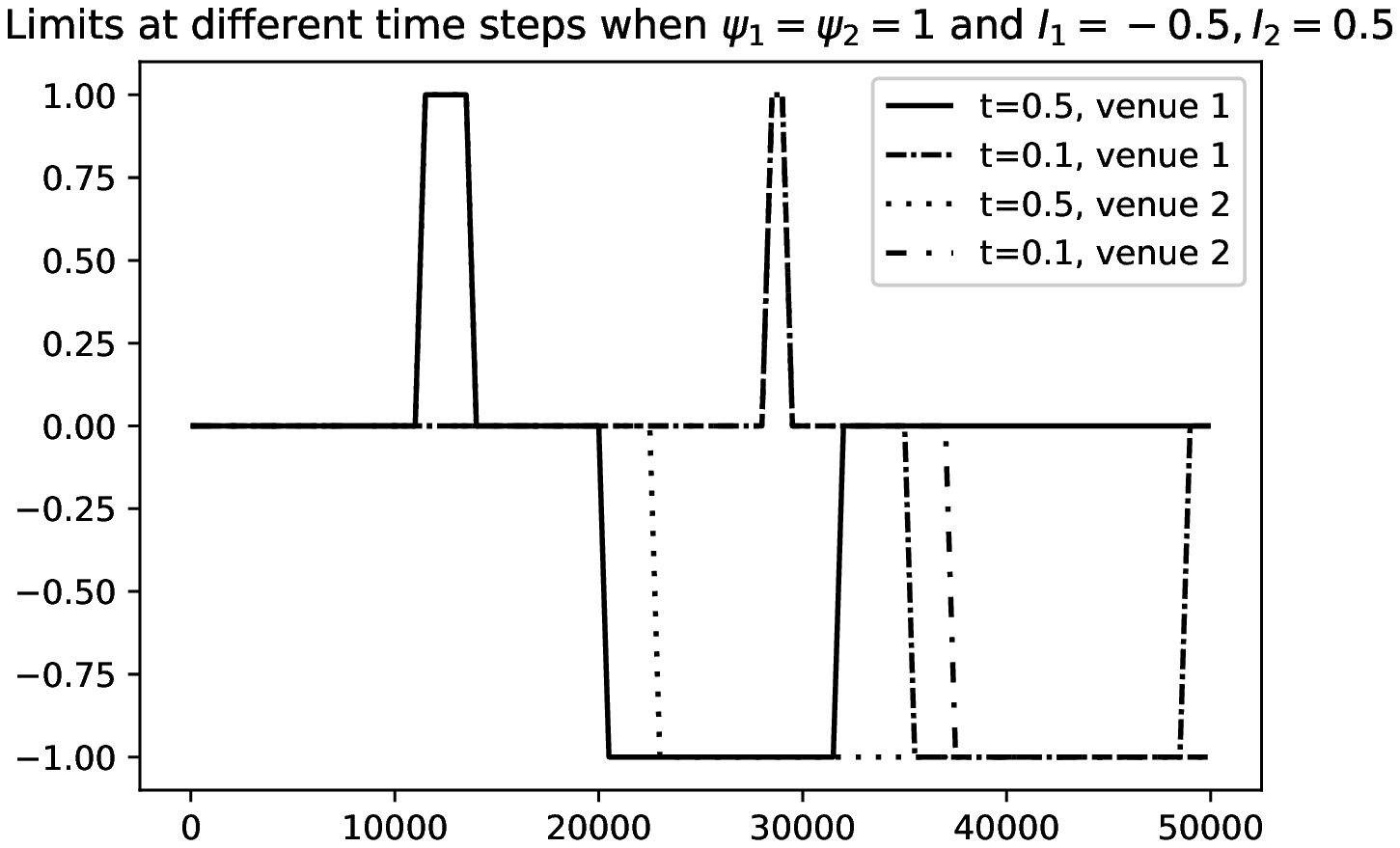}
      \vspace{-7mm}
      \caption{Limit order strategy, $\psi^1 =\psi^2 = \delta,$ \protect\\$ I^1 =-0.5, I^2 = 0.5$.}\label{limits_venue_imb}
    \end{center}
\end{minipage} \hfill
\begin{minipage}[c]{.46\linewidth}
   \begin{center}
       \includegraphics[width=0.9\textwidth]{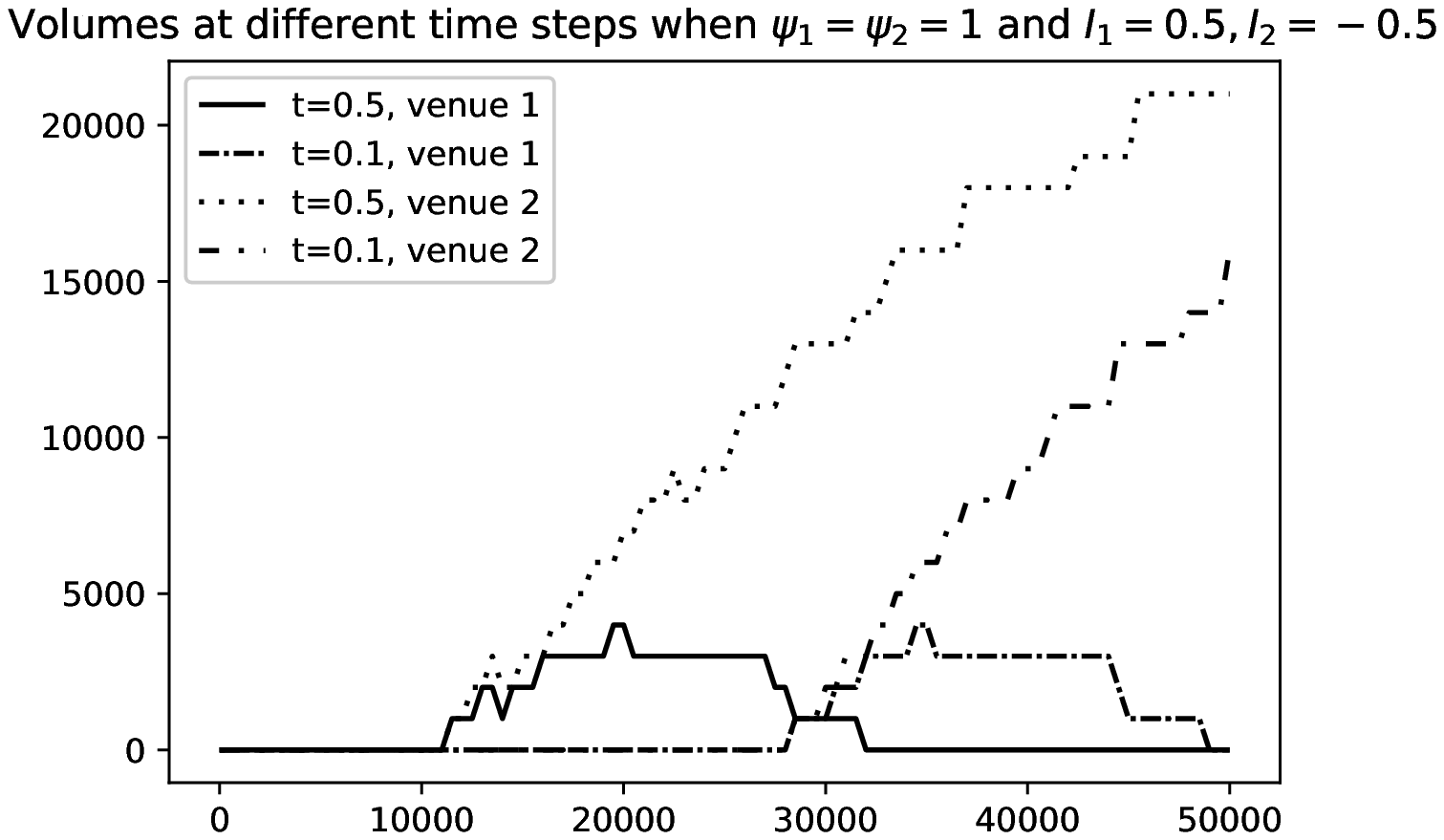}
       \vspace{-7mm}
       \caption{Volume strategy, $\psi^1 =\psi^2 = \delta,$\protect\\$ I^1 =-0.5, I^2 = 0.5$.}\label{volumes_venue_imb}
     \end{center}
  \end{minipage} 
\end{figure}

We now describe the strategies on the limits and the volumes obtained by a reinforcement learning approach.

\subsubsection{Strategy: limit orders and volumes with neural networks}

We plot in Figures \ref{limits_venue1_nets} and \ref{limits_venue2_nets} the strategies on the limits used by the trader. As soon as limits are represented by probabilities to send an order to each precise limit, for graphical representation, we plot the limit corresponding to the highest of the three probabilities. We see that the choice of the limits is in line with the ones of Figures \ref{limits_venue1} and \ref{limits_venue2} up to states where optimal order volume is at $0$ (in this case limit values are indistinguishable for finite differences). When the trader is at the beginning of the slice, for a small inventory, he prefers to collect a higher spread by being executed at the second best limit. When he is near the end of the slice, he prefers to be filled at a less favorable price, at the best or new best limit, in order to lower his execution risk. We can also see that neural networks preserve the monotonicity of the optimal limit function.
\vspace{-3mm}
\begin{figure}[H]
\begin{minipage}[c]{.48\linewidth}
  \begin{center}
      \includegraphics[width=0.88\textwidth]{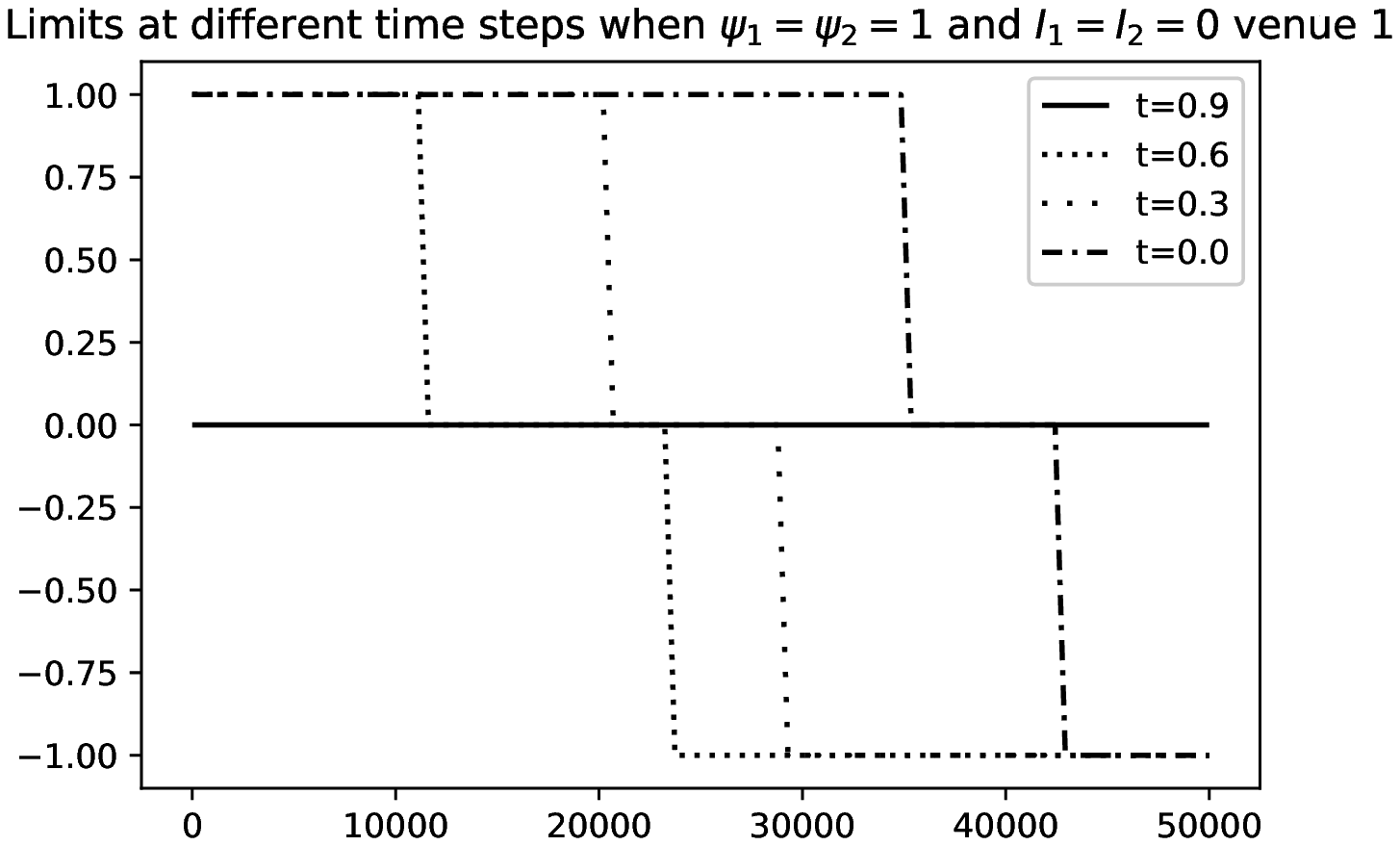}
      \vspace{-3mm}
      \caption{Limit order strategy in the first \protect\\venue, $\psi^1 = \psi^2 = \delta, I^1 = I^2 = 0$ using neural \protect\\ networks.}\label{limits_venue1_nets}
    \end{center}
\end{minipage} \hfill
\begin{minipage}[c]{.46\linewidth}
   \begin{center}
       \includegraphics[width=0.9\textwidth]{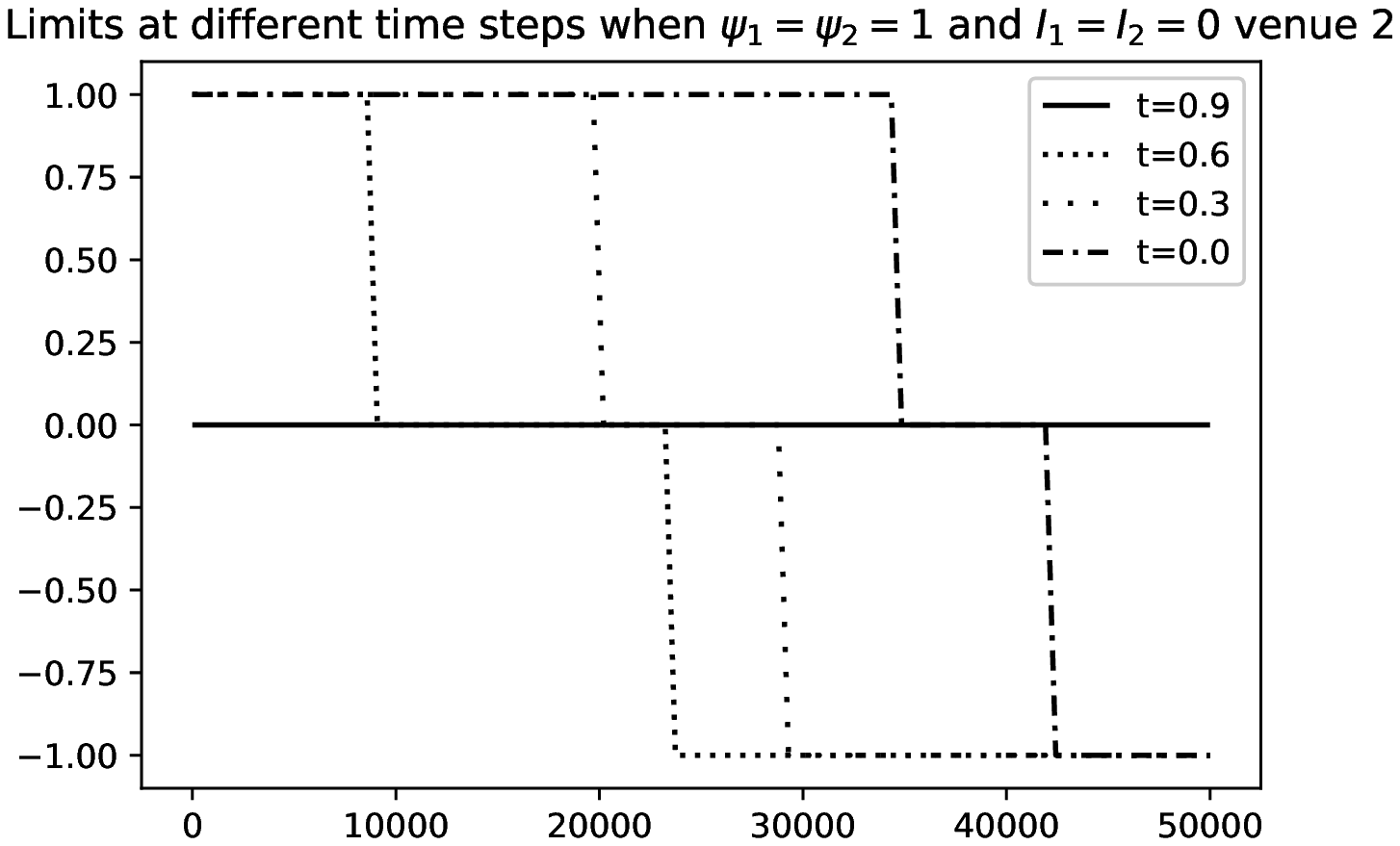}
       \vspace{-3mm}
       \caption{Limit order strategy in the second venue, $\psi^1 = \psi^2 = \delta, I^1 = I^2 = 0$ using neural networks.}\label{limits_venue2_nets}
     \end{center}
  \end{minipage} 
\end{figure}
\vspace{-3mm}
In Figures \ref{volumes_venue1_nets} and \ref{volumes_venue2_nets}, we plot the posted volumes of the trader in both venues for the same spread and imbalance. We see that the strategy is a smoothed approximation of the one obtained using finite differences in Figures \ref{volumes_venue1} and \ref{volumes_venue2}. We see that at the very beginning of the slice, the trader is not going to trade if his inventory is already small enough. The strategy in both venues is the same up to some negligible numerical effects.
\begin{figure}[H]

\begin{minipage}[c]{.46\linewidth}
  \begin{center}
      \includegraphics[width=0.9\textwidth]{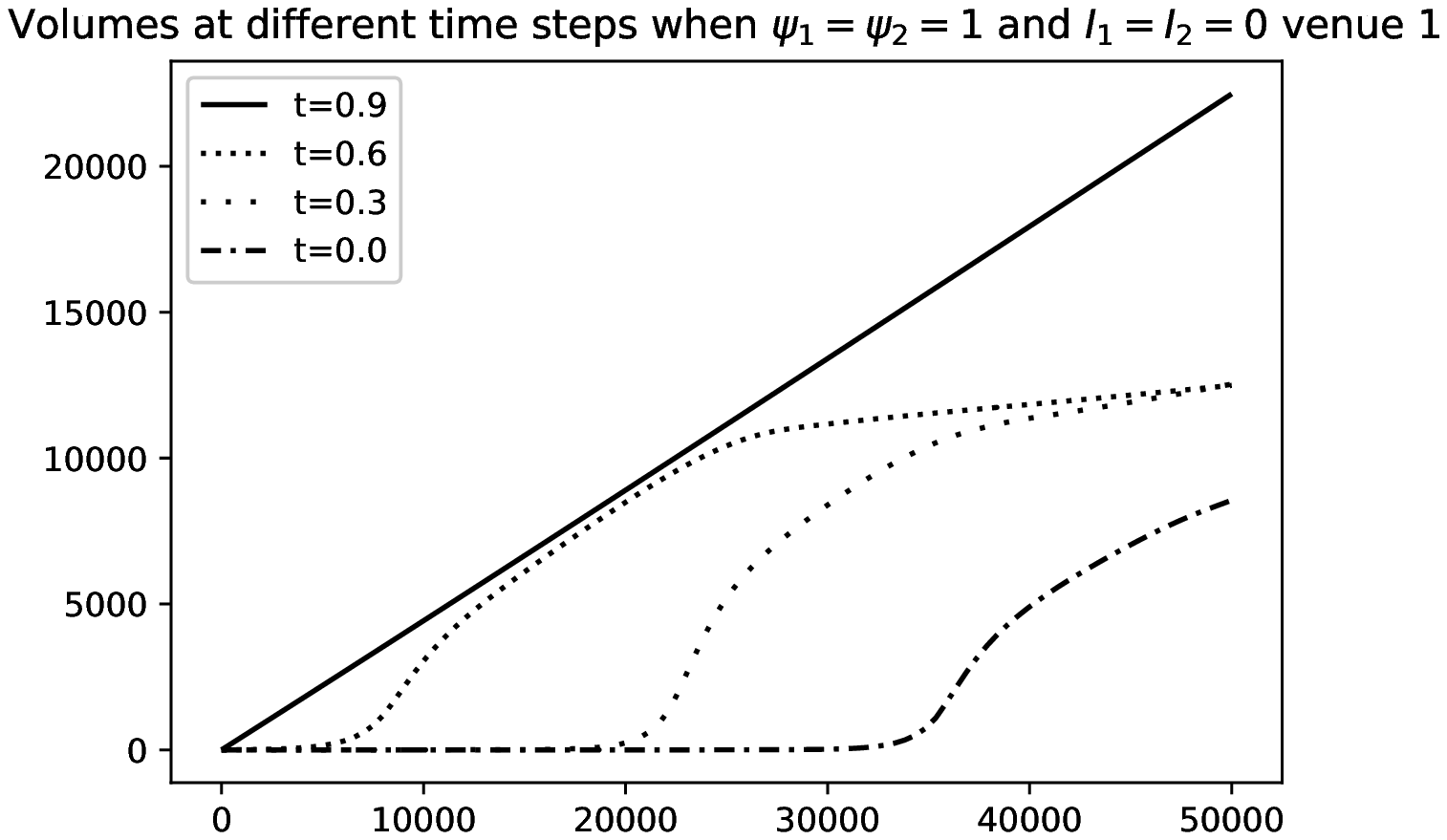}
      \vspace{-3mm}
      \caption{Volume posted in the first venue, $\psi^1 = \psi^2 = \delta, I^1 = I^2 = 0$ using neural networks.}\label{volumes_venue1_nets}
    \end{center}
\end{minipage} \hfill
\begin{minipage}[c]{.46\linewidth}
   \begin{center}
       \includegraphics[width=0.9\textwidth]{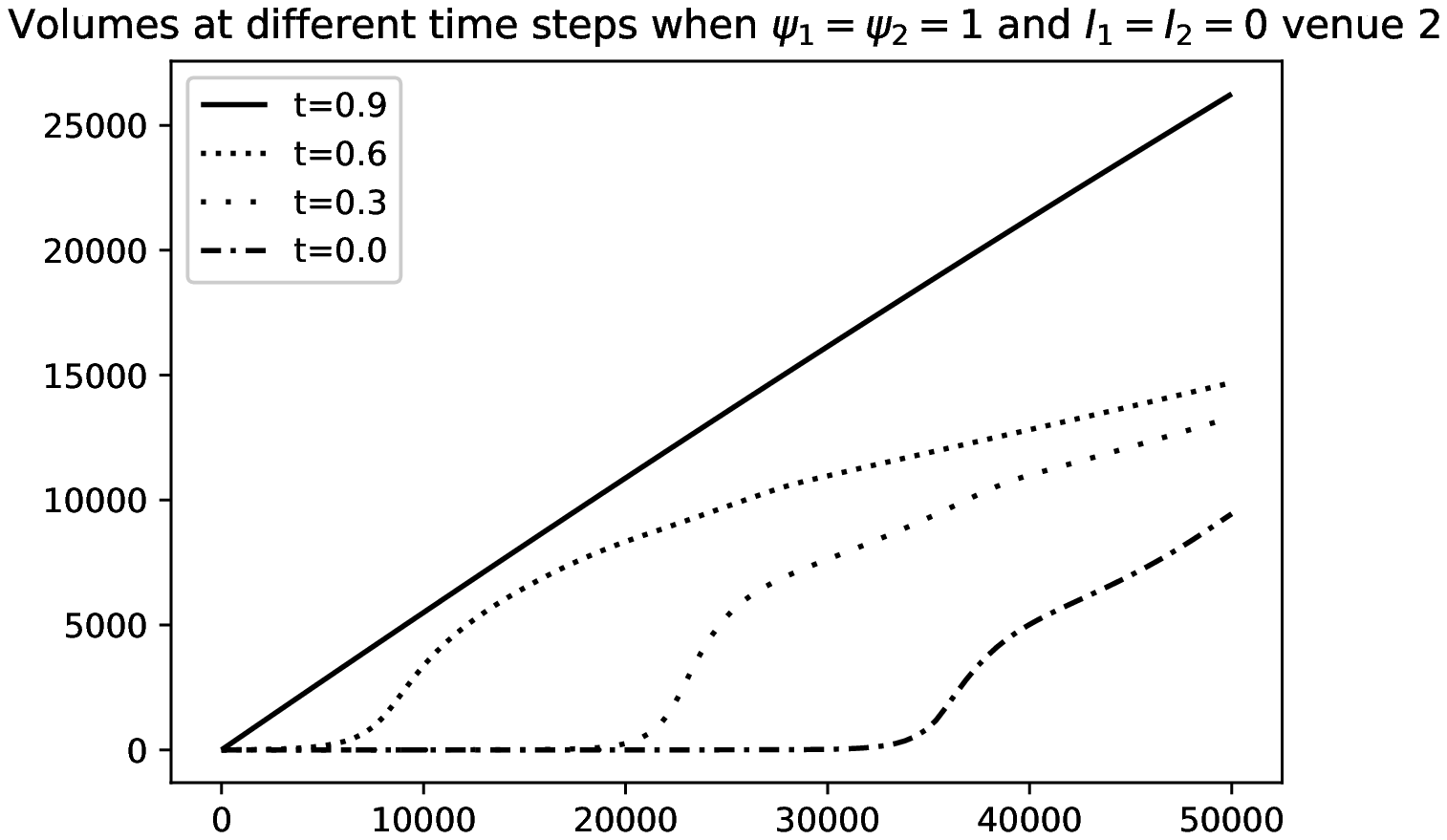}
       \vspace{-3mm}
       \caption{Volume posted in the second venue, $\psi^1 = \psi^2 = \delta, I^1 = I^2 = 0$ using neural networks.}\label{volumes_venue2_nets}
     \end{center}
  \end{minipage} 
\end{figure}

If the spread of the second venue is higher, we see in Figure \ref{limits_venue_spr_nets} that the strategy with the limits is the same as in Figure \ref{limits_venue_spr}. It is interesting to note in Figure \ref{volumes_venue_spr_nets} that the trader does not stop posting in the second venue, as in Figure \ref{volumes_venue_spr}, again because of the approximation coming from neural networks. However, this behavior enables to perform some exploration of the venue parameters. For example, if the trader follows the strategy given by finite differences in Figure \ref{volumes_venue_spr}, he posts a volume equal to $0$ in the second venue when $q>32000$ for $t=0.5$. However, if the trader underestimates the prior on the filling probability in the second venue $\hat \lambda^2$, he will keep sending orders in the first venue, neglecting the possibility of splitting his orders which can potentially improve his execution. Moreover, Figures \ref{value_funcions_04_nets} and~\ref{value_funcions_59_nets} show that this slight difference in the obtained controls does not change drastically the performance of the trader in terms of the value function. 

\begin{figure}[H]

\begin{minipage}[c]{.48\linewidth}
  \begin{center}
      \includegraphics[width=0.88\textwidth]{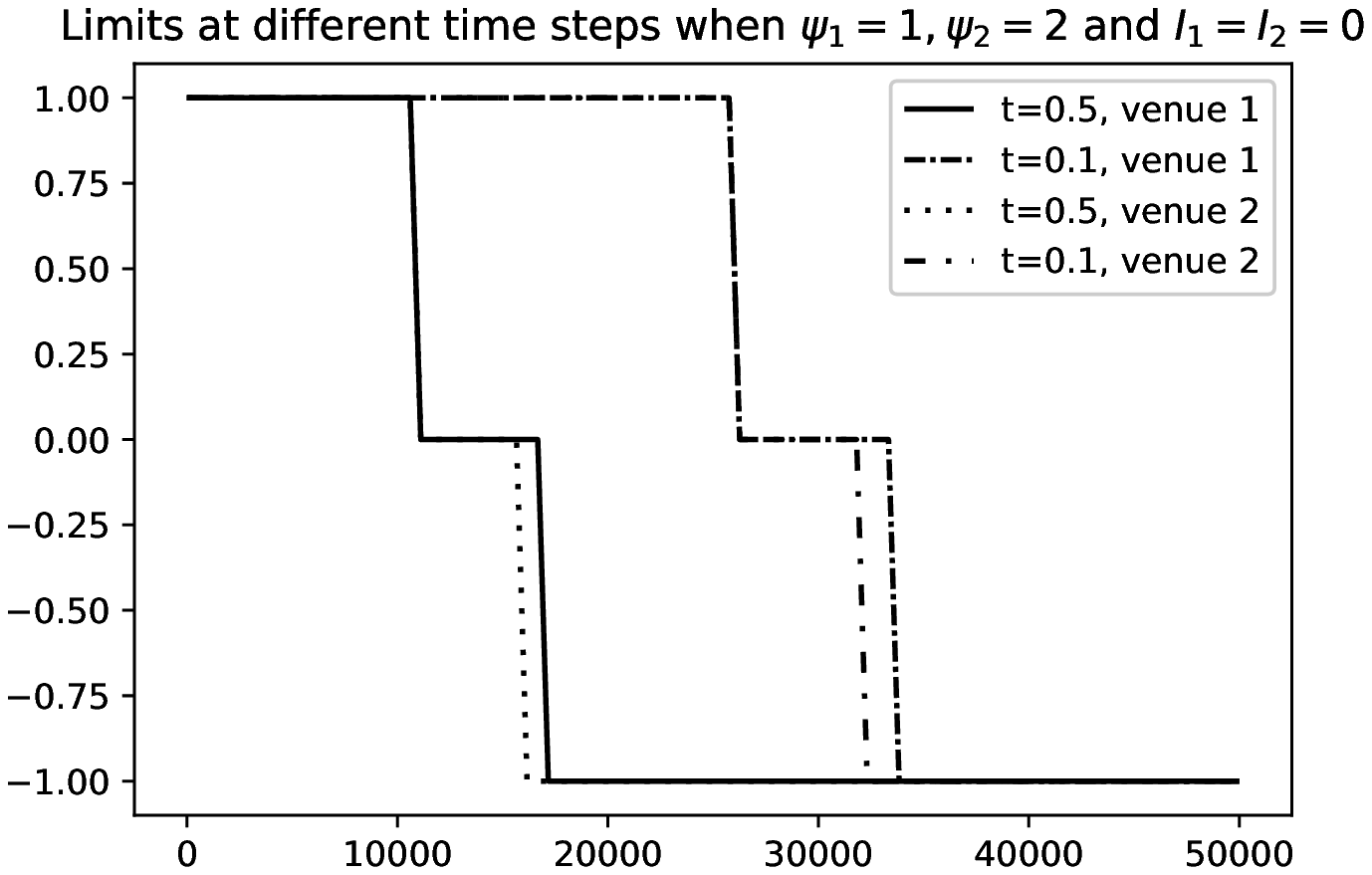}
      \vspace{-3mm}
      \caption{Limit order strategy, $\psi^1 =\delta, \psi^2 = 2\delta,$\protect\\ $I^1 = I^2 = 0$ using neural networks.}\label{limits_venue_spr_nets}
    \end{center}
\end{minipage} \hfill
\begin{minipage}[c]{.46\linewidth}
   \begin{center}
       \includegraphics[width=0.9\textwidth]{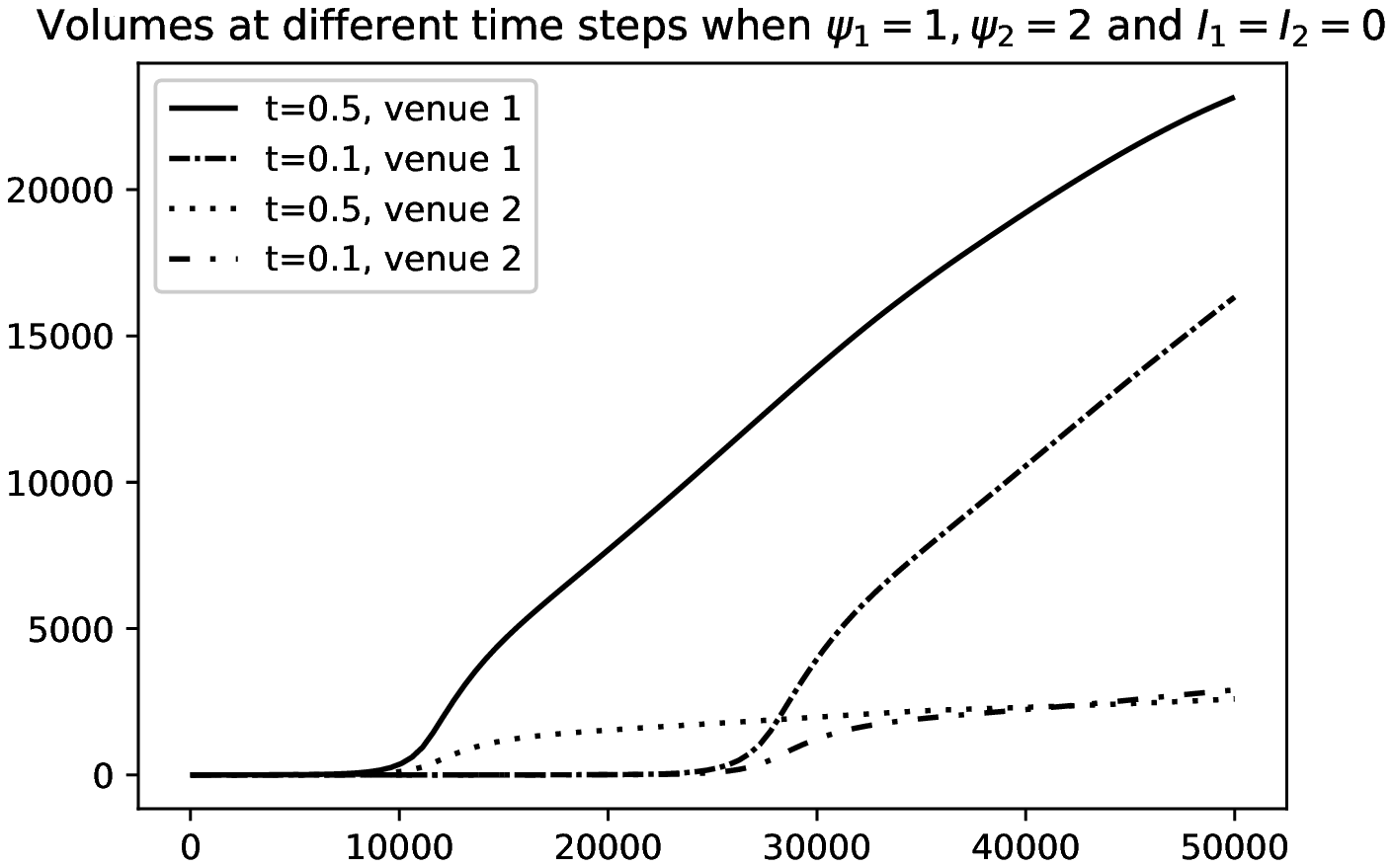}
       \vspace{-3mm}
       \caption{Volume strategy, $\psi^1 = \delta, \psi^2 = 2\delta,$\protect\\ $I^1 = I^2 = 0$ using neural networks.}\label{volumes_venue_spr_nets}
     \end{center}
  \end{minipage} 
\end{figure}

The same comments apply to Figures \ref{limits_venue_imb_nets} and \ref{volumes_venue_imb_nets}, where we see that the trader posts a small but nonzero volume in the first venue with a less favorable imbalance which potentially allows to perform exploration in this venue and faster improve parameter estimations. 
\begin{figure}[H]

\begin{minipage}[c]{.46\linewidth}
  \begin{center}
      \includegraphics[width=0.9\textwidth]{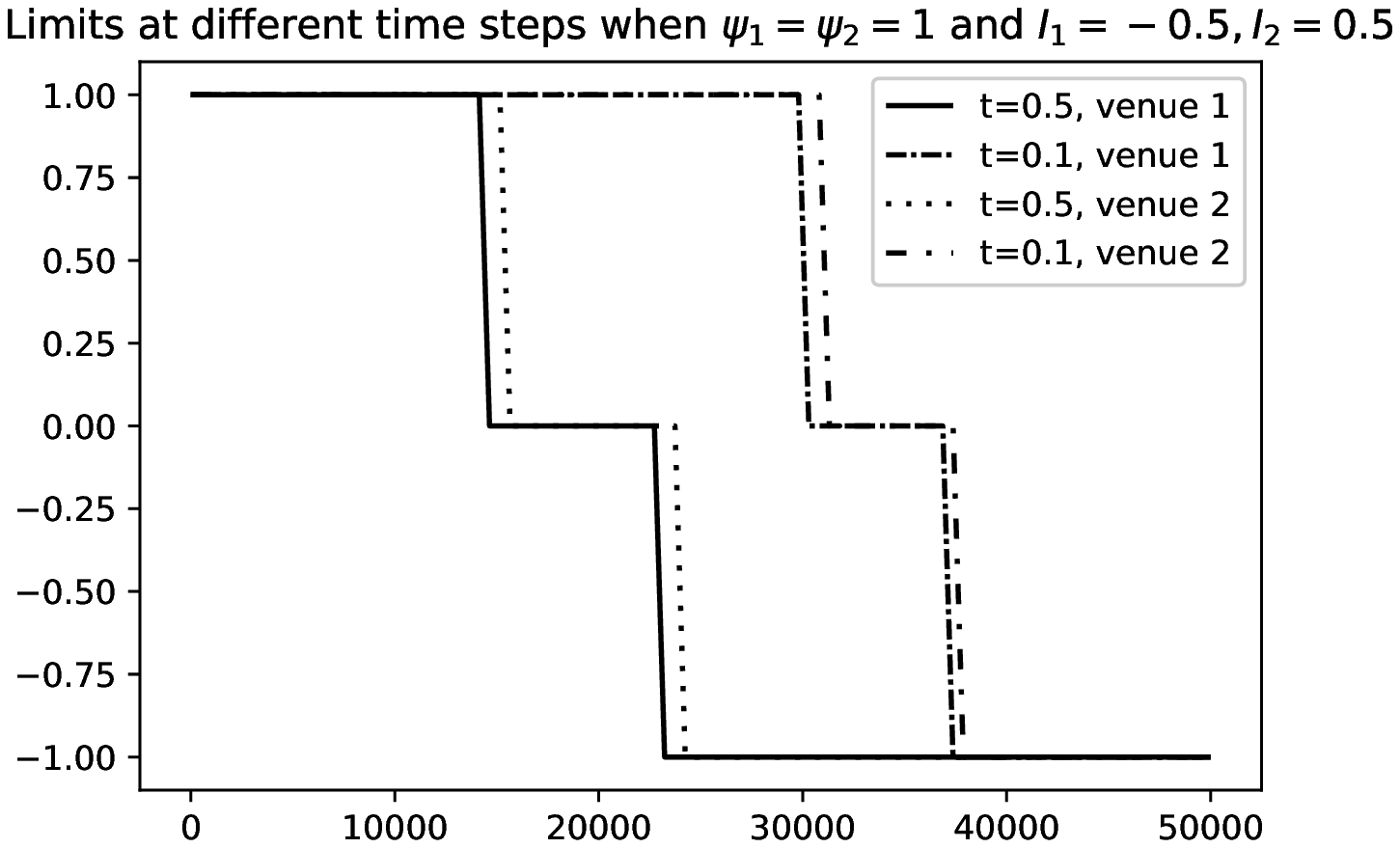}
      \vspace{-3mm}
      \caption{Limit order strategy, $\psi^1 =\psi^2 = \delta,$\protect\\ $I^1 =-0.5, I^2 = 0.5$ using neural networks.}\label{limits_venue_imb_nets}
    \end{center}
\end{minipage} \hfill
\begin{minipage}[c]{.46\linewidth}
   \begin{center}
       \includegraphics[width=0.9\textwidth]{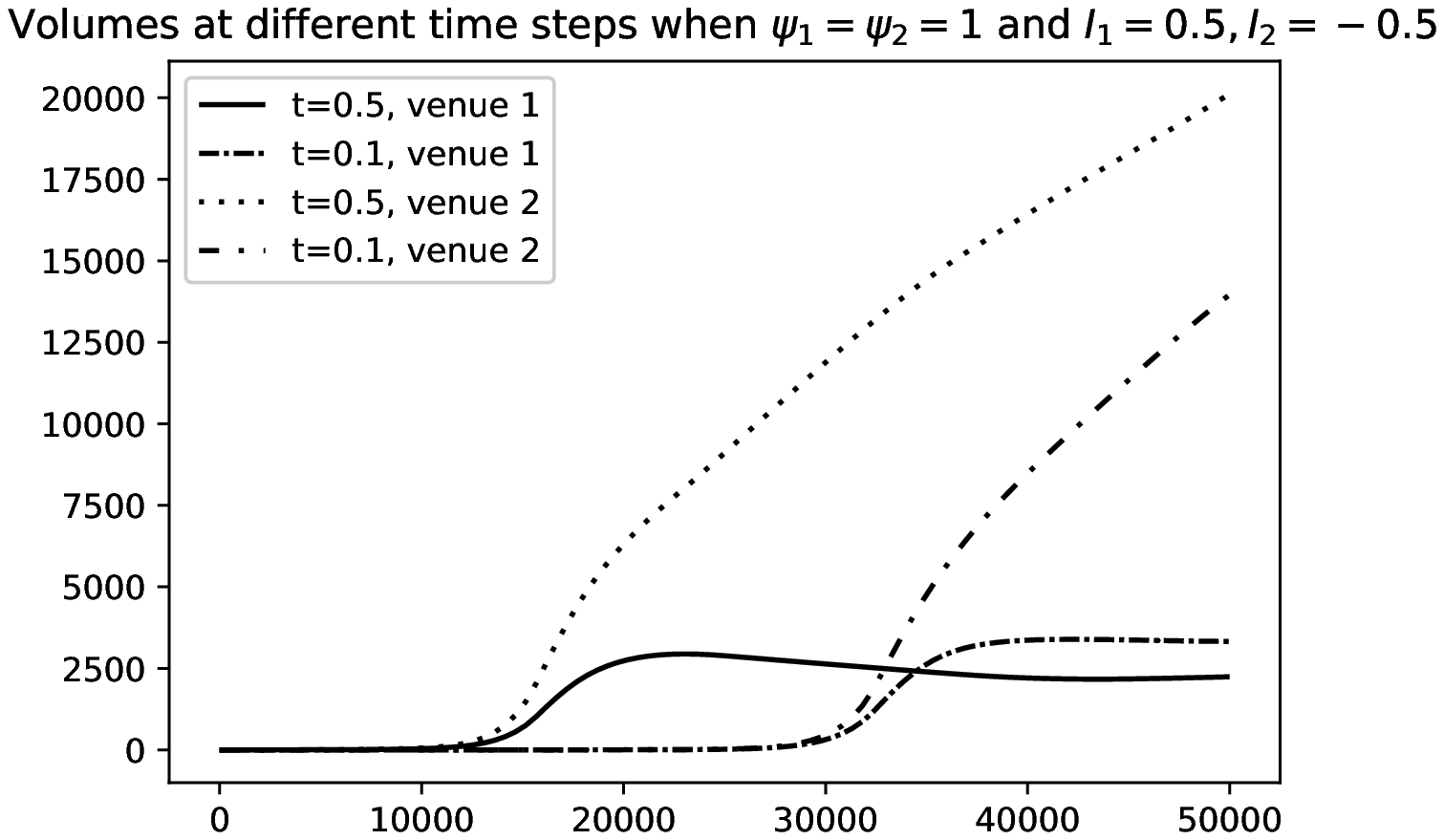}
       \vspace{-3mm}
       \caption{Volume strategy, $\psi^1 =\psi^2 = \delta,$\protect\\ $I^1 =-0.5, I^2 = 0.5$ using neural networks.}\label{volumes_venue_imb_nets}
     \end{center}
  \end{minipage} 
\end{figure}

\subsection{Two different venues}

In this section, we analyze the behavior of the trader believing that the first venue is better than the second venue in terms of filling rate. We compare the solutions obtained via finite difference schemes and neural networks.

\subsubsection{Value function}

We show in Figures \ref{value_funcions_04diff} and \ref{value_funcions_59diff} the evolution of the value function of the trader during a slice of execution, obtained through the finite difference method.

\begin{figure}[H]

\begin{minipage}[c]{.46\linewidth}
  \begin{center}
      \includegraphics[width=0.9\textwidth]{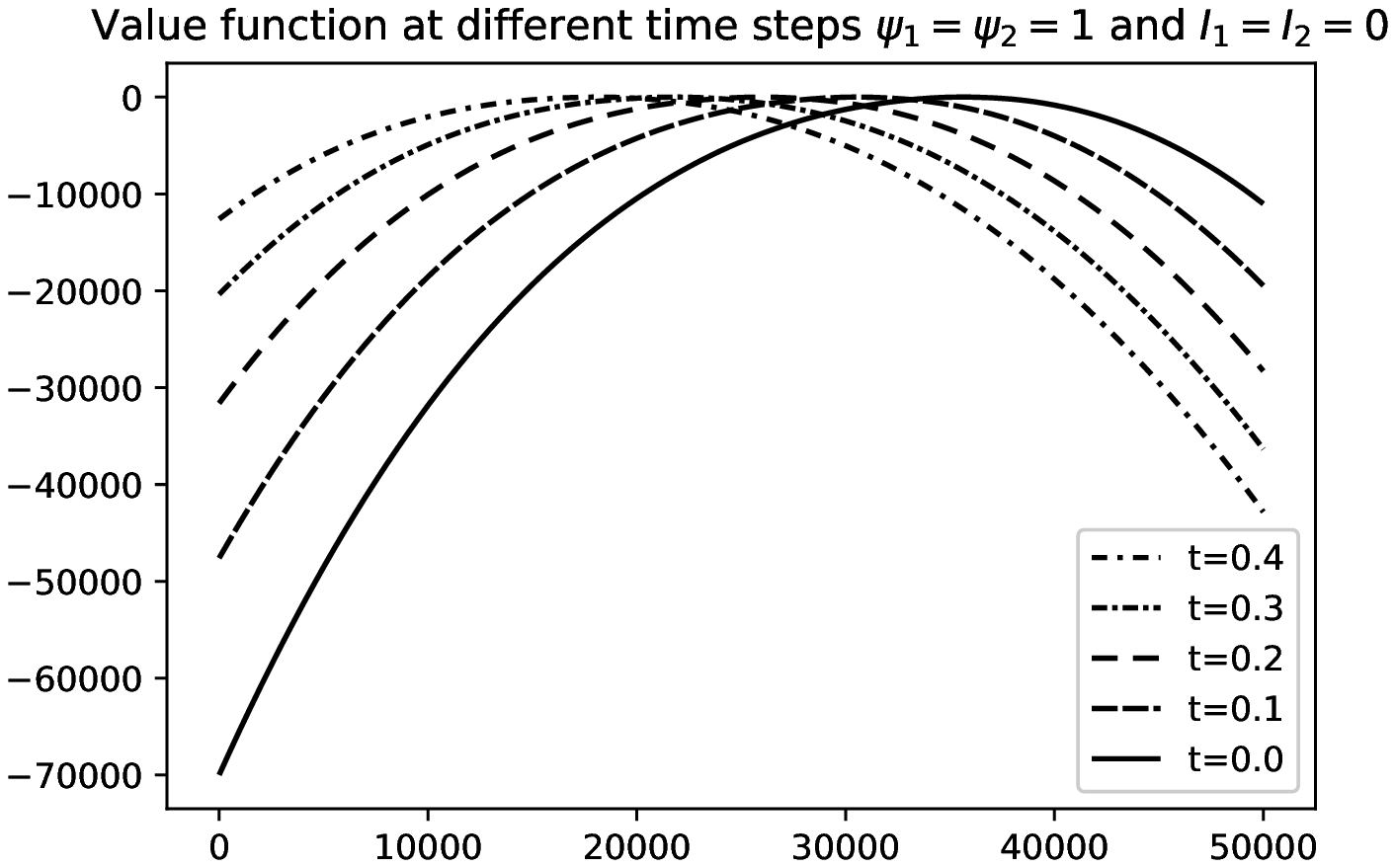}
      \vspace{-3mm}
      \caption{Evolution of the value function $v$ between $t=0$ and $t=0.4$.}\label{value_funcions_04diff}
    \end{center}
\end{minipage} \hfill
\begin{minipage}[c]{.46\linewidth}
   \begin{center}
       \includegraphics[width=0.9\textwidth]{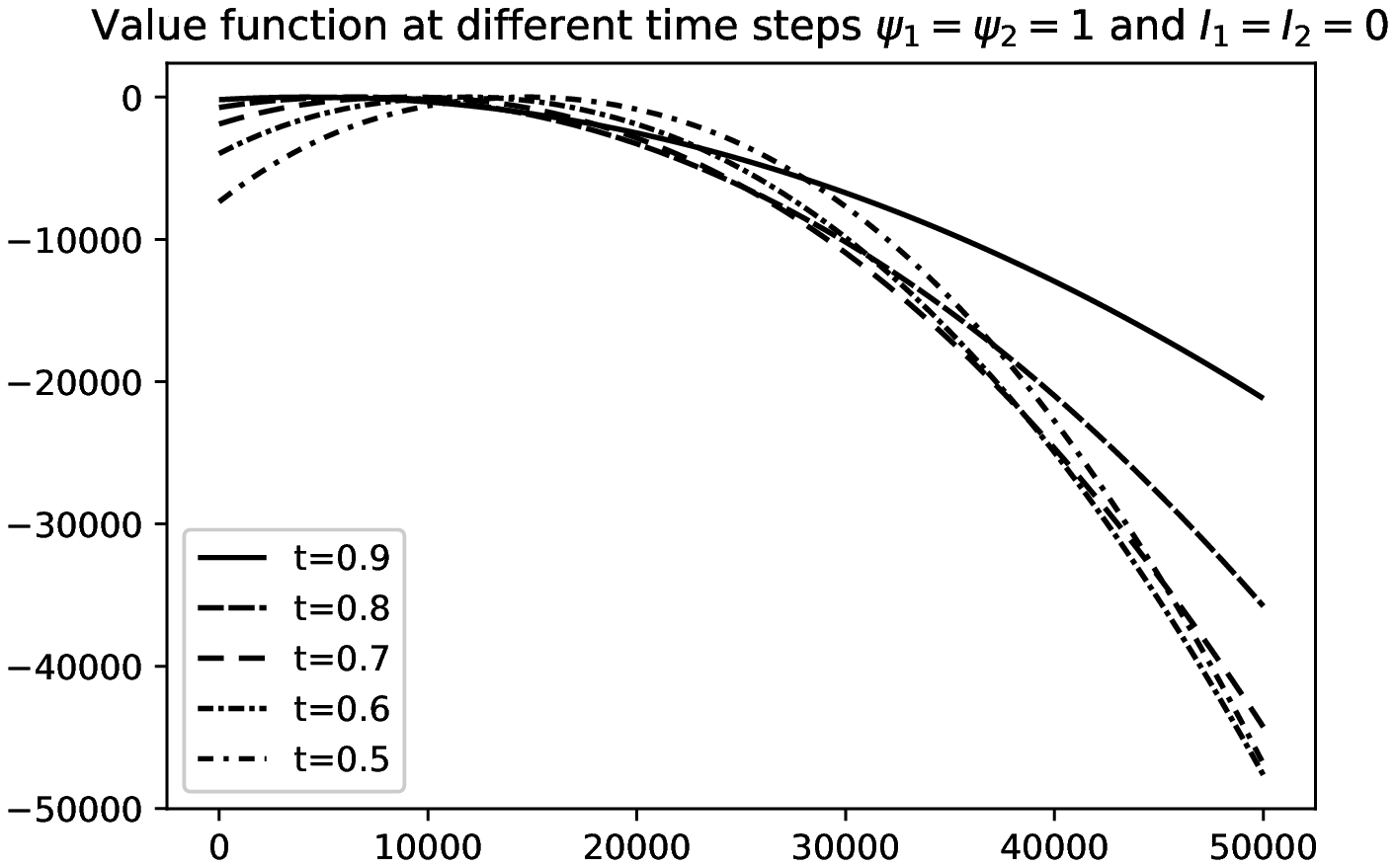}
       \vspace{-3mm}
       \caption{Evolution of the value function $v$ between $t=0.5$ and $t=0.9$.}\label{value_funcions_59diff}
     \end{center}
  \end{minipage} 
\end{figure}

One can see that the value function deteriorates compared to the previous example, which is predictable in view of the fact that one of the venues is exactly like in the above example, and another one is worse in terms of filling ratio. For example in Figure \ref{value_funcions_59diff}, the minimum of the function $v$ at $t=0.5$ when $q=50000$ is $-49000$ compared to a minimum of $-45000$ in the example above. This is a natural consequence of a worse prior distribution on the filling ratio of the second venue while keeping the prior on the first venue unchanged.

\begin{figure}[H]

\begin{minipage}[c]{.46\linewidth}
  \begin{center}
      \includegraphics[width=0.9\textwidth]{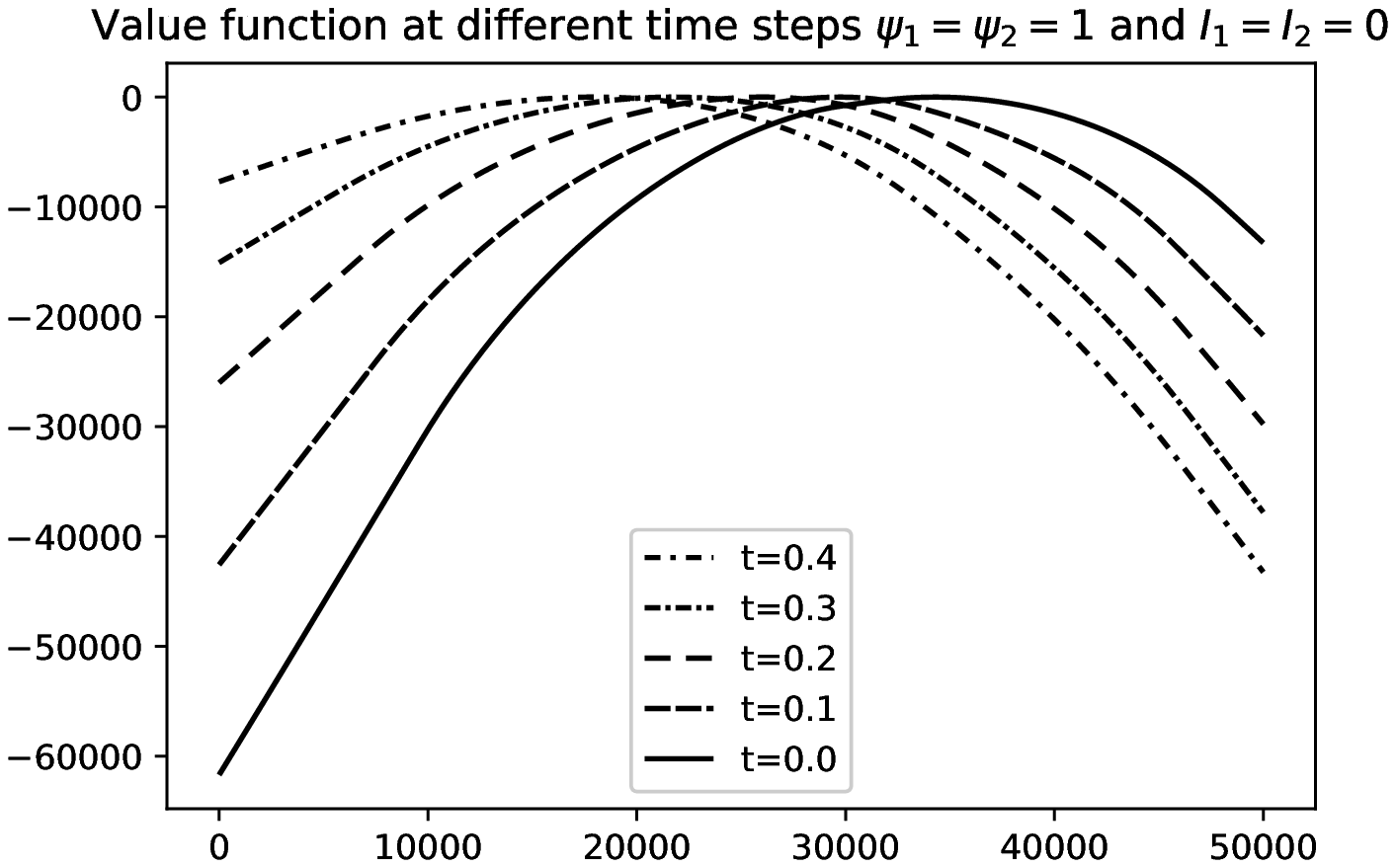}
      \vspace{-3mm}
      \caption{Evolution of the value function $v$ between $t=0$ and $t=0.4$ using neural networks.}\label{value_funcions_04diff_nets}
    \end{center}
\end{minipage} \hfill
\begin{minipage}[c]{.46\linewidth}
   \begin{center}
       \includegraphics[width=0.9\textwidth]{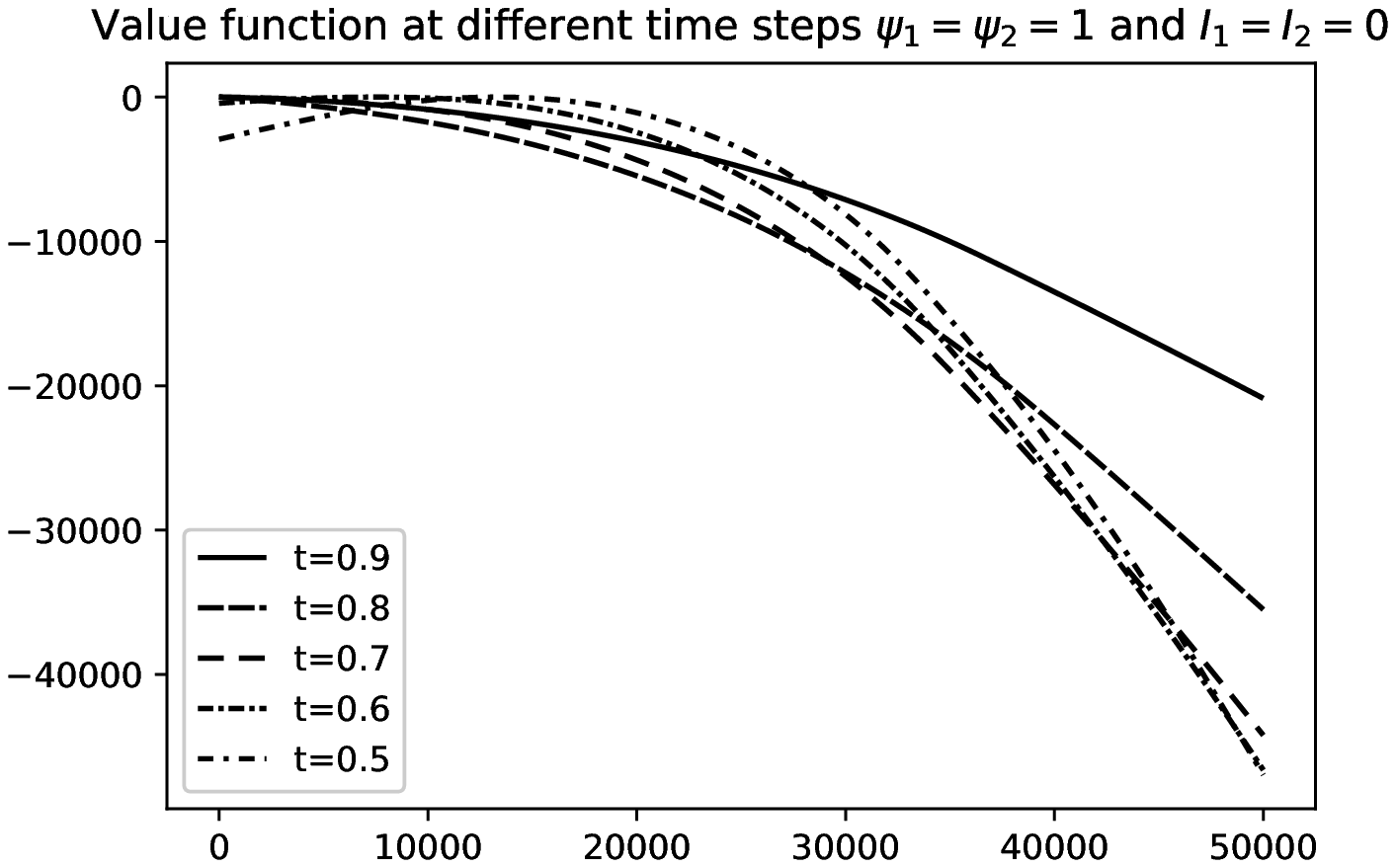}
       \vspace{-3mm}
       \caption{Evolution of the value function $v$ between $t=0.5$ and $t=0.9$ using neural networks.}
       \label{value_funcions_59diff_nets}
     \end{center}
  \end{minipage} 
\end{figure}

We check in Figures \ref{value_funcions_04diff_nets} and \ref{value_funcions_59diff_nets} that we obtain a similar shape for the value function using neural networks.\\

We now describe the strategy of the trader on the limits and the posted volumes and compare it to the case of two identical venues.

\subsubsection{Strategy: limit orders and volumes with finite difference schemes}

In Figures \ref{limits_venue1_diff} and \ref{limits_venue2_diff}, we show the limit order strategy of the trader in the two venues for the same spreads and imbalances. As the second venue is less favorable for execution, the trader prefers to create a new best limit for smaller inventories. For example, when $t=0.6$, he posts an order on the new best limit starting from $q=19000$, and in the second venue, he prefers to create a new limit starting from $q=18000$. Generally, either at the beginning or at the end of the slice, the trader prefers to post at a lower limit in the second venue in order to increase his execution rate there, sacrificing the spread that could have been collected. 

\begin{figure}[H]

\begin{minipage}[c]{.46\linewidth}
  \begin{center}
      \includegraphics[width=0.9\textwidth]{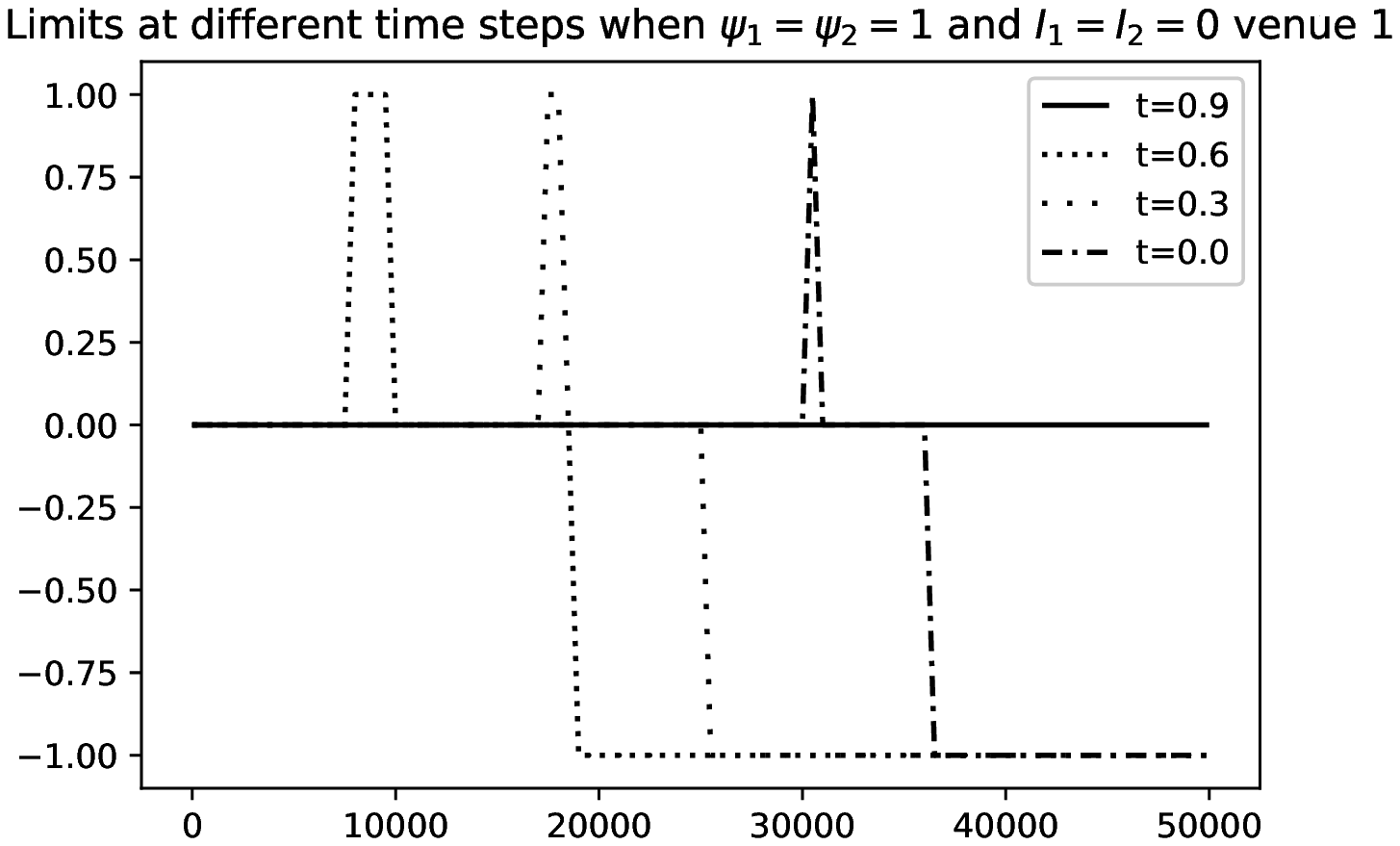}
      \vspace{-3mm}
      \caption{Limit order strategy in the first venue, $\psi^1 = \psi^2 = \delta, I^1 = I^2 = 0$.}\label{limits_venue1_diff}
    \end{center}
\end{minipage} \hfill
\begin{minipage}[c]{.46\linewidth}
   \begin{center}
       \includegraphics[width=0.9\textwidth]{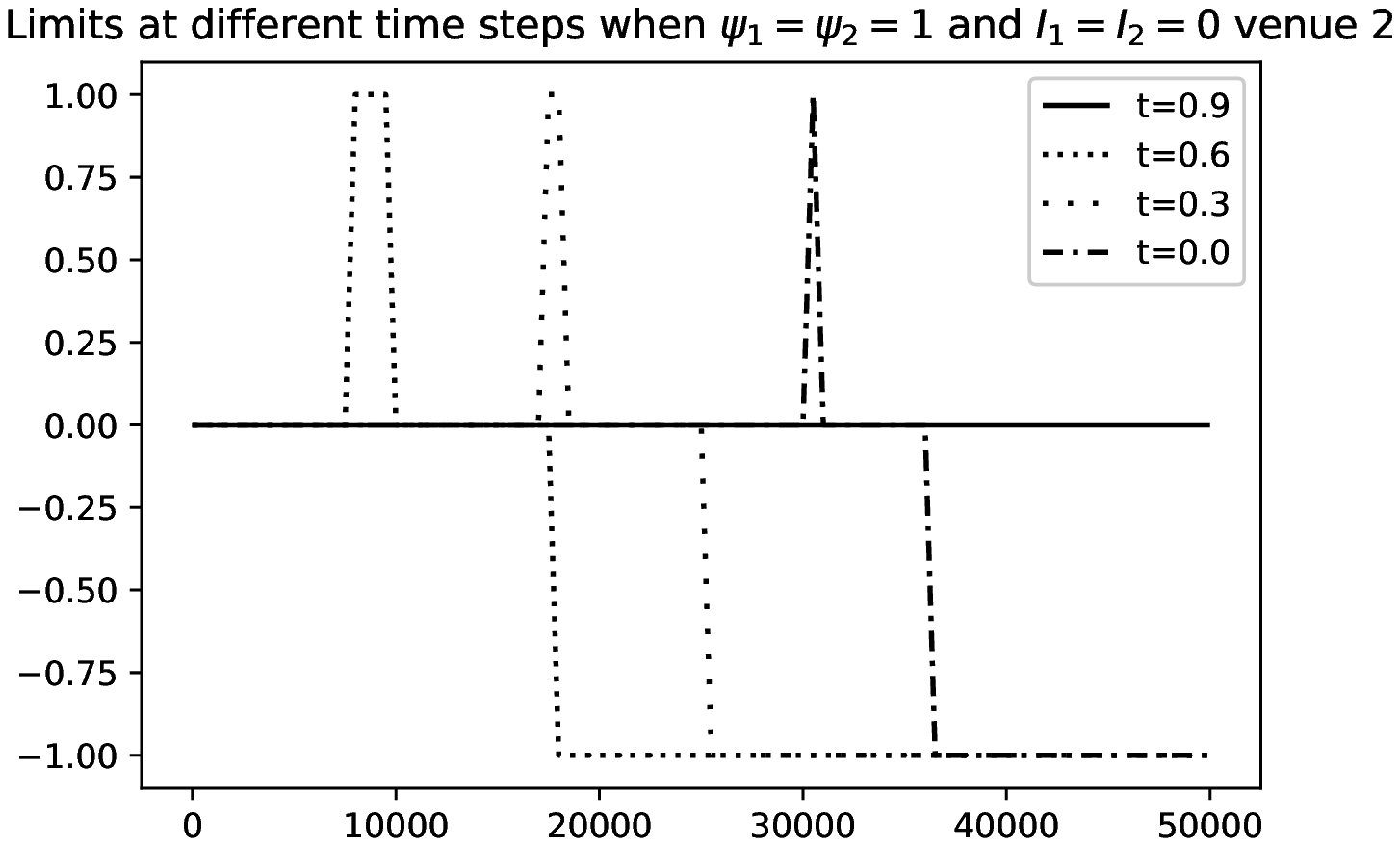}
       \vspace{-3mm}
       \caption{Limit order strategy in the second venue, $\psi^1 = \psi^2 = \delta, I^1 = I^2 = 0$.}\label{limits_venue2_diff}
     \end{center}
  \end{minipage} 
\end{figure}

The strategy of the trader differs drastically in terms of order volumes. In Figures \ref{volumes_venue1_diff} and \ref{volumes_venue2_diff}, we see that the trader posts the majority of his volume in the first venue. Especially when at $t=0.9$ the trader stops posting in the second venue to reduce his liquidity consumption and maximize his probability of execution in the first venue.

\begin{figure}[H]
\begin{minipage}[c]{.46\linewidth}
  \begin{center}
      \includegraphics[width=0.9\textwidth]{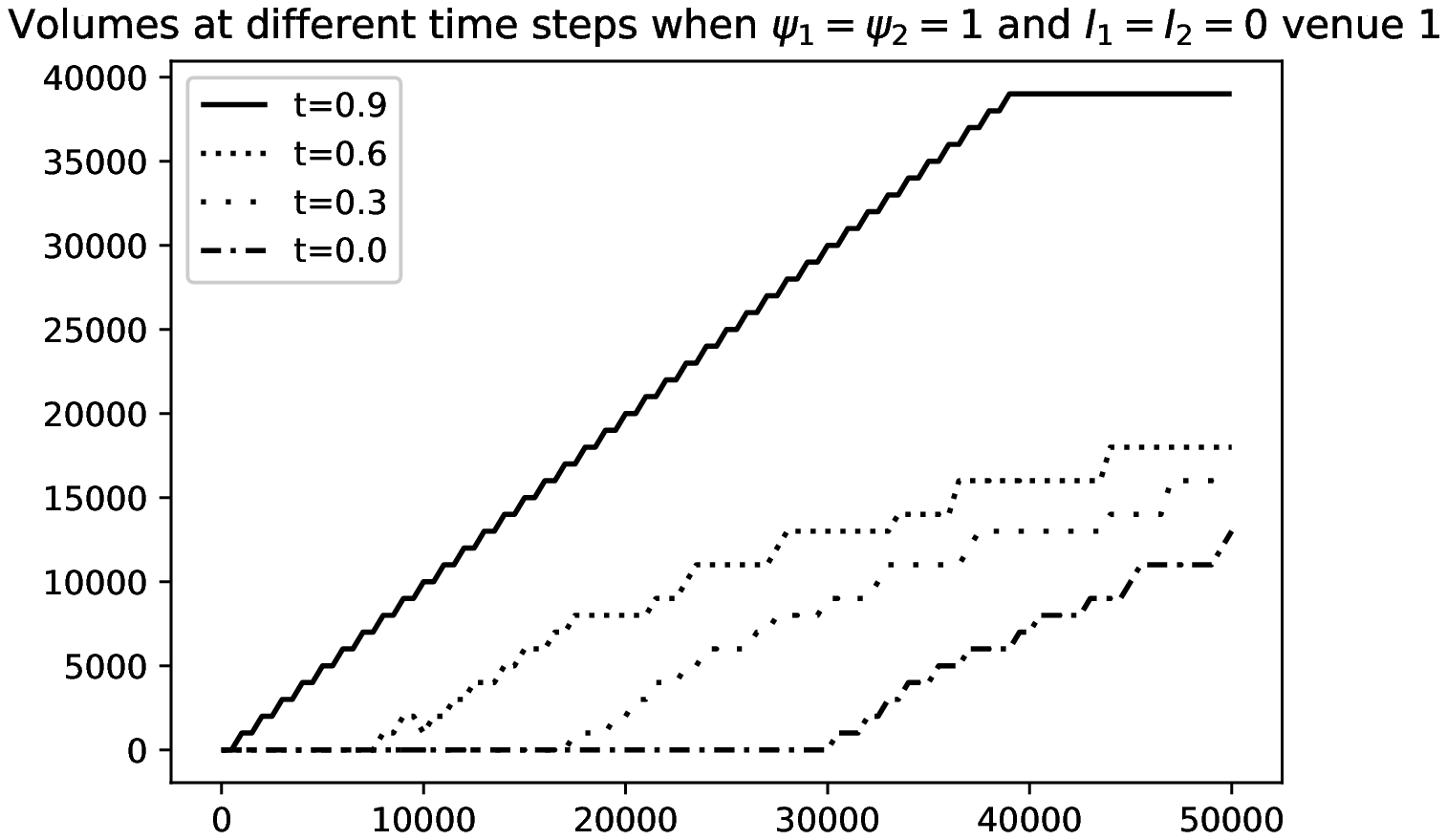}
      \vspace{-3mm}
      \caption{Volume posted in the first venue, $\psi^1 = \psi^2 = \delta, I^1 = I^2 = 0$.}\label{volumes_venue1_diff}
    \end{center}
\end{minipage} \hfill
\begin{minipage}[c]{.46\linewidth}
   \begin{center}
       \includegraphics[width=0.9\textwidth]{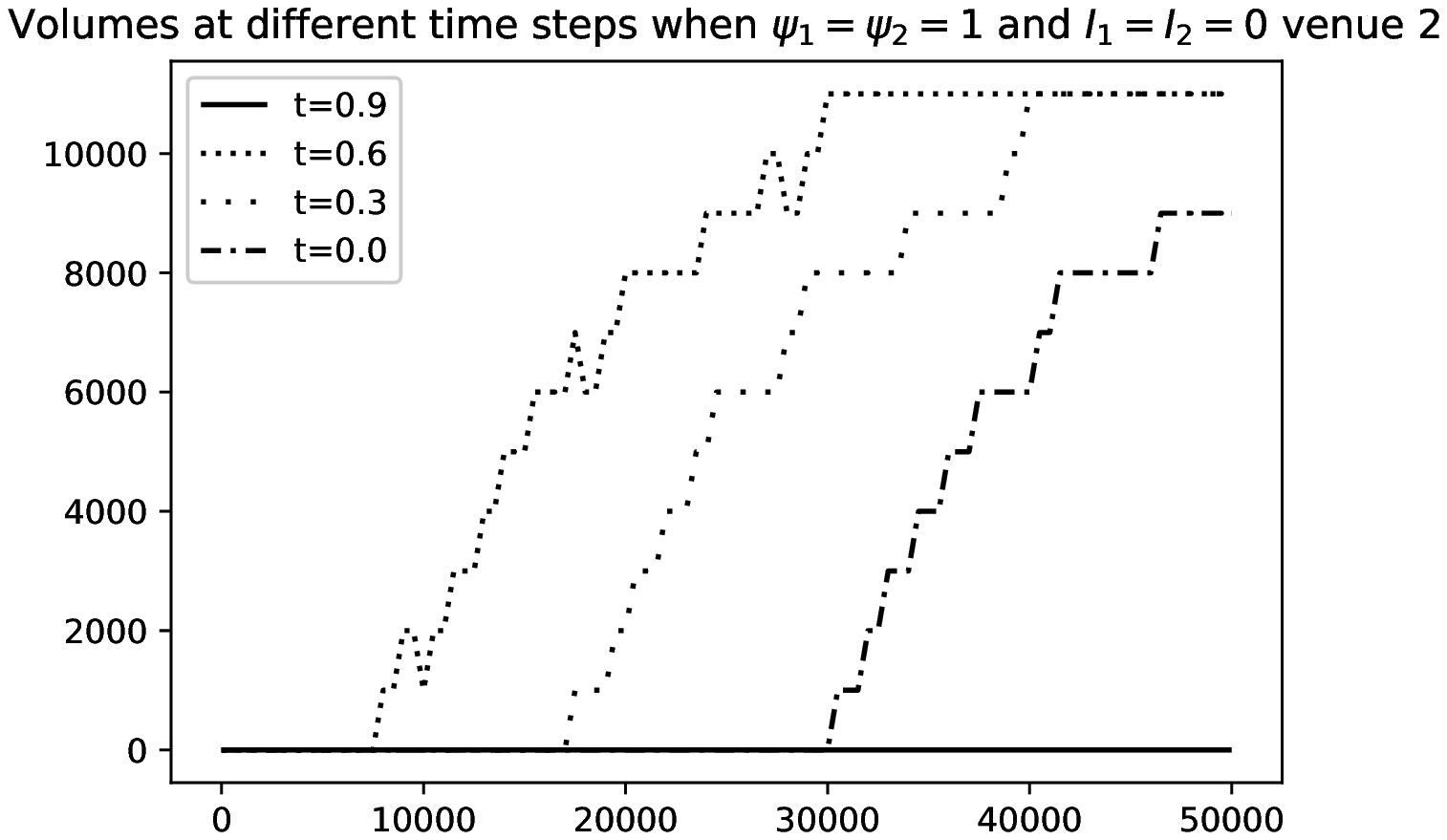}
       \vspace{-3mm}
       \caption{Volume posted in the second venue, $\psi^1 = \psi^2 = \delta, I^1 = I^2 = 0$.}\label{volumes_venue2_diff}
     \end{center}
  \end{minipage} 
\end{figure}

In Figures \ref{limits_venue_sprdiff} and \ref{volumes_venue_sprdiff}, we see the limits and the volumes recommended to the trader when the second venue has a higher spread, and the imbalances are equal. The trader posts an even smaller volume in the second venue, compared to Figure \ref{volumes_venue_spr}. As the filling rate is lower in the second venue, the trader decreases his liquidity consumption in this venue, because of the smaller probability of collecting a higher spread.\\

The strategy on the limits in Figure \ref{limits_venue_sprdiff} is also different from the one in Figure \ref{limits_venue_spr}. When $t=0.5$ and the two venues are the same, the trader posts at the second best limit in the second venue when $q\in[11000,13000]$, then at the first best limit when $q\in[13000,18000]$ and at a new best limit for $q\in[18000,30000]$. When the venues are different, the trader posts at the second best limit in the second venue for $q\in[10000,12000]$, at the first best limit for $q\in[12000,17000]$ and at a new best limit when $q\in[17000,19000]$. Therefore, when the second venue has a worse filling rate, the trader posts in the second venue earlier (for a higher inventory) and less compared to the case with two equivalent venues. 
\vspace{-3mm}
\begin{figure}[H]
\begin{minipage}[c]{.48\linewidth}
  \begin{center}
      \includegraphics[width=0.88\textwidth]{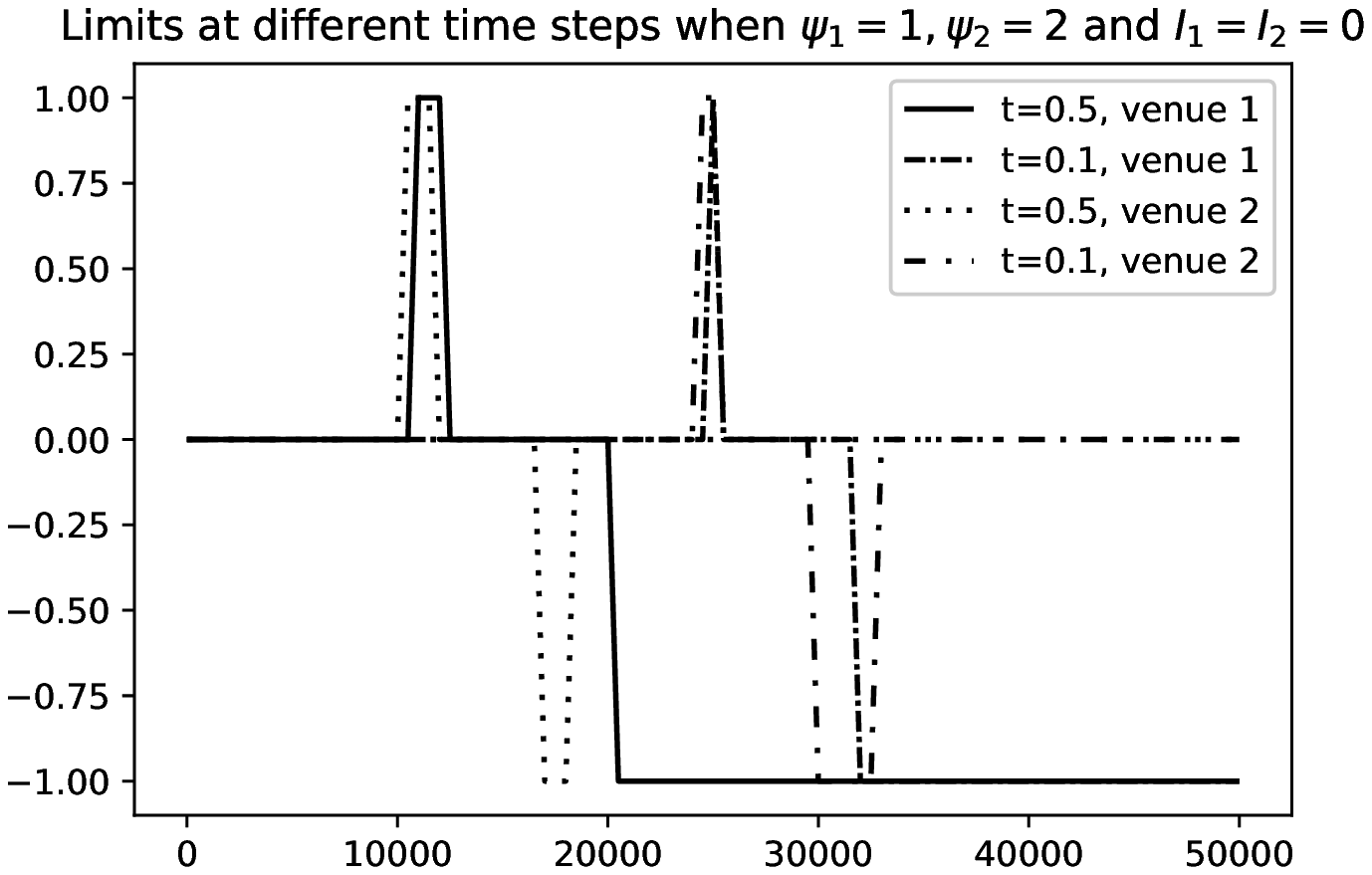}
      \vspace{-3mm}
      \caption{Limit order strategy, $\psi^1 =\delta, \psi^2 = 2\delta,$ \protect\\ $I^1 = I^2 = 0$.}\label{limits_venue_sprdiff}
    \end{center}
\end{minipage} \hfill
\begin{minipage}[c]{.46\linewidth}
   \begin{center}
       \includegraphics[width=0.9\textwidth]{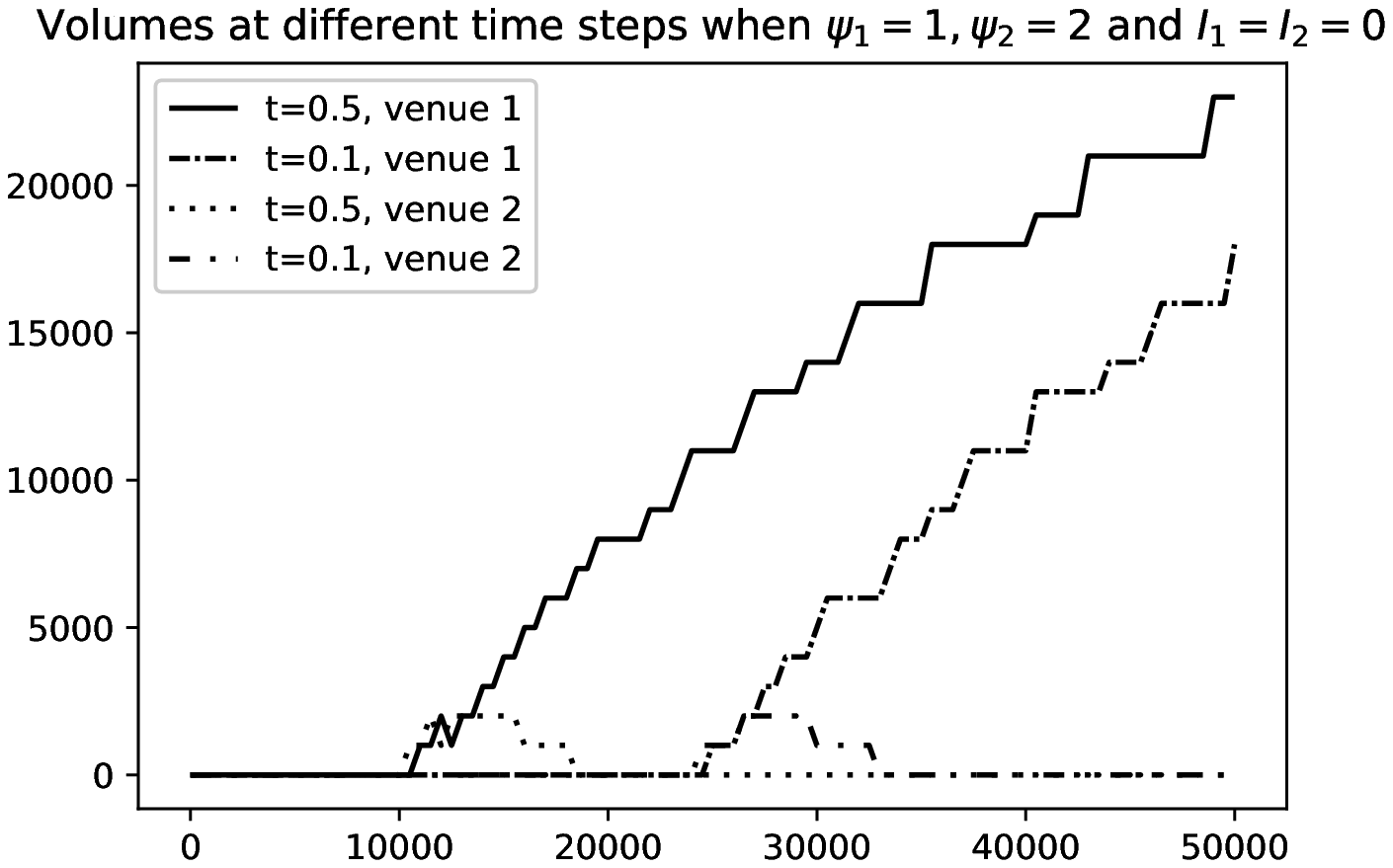}
       \vspace{-3mm}
       \caption{Volume strategy, $\psi^1 = \delta, \psi^2 = 2\delta,$ \protect\\ $ I^1 = I^2 = 0$.}\label{volumes_venue_sprdiff}
     \end{center}
  \end{minipage} 
\end{figure}
\vspace{-3mm}
If the imbalance is more favorable in the second venue, we see in Figures \ref{limits_venue_imbdiff} and \ref{volumes_venue_imbdiff} that the strategy is very different from the one in Figures \ref{limits_venue_imb} and \ref{volumes_venue_imb} where the two venues shared the same characteristics. As the second venue has a more favorable imbalance, the trader posts a higher volume in it. However, he posts a nonzero volume in the first venue, because of the overall better filling ratio. This contrasts with Figure \ref{volumes_venue_imb} where at some sufficiently high inventories, the trader stops sending orders to the first venue. Due to the trade-off between an overall higher filling ratio in the first venue and a more favorable imbalance in the second venue, the trader splits his liquidity consumption between the two venues. \\

The strategy on the limits in Figure \ref{limits_venue_imbdiff} also differs from the one with two identical venues in Figure \ref{limits_venue_imb}. For $t=0.5$ in Figure \ref{limits_venue_imb}, the trader posts in the first venue at the second best limit for $q\in[10000,13000]$, at the first best limit for $q\in[13000,20000]$ and at a new best limit for $q\in[20000,32000]$. In Figure \ref{limits_venue_imbdiff}, the trader posts in the first venue at the second best limit for $q\in[10000,12000]$, at the first best limit for $q\in[12000,18000]$ and at a new best limit for $q>18000$. Therefore, he posts at a more favorable limit in terms of filling rate in the first venue in order to compensate for the unfavorable imbalance compared to the second venue.

\begin{figure}[H]

\begin{minipage}[c]{.46\linewidth}
  \begin{center}
      \includegraphics[width=0.9\textwidth]{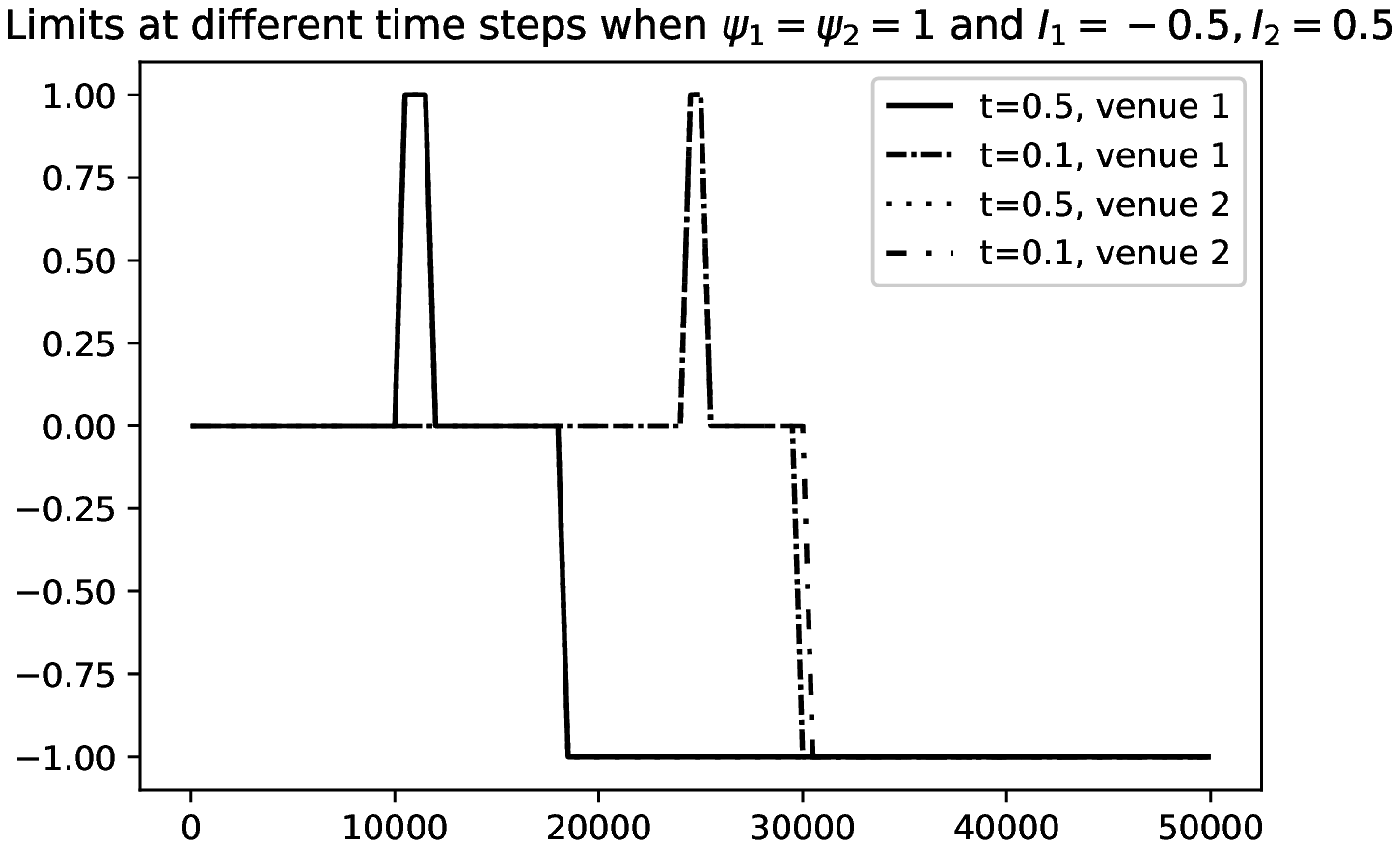}
      \vspace{-3mm}
      \caption{Limit order strategy, $\psi^1 =\psi^2 = \delta,$ \protect\\ $ I^1 =-0.5, I^2 = 0.5$.}\label{limits_venue_imbdiff}
    \end{center}
\end{minipage} \hfill
\begin{minipage}[c]{.46\linewidth}
   \begin{center}
       \includegraphics[width=0.9\textwidth]{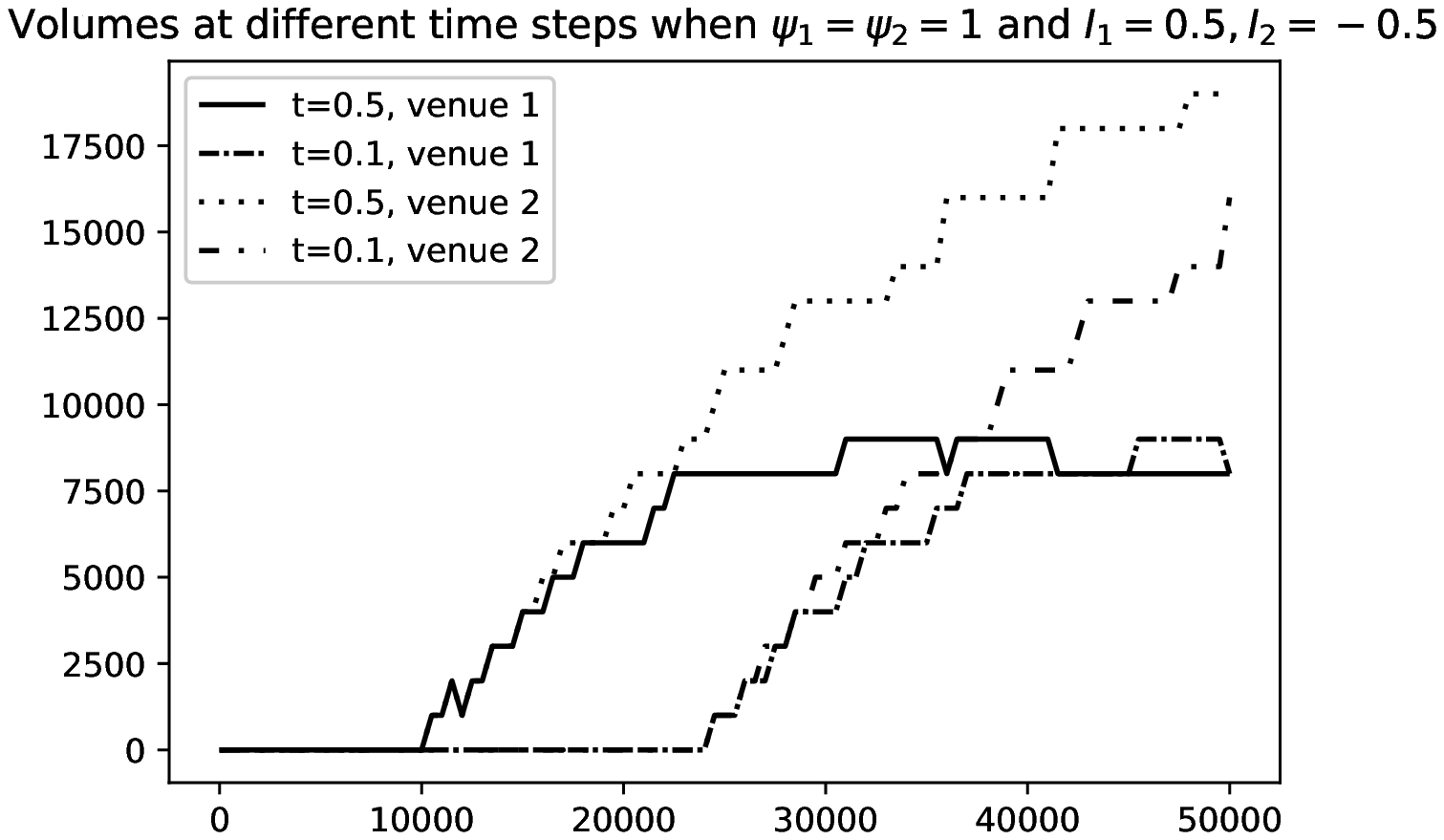}
       \vspace{-3mm}
       \caption{Volume strategy, $\psi^1 =\psi^2 = \delta,$ \protect\\ $ I^1 =-0.5, I^2 = 0.5$.}\label{volumes_venue_imbdiff}
     \end{center}
  \end{minipage} 
\end{figure}

Before moving to the analysis of the effectiveness of the Bayesian update of market parameters, we conclude with a comparison of the strategies obtained via neural networks optimization.

\subsubsection{Strategy: limit orders and volumes with neural networks}

We observe in Figures \ref{limits_venue1_diff_nets} and \ref{limits_venue2_diff_nets} that the strategy of the trader on the limits is in line with the one in Figures \ref{limits_venue1_diff} and \ref{limits_venue2_diff} up to the states where the optimal volume of the order equals $0$. 

\begin{figure}[H]

\begin{minipage}[c]{.48\linewidth}
  \begin{center}
      \includegraphics[width=0.88\textwidth]{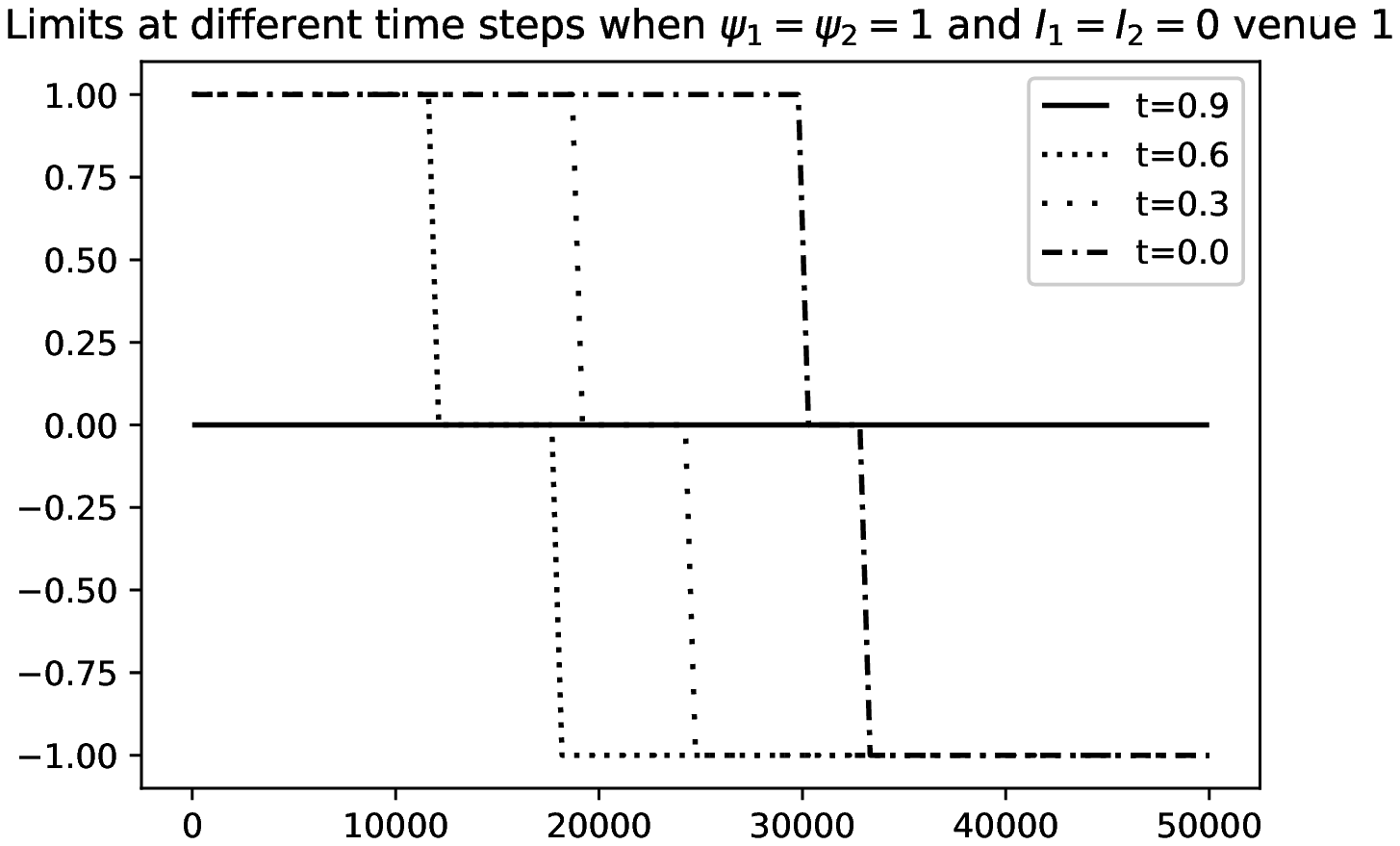}
      \vspace{-3mm}
      \caption{Limit order strategy in the first \protect\\ venue, $\psi^1 = \psi^2 = \delta, I^1 = I^2 = 0$ using neural networks.}
      \label{limits_venue1_diff_nets}
    \end{center}
\end{minipage} \hfill
\begin{minipage}[c]{.46\linewidth}
   \begin{center}
       \includegraphics[width=0.9\textwidth]{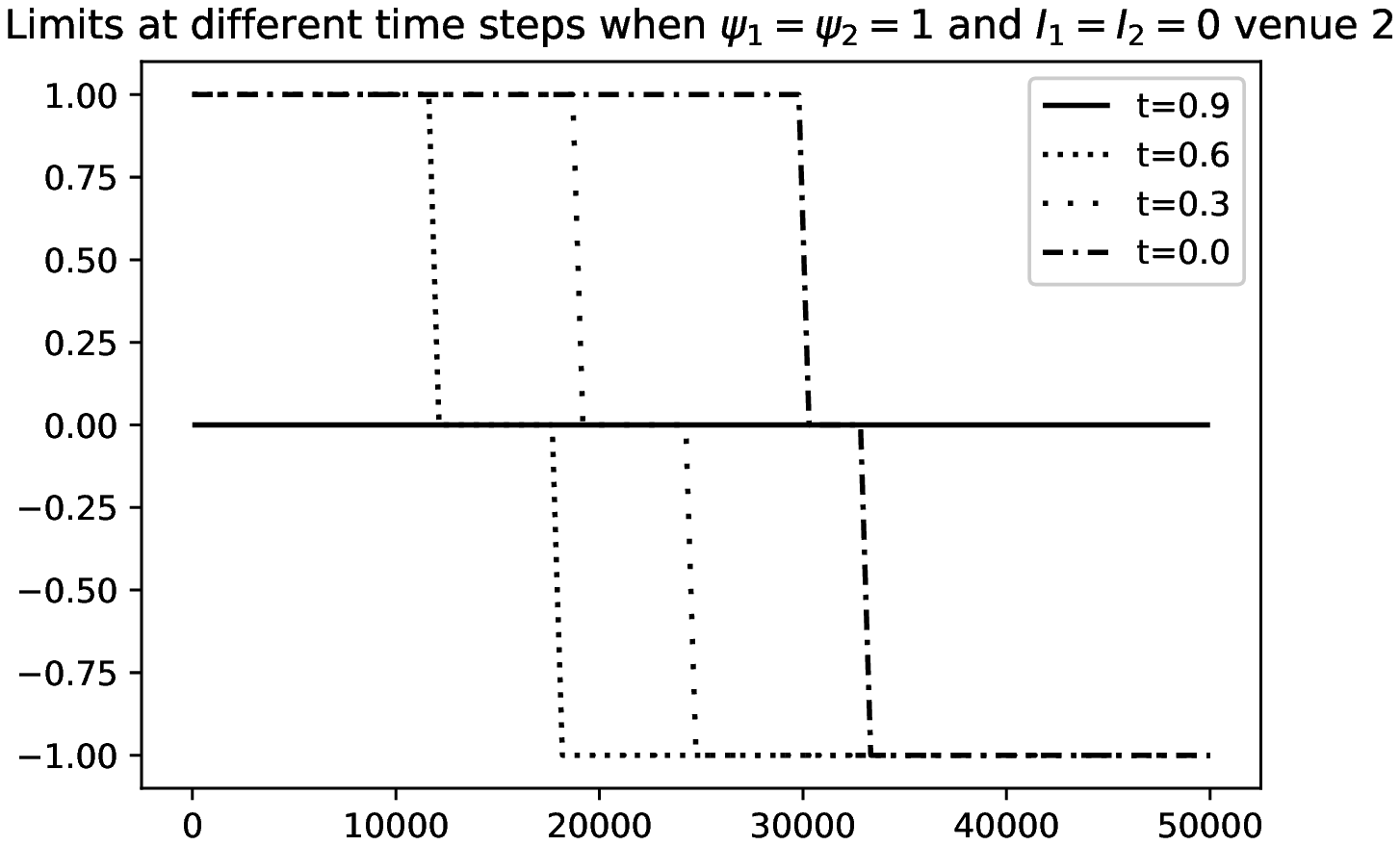}
       \vspace{-3mm}
       \caption{Limit order strategy in the second venue, $\psi^1 = \psi^2 = \delta, I^1 = I^2 = 0$ using neural networks.}
       \label{limits_venue2_diff_nets}
     \end{center}
  \end{minipage} 
\end{figure}

In Figures \ref{volumes_venue1_diff_nets} and \ref{volumes_venue2_diff_nets}, we see that the strategy of the trader on the posted volumes is well approximated and smoothed by neural networks.

\begin{figure}[H]

\begin{minipage}[c]{.46\linewidth}
  \begin{center}
      \includegraphics[width=0.9\textwidth]{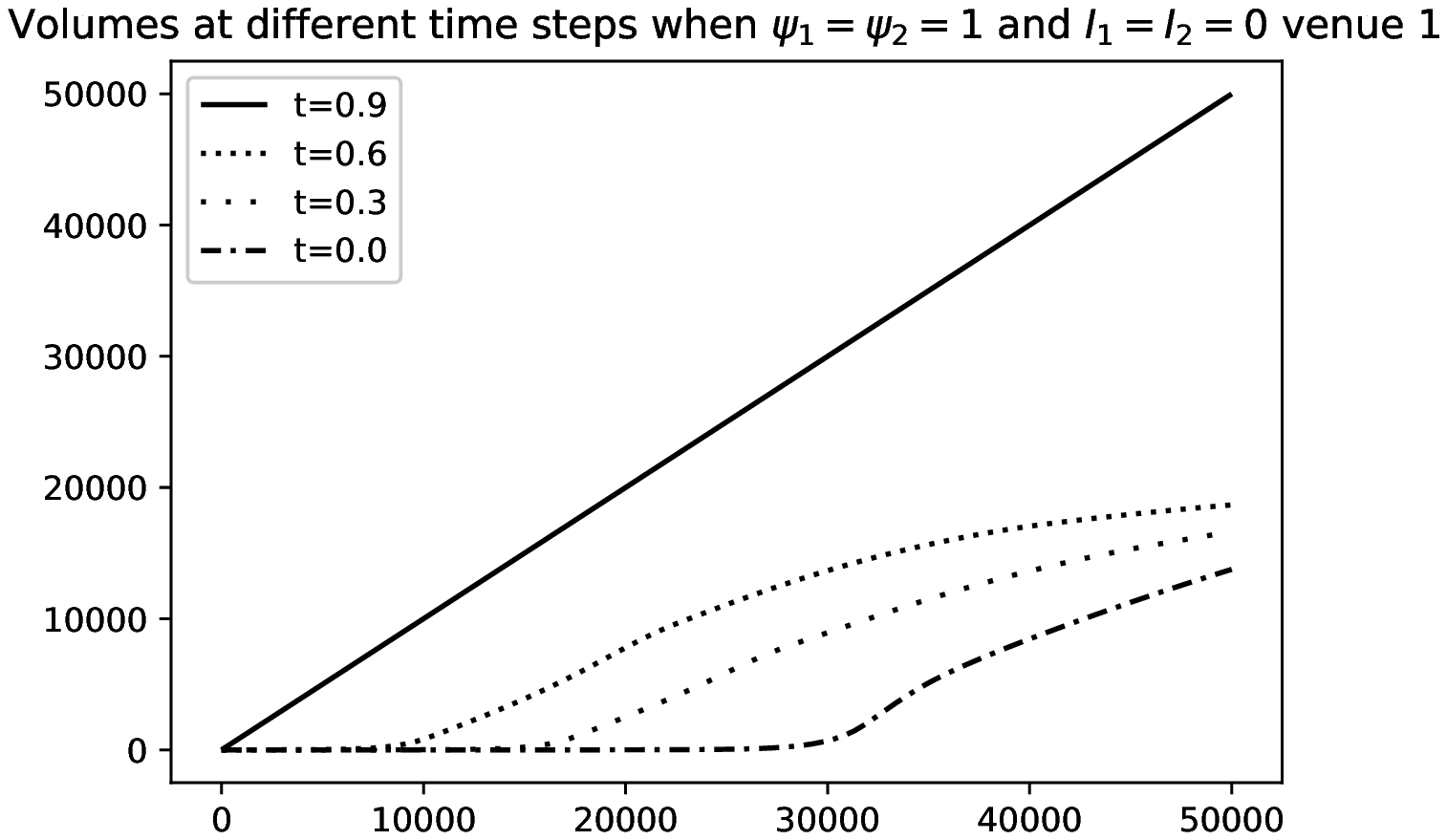}
      \vspace{-3mm}
      \caption{Volume posted in the first venue, $\psi^1 = \psi^2 = \delta, I^1 = I^2 = 0$ using neural networks.}
      \label{volumes_venue1_diff_nets}
    \end{center}
\end{minipage} \hfill
\begin{minipage}[c]{.46\linewidth}
   \begin{center}
       \includegraphics[width=0.9\textwidth]{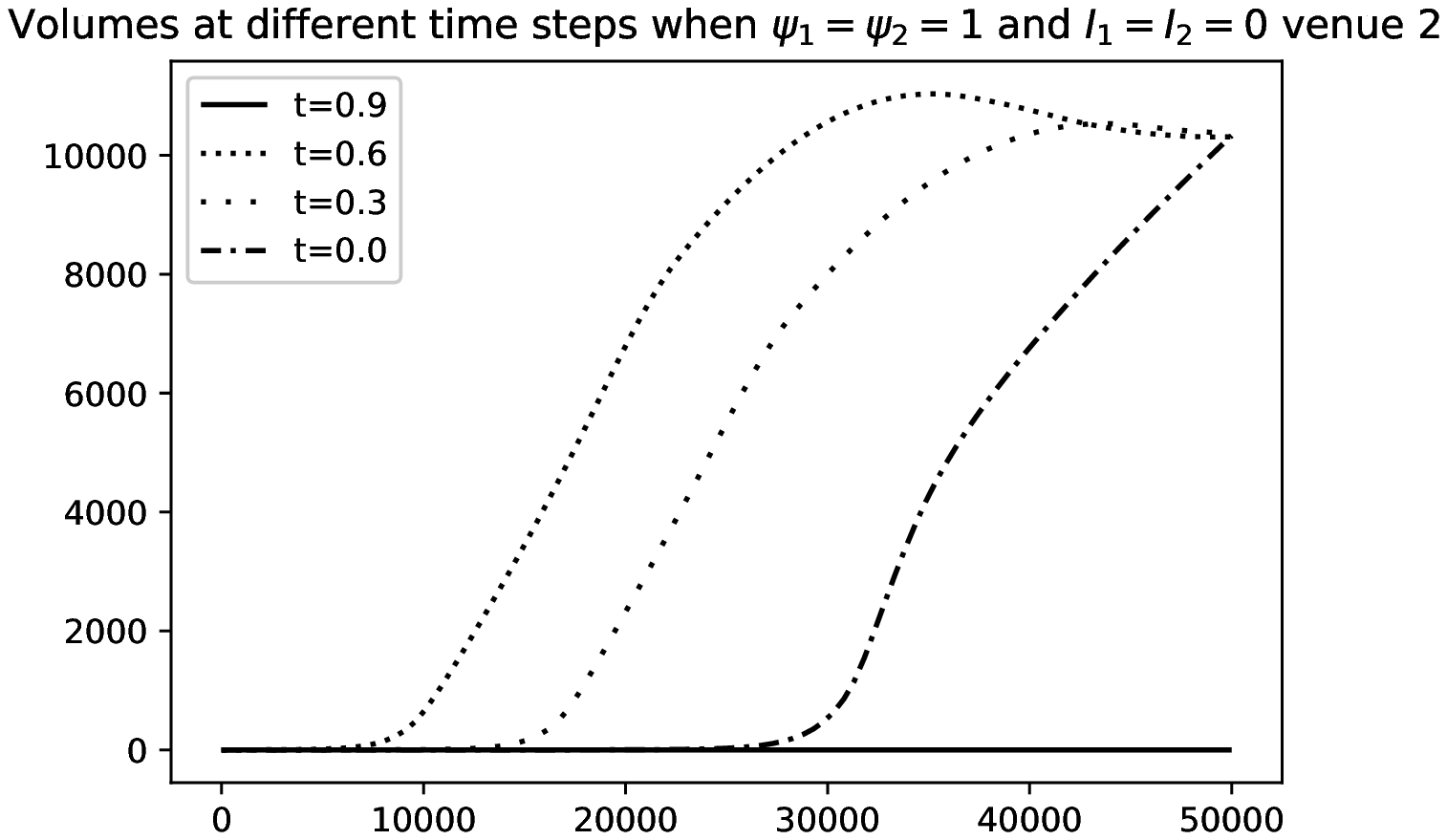}
       \vspace{-3mm}
       \caption{Volume posted in the second venue, $\psi^1 = \psi^2 = \delta, I^1 = I^2 = 0$ using neural networks.}
       \label{volumes_venue2_diff_nets}
     \end{center}
  \end{minipage} 
\end{figure}

In Figures \ref{limits_venue_sprdiff_nets} and \ref{volumes_venue_sprdiff_nets}, we see in the case of a higher spread in the second venue that, because of neural network parametrization of the strategy, the trader posts a nonzero volume in the second venue leaving the possibility to better explore the filling ratios. Results are in line with the ones in Figures \ref{limits_venue_sprdiff} and \ref{volumes_venue_sprdiff}: the trader posts the majority of his volume in the first venue because of a lower spread and a more favorable filling ratio.

\begin{figure}[H]

\begin{minipage}[c]{.49\linewidth}
  \begin{center}
      \includegraphics[width=0.94\textwidth]{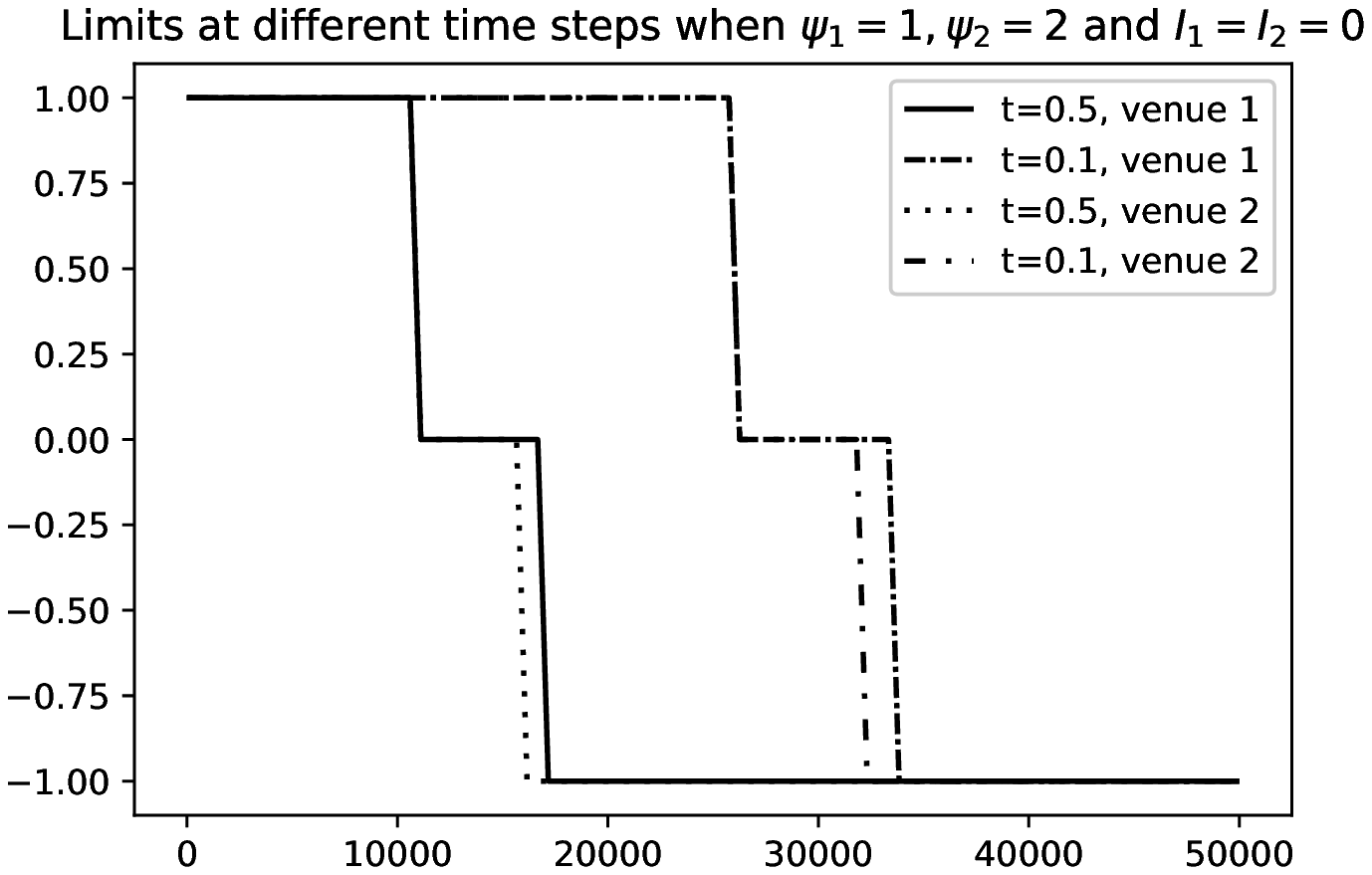}
      \vspace{-3mm}
      \caption{Limit order strategy, $\psi^1 =\delta, \psi^2 = 2\delta,$ \protect\\ $ I^1 = I^2 = 0$ using neural networks.}
      \label{limits_venue_sprdiff_nets}
    \end{center}
\end{minipage} \hfill
\begin{minipage}[c]{.46\linewidth}
   \begin{center}
       \includegraphics[width=0.97\textwidth]{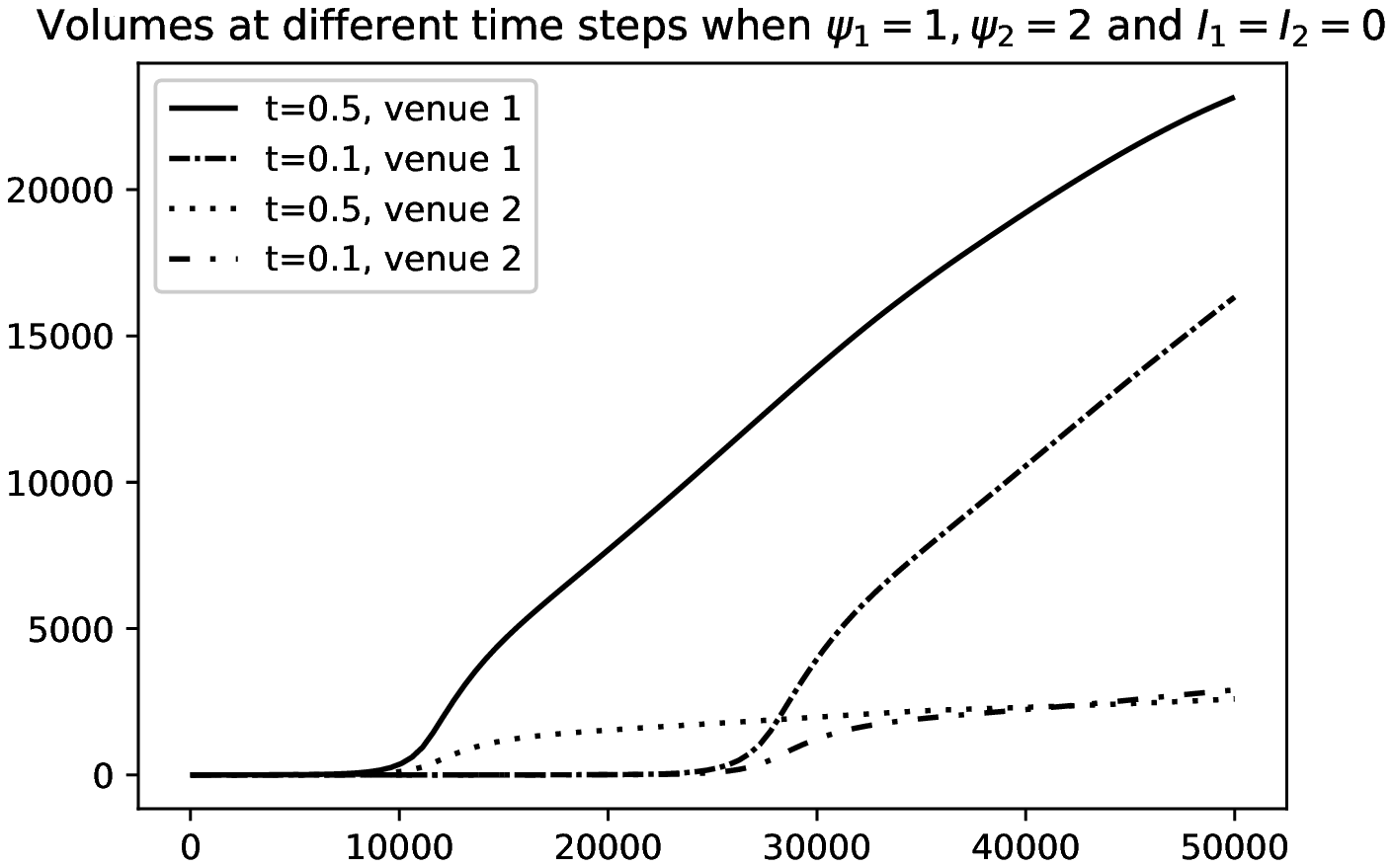}
       \vspace{-3mm}
       \caption{Volume strategy, $\psi^1 = \delta, \psi^2 = 2\delta,$ \protect\\ $ I^1 = I^2 = 0$ using neural networks.}
       \label{volumes_venue_sprdiff_nets}
     \end{center}
  \end{minipage} 
\end{figure}

Finally, we show in Figures \ref{limits_venue_imbdiff_nets} and \ref{volumes_venue_imbdiff_nets} a similar behavior compared to the finite difference schemes in Figures \ref{limits_venue_imbdiff} and \ref{volumes_venue_imbdiff}: the trader posts a higher volume in the second venue due to a more favorable imbalance, and keeps posting in the first venue due to an overall more favorable filling ratio. 

\begin{figure}[H]

\begin{minipage}[c]{.49\linewidth}
  \begin{center}
      \includegraphics[width=0.88\textwidth]{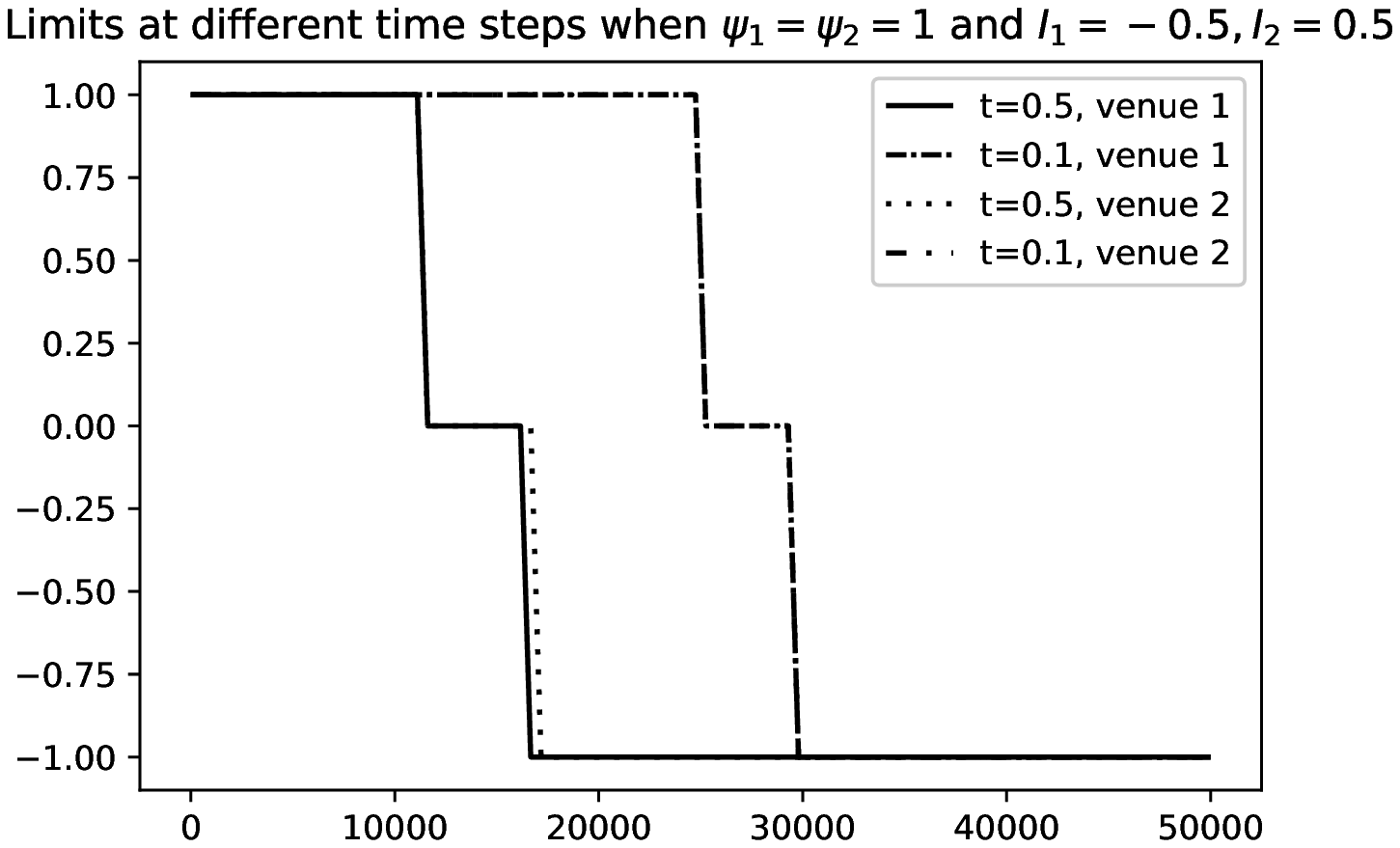}
      \vspace{-3mm}
      \caption{Limit order strategy, $\psi^1 =\psi^2 = \delta,$ \protect\\ $ I^1 =-0.5, I^2 = 0.5$ using neural networks.}
      \label{limits_venue_imbdiff_nets}
    \end{center}
\end{minipage} \hfill
\begin{minipage}[c]{.46\linewidth}
   \begin{center}
       \includegraphics[width=0.9\textwidth]{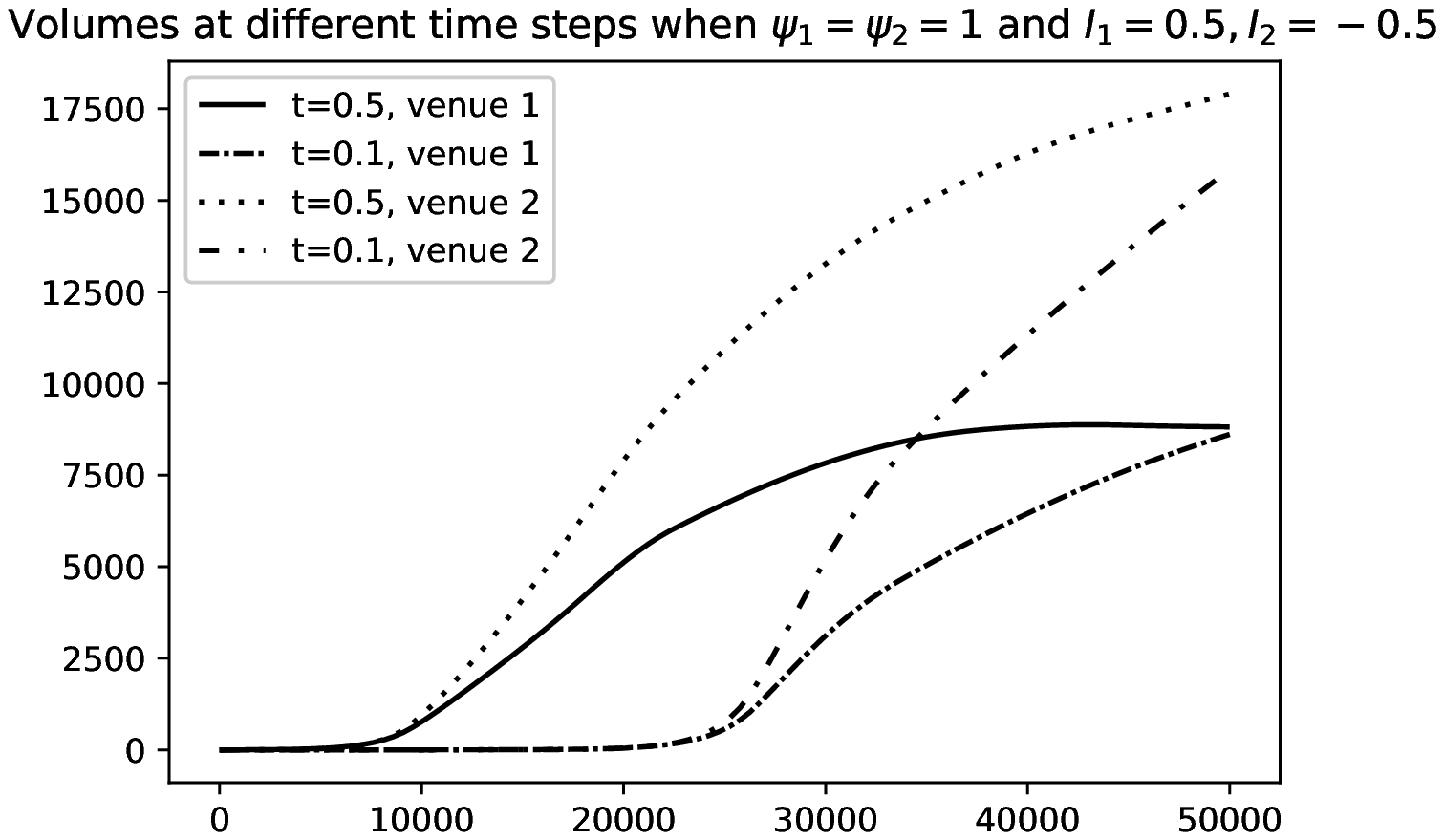}
       \vspace{-3mm}
       \caption{Volume strategy, $\psi^1 =\psi^2 = \delta,$ \protect\\ $ I^1 =-0.5, I^2 = 0.5$ using neural networks.}
       \label{volumes_venue_imbdiff_nets}
     \end{center}
  \end{minipage} 
\end{figure}

\subsection{Bayesian update}

In this section, we analyze the effectiveness of the Bayesian update framework through several execution slices. 

\subsubsection{Market simulation on a slice}

We first show an example of a market simulation of one slice and demonstrate the trading strategy through the slice, which are illustrated in Figure \ref{simulation}. \\

At $t=0.2$, the spread in both venues is equal to $\delta$, with an unfavorable imbalance in both venues. In that case, as two venues share the same characteristics, and the inventory is sufficiently close to the optimal for the next step, so the trader sends the same quantity to both venues, which is close to zero. \\

When $t=0.5$, the first venue has an unfavorable imbalance, and the second venue has a higher spread. In this configuration, the trader sends a higher volume in the first venue, in order to get a better filling rate due to a lower spread. \\

Finally, at $t=0.7$, the first venue has a higher spread and a more favorable imbalance compared to the second venue. This leads to a higher volume in the first venue at the second best limit and a lower volume in the second venue at first best limit. The favorable imbalance in the first venue indicates a higher probability of execution for an order at a higher limit, because the price may move in this direction. Therefore, even if the spread is equal to two ticks, the trader posts in this venue in order to be executed at a more favorable price. As the spread in the second venue is lower, but the imbalance is less favorable, he posts at the first best limit to benefit from the trade-off between execution and profit through collecting the spread. 
\vspace{-5mm}
\begin{figure}[H]
\begin{center}
      \includegraphics[width=0.7\textwidth]{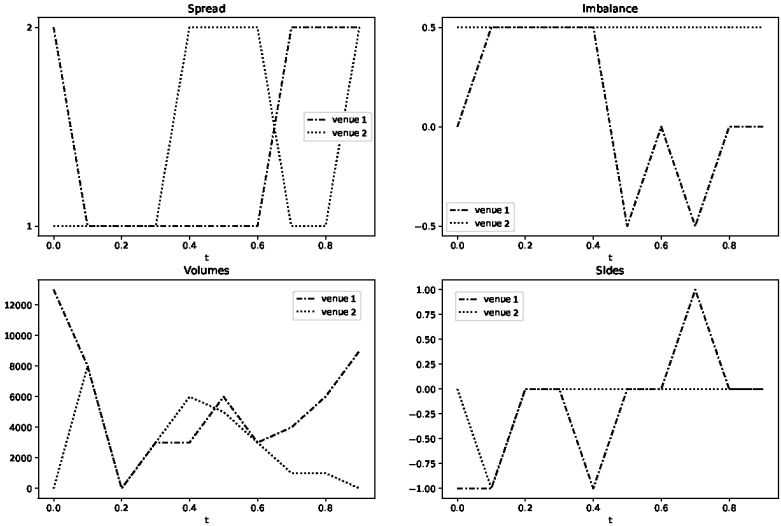}
      \vspace{-10mm}
      \caption{Market simulation: spreads (upper left), imbalances (upper right), volumes (lower left) and limits (lower right) in both venues.}\label{simulation}
\end{center}
\end{figure}
\vspace{-5mm}
The corresponding execution trajectory is shown in Figure \ref{simq}, where we can see the typical Implementation Shortfall execution shape, coming from the pre-computed trading curve $q^\star$. 
\vspace{-5mm}
\begin{figure}[H]
\begin{center}
  \includegraphics[width=0.44\textwidth]{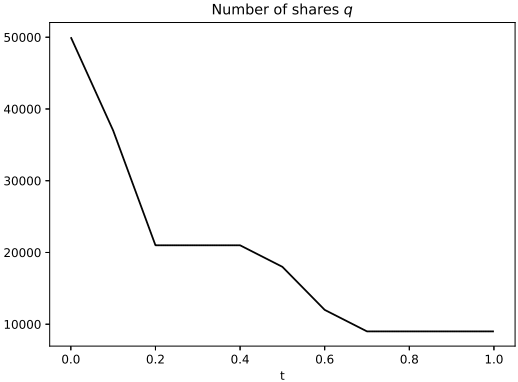}
  \vspace{-5mm}
  \caption{Evolution of the inventory of the trader on a slice of execution.}\label{simq}
\end{center}
\end{figure}
\vspace{-5mm}
\subsubsection{Update of the execution proportion}

We show how the trader updates the market parameters through observations and trading. The update of the execution proportion is quite fast, as it can be seen in Figures \ref{bayesianrho0} and \ref{bayesianrho1} the good estimation can be achieved after completing 1-2 slices. In this example we started from the correct prior for the second venue and the inaccurate one for the first:
\vspace{-3mm}
\begin{align*}
  &\rho^{1} = \left[\begin{matrix}
0.1 & 0.9   
\end{matrix}\right],
  &\rho^{2} = \left[\begin{matrix}
0.1 & 0.9    
\end{matrix}\right].
\end{align*}
\vspace{-6mm}
\begin{figure}[H]
\begin{minipage}[c]{.48\linewidth}
\hspace{-10mm}
      \includegraphics[width=1.15\textwidth]{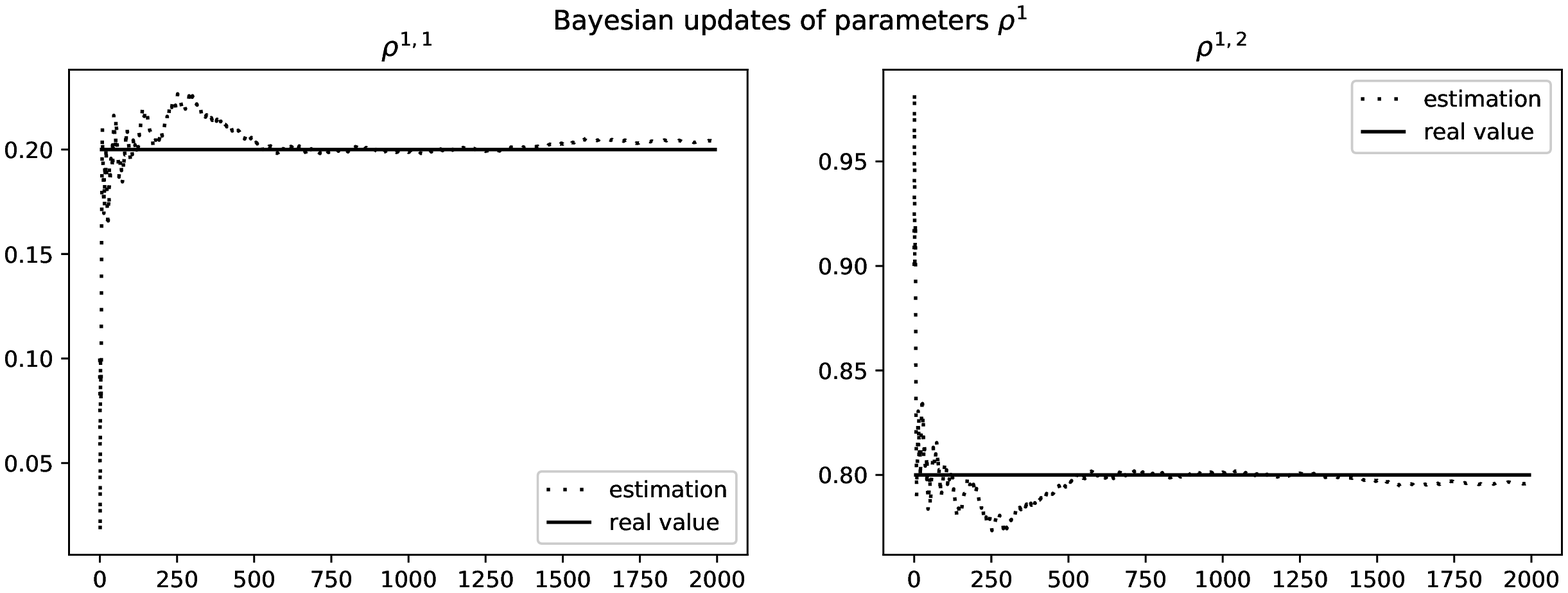}
      \vspace{-3mm}
      \caption{Bayesian update of the executed proportion in the first venue.}\label{bayesianrho0}
\end{minipage} \hfill
\begin{minipage}[c]{.48\linewidth}
\hspace{-10mm}
      \includegraphics[width=1.15\textwidth]{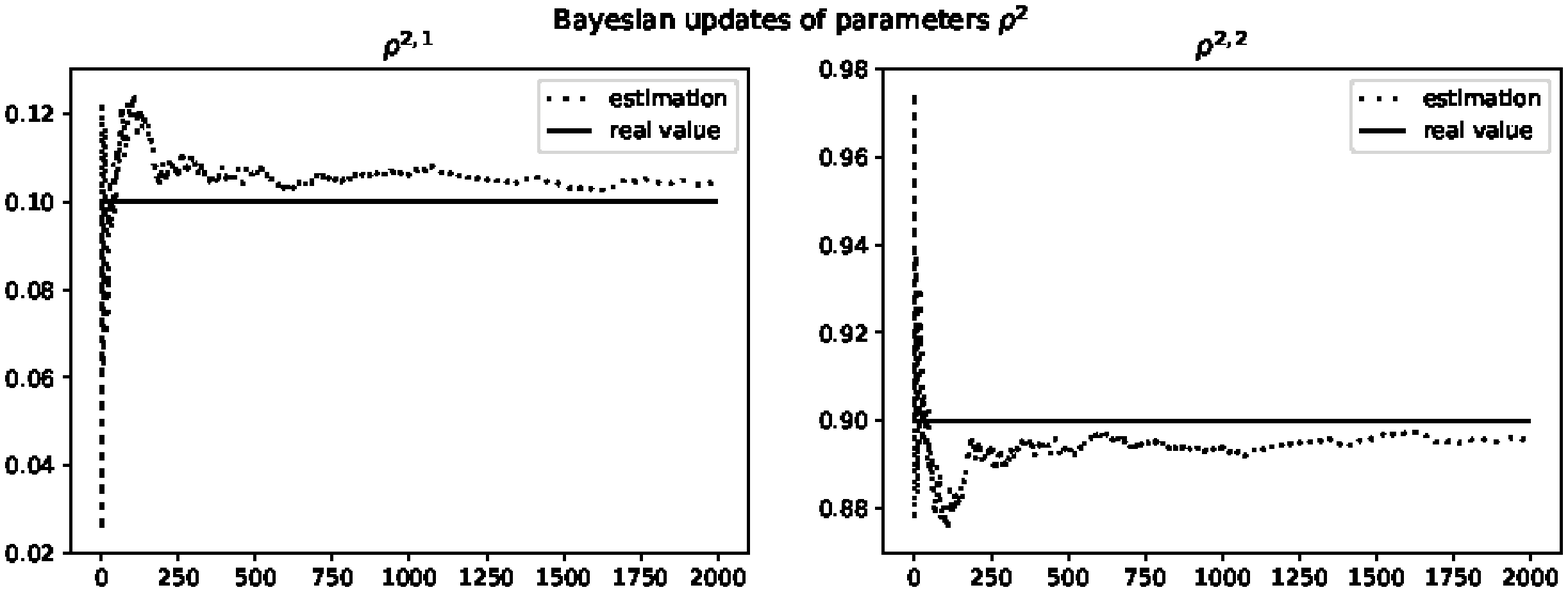}
      \vspace{-3mm}
      \caption{Bayesian update of the executed proportion in the second venue.}\label{bayesianrho1}
  \end{minipage} 
\end{figure}

\subsubsection{Update of the imbalance and the spread dynamics}
We plot the convergence of the estimated transition matrices $r^{1,\psi},r^{2,\psi}$ in Figures \ref{bayesianrpsi} and \ref{bayesianrpsi1}. We observe quite fast convergence to a good approximation of the spread dynamics parameters. The prior values are respectively:
\begin{align*}
  &r^{1, \psi} = \left[\begin{matrix}
-5 & 5   \\
 5 & -5    
\end{matrix}\right],
  &r^{2, \psi} &= \left[\begin{matrix}
-5 & 5   \\
 5 & -5    
\end{matrix}\right].
\end{align*}
\vspace{-8mm}
\begin{figure}[H]
\begin{minipage}[c]{.48\linewidth}
\hspace{-12mm}
      \includegraphics[width=1.23\textwidth]{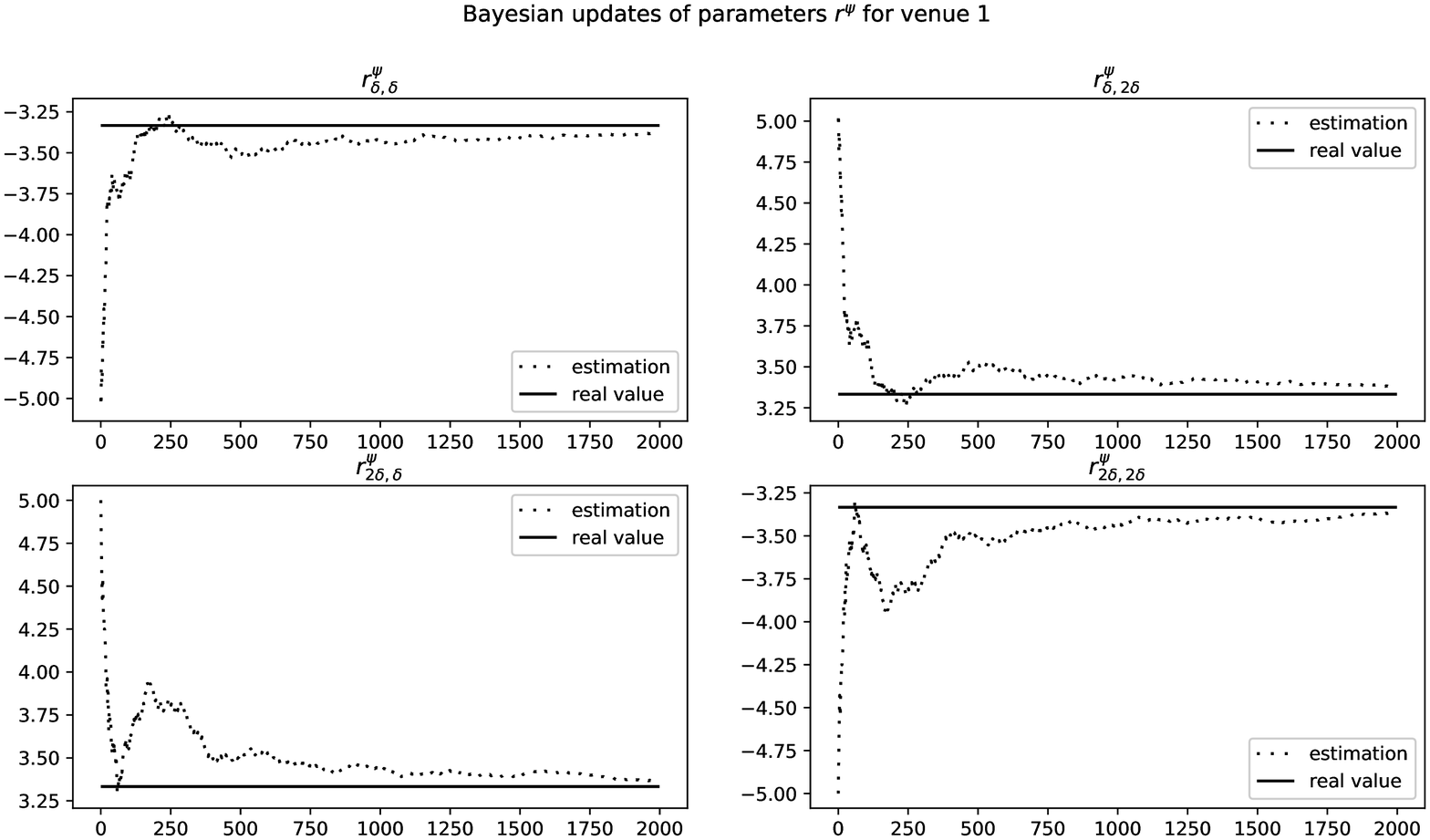}
      \vspace{-10mm}
      \caption{Bayesian update of the transition matrix $r^{1,\psi}$.}\label{bayesianrpsi}
\end{minipage} \hfill
\begin{minipage}[c]{.48\linewidth}
\hspace{-10mm}
 \includegraphics[width=1.23\textwidth]{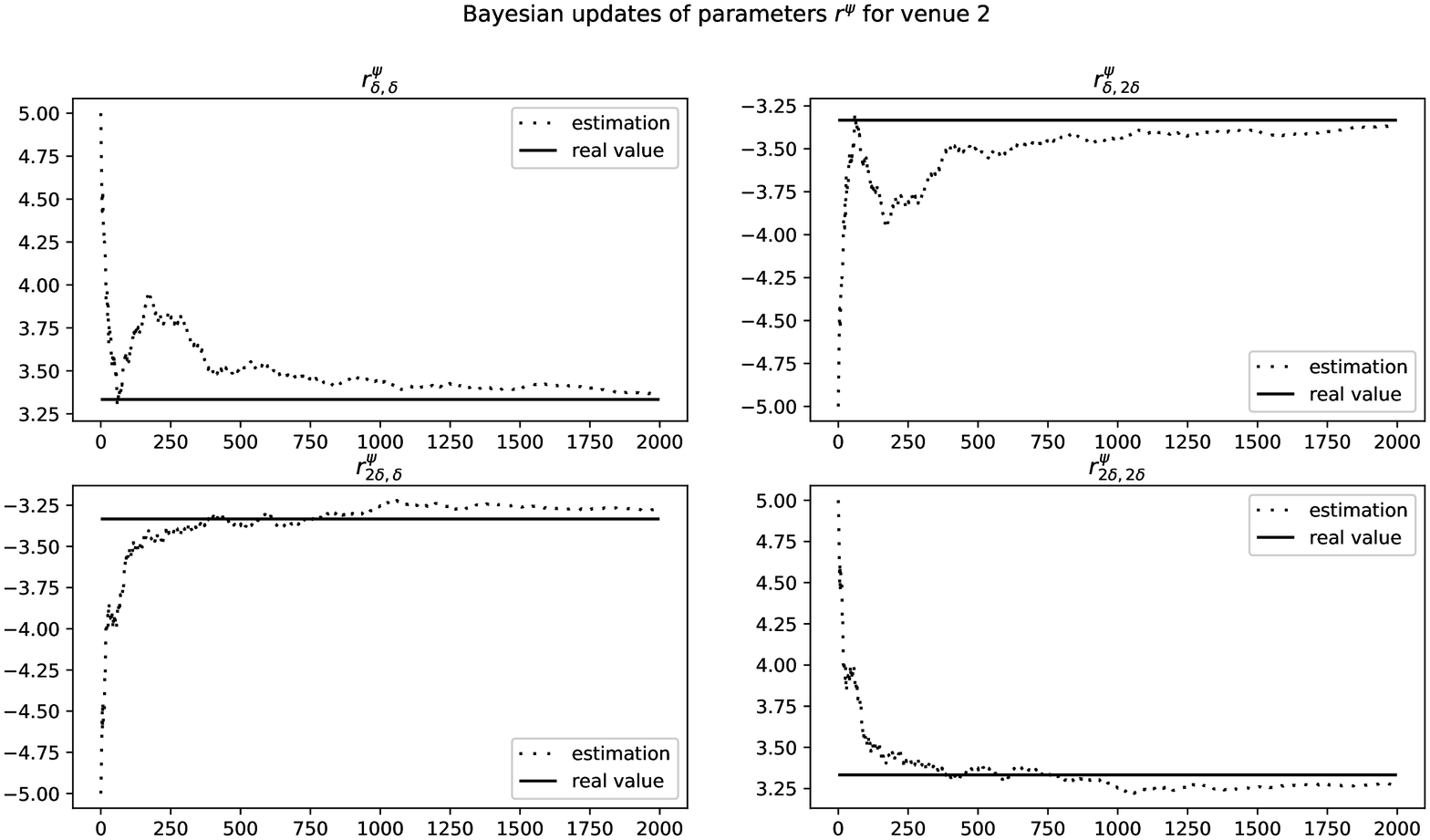}
      \vspace{-10mm}
      \caption{Bayesian update of the transition matrix $r^{2,\psi}$.}\label{bayesianrpsi1}
  \end{minipage} 
\end{figure}

We perform the same study for the transition matrices $r^{1,I}, r^{2,I}$ of the imbalance processes through the observed one in Figures \ref{bayesiani0} and \ref{bayesiani1}. We see that we need just a couple of slices to have a quite good approximation and only a dozen of slices (less for more granular slices) to achieve the right estimation. \\

\begin{figure}[H]

\begin{center}
      \includegraphics[width=0.85\textwidth]{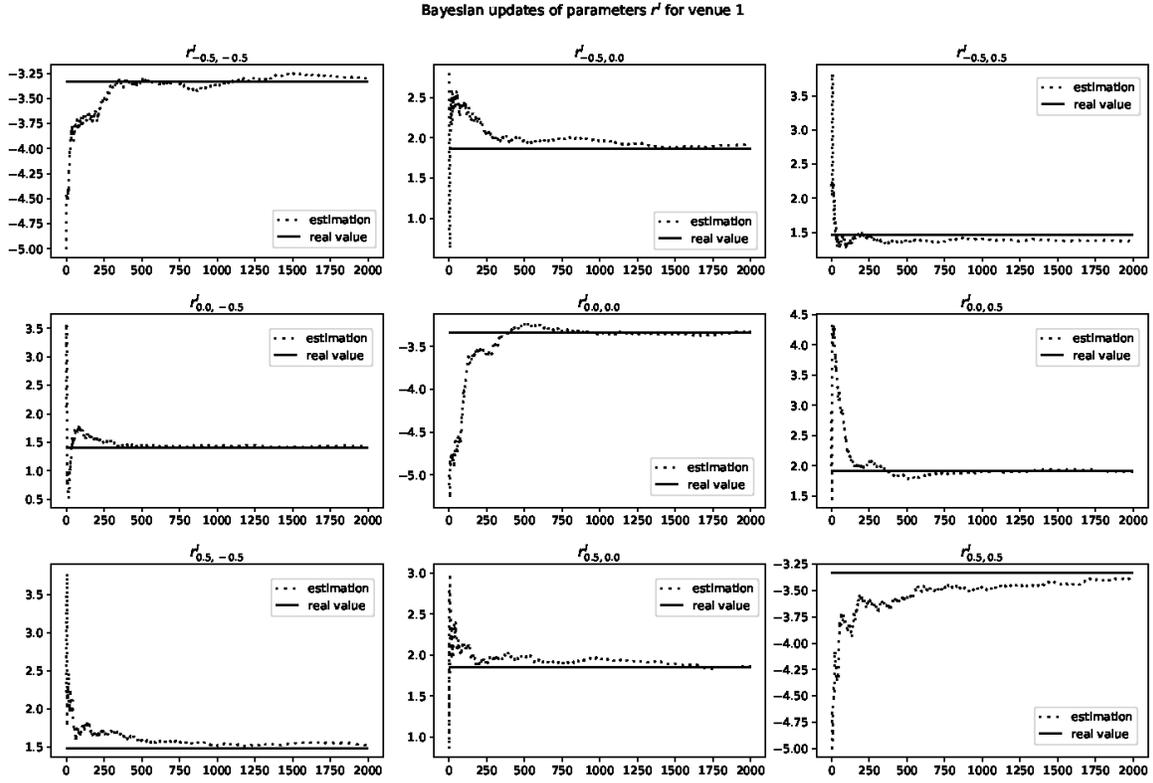}
      \vspace{-5mm}
      \caption{Bayesian update of the transition matrix $r^{1,I}$.}\label{bayesiani0}
\end{center}
\vspace{-10mm}
\end{figure}

\begin{figure}[H]
\begin{center}
      \includegraphics[width=0.85\textwidth]{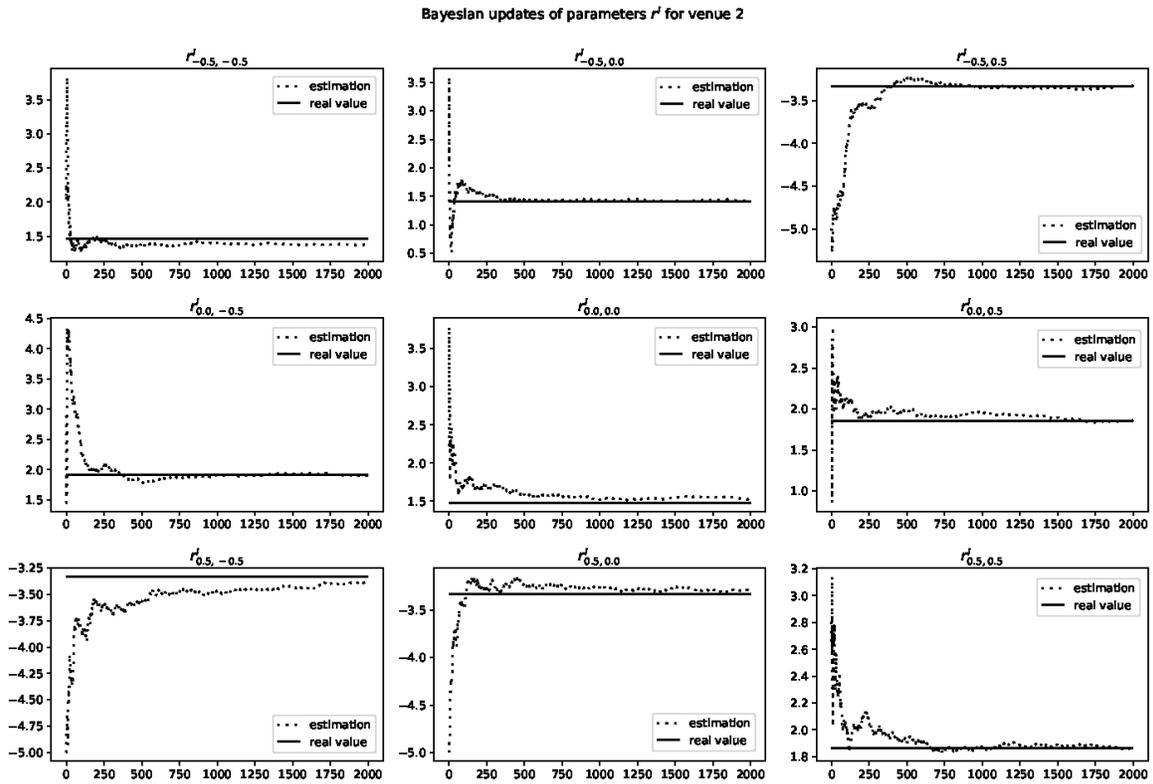}
      \vspace{-5mm}
      \caption{Bayesian update of the transition matrix $r^{2,I}$.}\label{bayesiani1}
\end{center}
\vspace{-10mm}
\end{figure}
In the examples in Figures \ref{bayesiani0} and \ref{bayesiani1} we started from the following prior parameters: 
\begin{align*}
  &r^{1, I} = \left[\begin{matrix}
-5    & 2.8    & 2.2    \\
 2.2& -5. & 2.8\\
 2.2& 2.8& -5    
\end{matrix}\right],
  &r^{2, I} &= \left[\begin{matrix}
-5    & 2.8    & 2.2    \\
 2.2& -5. & 2.8\\
 2.2& 2.8& -5    
\end{matrix}\right].
\end{align*}

\subsubsection{Update of the long term drift of the asset}

As we observe the increments of the price process $S_t$ continuously, it is easy to converge toward a real market drift, the example is in Figure \ref{bayesianmu}, we find $\mu^{\text{true}}=-0.5$, starting from a prior of $\mu=0.1$. It took $20$ slices to find a real value even if in the considered example we supposed to be sure in our prior estimation $\nu=0.02$, which appeared to be incorrect. 
\vspace{-4mm}
\begin{figure}[H]
\begin{center}
      \includegraphics[width=0.35\textwidth]{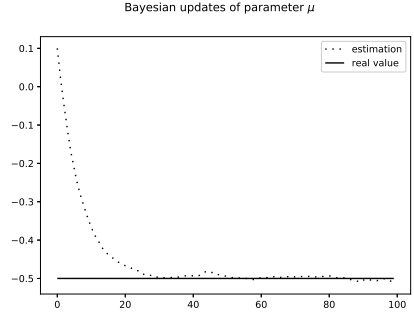}
      \vspace{-5mm}
      \caption{Bayesian update of the drift of the asset.}
      \label{bayesianmu}
\end{center}
\end{figure}
\vspace{-5mm}
\subsubsection{Update of the intensity of limit orders}
The hardest parameter to update quickly is obviously the intensity of filling which depends on states of both venues. In our numerical setting we have $32$ possible states, so during one slice of 10 time steps we have no possibility to even visit all the states. The results of convergence of the parameters $\lambda$ can be found in Figures \ref{bayesianrlambda0} and \ref{bayesianrlambda1}. We see that full convergence requires a lot of observations, however we should keep in mind that to have a strategy close to the optimal one we do not necessitate an excessive precision. \\

In this example, we started from the priors same for both venues, whereas the real parameters are different. The priors are:

\begin{align*}
&\lambda^{1}_{\delta, \delta} = \lambda^{2}_{\delta, \delta} = \left[\begin{matrix}
    5.35   & 6.52   & 7.11   \\
    2.75   & 3.4    & 3.79   \\
    1.5    & 1.86 & 2.1 
\end{matrix}\right],
&\lambda^{1}_{\delta, 2\delta} = \lambda^{2}_{\delta, 2\delta} =
\left[\begin{matrix}
    8.28   & 10.03   & 10.9  \\
    4.38   & 5.35   & 5.9   \\
    2.5    & 3.05 & 3.4    
\end{matrix}\right],\\
&\lambda^{1}_{2\delta, \delta} = \lambda^{2}_{2\delta, \delta} =
\left[\begin{matrix}
    1.81 & 2.27 & 2.5    \\
    0.78 & 1.04 & 1.19\\
    0.29 & 0.43 & 0.53
\end{matrix}\right],
&\lambda^{1}_{2\delta, 2\delta} = \lambda^{2}_{2\delta, 2\delta} =
\left[\begin{matrix}
    2.96 & 3.65 & 4.    \\
    1.42 & 1.81 & 2.04\\
    0.68 & 0.9 & 1.04
\end{matrix}\right].
\end{align*}

\begin{figure}[H]
\begin{center}
      \includegraphics[width=0.9\textwidth]{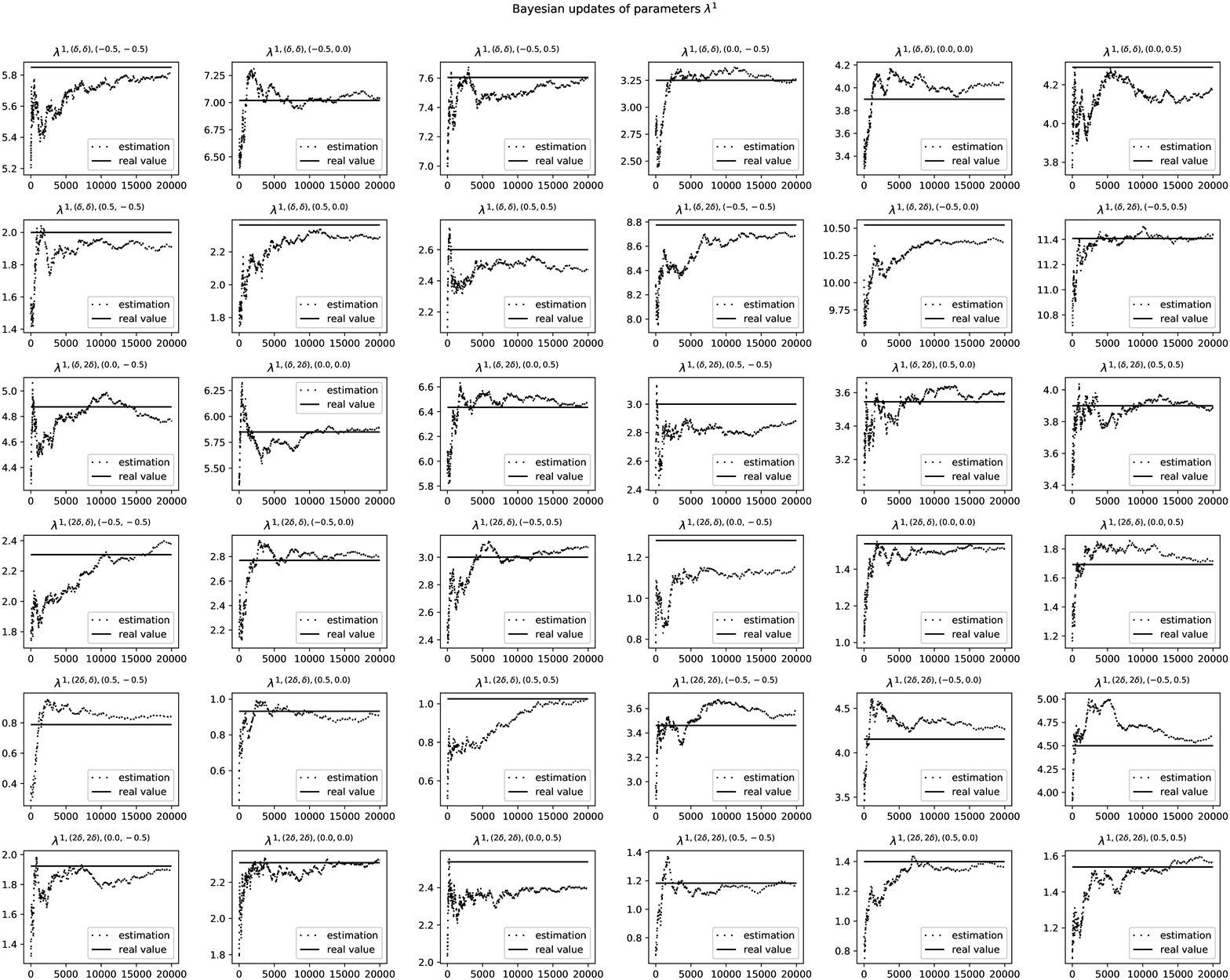}
      \vspace{-3mm}
      \caption{Bayesian update of the intensity of limit orders in the first venue.}
      \label{bayesianrlambda0}
\end{center}

\end{figure}

\vspace{-10mm}
\begin{figure}[H]
\begin{center}
      \includegraphics[width=0.9\textwidth]{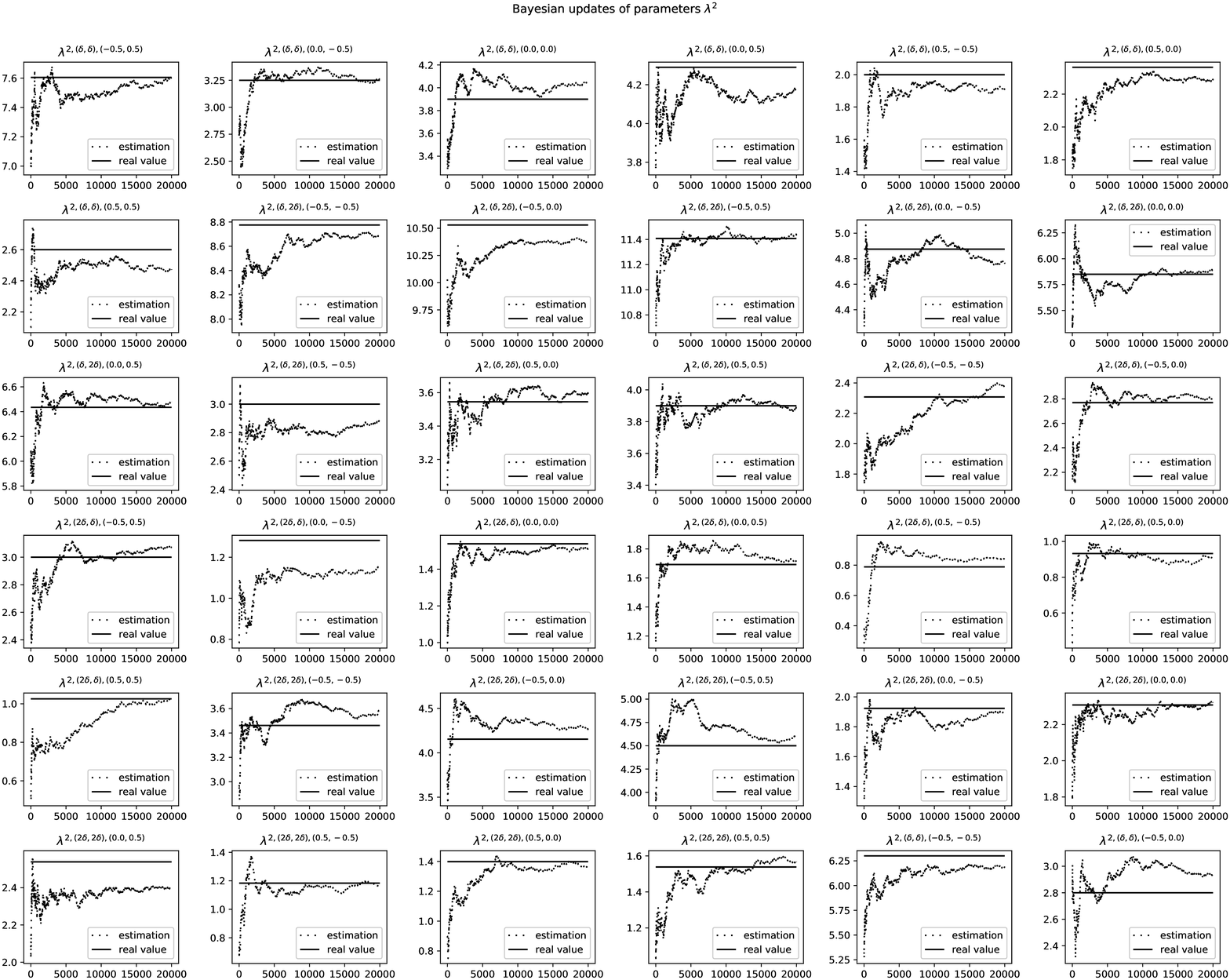}
      \vspace{-3mm}
      \caption{Bayesian update of the intensity of limit orders in the second venue.}\label{bayesianrlambda1}
\end{center}
\end{figure}

\pagebreak

\begingroup
\setcounter{section}{0}
\renewcommand\thesection{Appendix \Alph{section}}
\section{Proof of Theorem \ref{HJBQVI}}\label{Section Proof HJBQVI}
It can be show with the dynamic programming principle that the HJBQVI \eqref{HJBQVI} does not depend on the cash variable $x$. We set $\big(\tilde{q},\tilde{\psi},\tilde{I}\big)\in \mathcal{D}= \mathcal{Q}\times \mathcal{K}$ and $(t_i,S_i)\in [0,T)\times \mathbb{R}$ such that
\begin{align*}
  & t_i\rightarrow_{i\rightarrow +\infty} \hat{t},\\
  & S_i \rightarrow_{i\rightarrow +\infty} \hat{S},\\
  & v(t_i,\tilde{q},S_i,\tilde{\psi},\tilde{I})\rightarrow_{i\rightarrow +\infty} v_{\star} (\hat{t},\tilde{q},\hat{S},\tilde{\psi},\tilde{I}).
\end{align*}
We begin with $\hat{t}=T$. By taking $\ell^{n}=0$ for all $n\in\{1,\dots,N\}$ we get 
\begin{align*}
  v(t_i,\tilde{q},S_i,\tilde{\psi},\tilde{I})\geq \mathbb{E}_{t_i,\tilde{q},S_i,\tilde{\psi},\tilde{I}}\Big[Q_T S_T - \int_0^T g(q_t^{\star} - q_t) dt \Big].
\end{align*}
By dominated convergence, we get $v_{\star}(T,\tilde{q},\hat{S},\tilde{\psi},\tilde{I})\geq \tilde{q}\hat{S}$. \\

Assume now that $\hat{t}<T$ and that the minimum in the HJBQVI is given by the first term. We take $\phi: [0,T)\times \mathbb{R} \times \mathcal{D}\rightarrow \mathbb{R}$ be $C^1$ in time, $C^2$ in $\hat{S}$ and such that $0=\min_{[0,T]\times\mathbb{R}\times \mathcal{D}}(v_{\star}-\phi)=(v_{\star}-\phi)(\hat{t},\tilde{q},\hat{S},\tilde{\psi},\tilde{I})$. If there exists $\eta>0$ such that \vspace{-2mm}
\begin{align*}
2\eta < &\partial_t \phi(\hat{t},\tilde{q},\hat{S},\tilde{\psi},\tilde{I})\! -\! g(q-q_t^{\star}) \!+\! \mu\partial_S \phi + \frac{1}{2}\sigma^2 \partial_{SS}\phi \!+\! \sum_{\mathbf{k}\in \mathcal{K}} \!r_{(\tilde{\psi},\tilde{I}),(\mathbf{k}^\psi,\mathbf{k}^I)}\big(\phi(\hat{t},\tilde{q},\hat{S},\mathbf{k}^\psi,\mathbf{k}^I)\!-\!\phi(\hat t,\tilde q,\hat S,\tilde \psi,\tilde I)\big) \\
  & + \sup_{p\in Q_\psi, \ell\in \mathcal{A}} \sum_{n=1}^N \lambda^{n}(\tilde{\psi},\tilde{I},p^{n},\ell)\mathbb{E}\Big[\epsilon^{n}\ell^{n}(\hat{S}+\frac{\tilde{\psi}^{n}}{2}+p^{n}\delta^n)+\phi\big(\hat{t},\tilde{q}-\ell^{n}\epsilon^{n},\hat{S},\tilde{\psi},\tilde{I}\big) - \phi(\hat{t},\tilde{q},\hat{S},\tilde{\psi},\tilde{I})\Big], 
\end{align*}
\vspace{-2mm}
we should have
\vspace{-2mm}
\begin{align*}
  0 \leq & \partial_t \phi(t,\tilde{q},S,\tilde{\psi},\tilde{I}) - g(q-q_t^{\star}) + \mu\partial_S \phi + \frac{1}{2} \partial_{SS}\phi\! +\! \sum_{\mathbf{k}\in \mathcal{K}} \!r_{(\tilde{\psi},\tilde{I}),(\mathbf{k}^\psi,\mathbf{k}^I)}\big(\phi(t,\tilde{q},S,\mathbf{k}^\psi,\mathbf{k}^I)\!-\!\phi(t,\tilde q, S,\tilde \psi,\tilde I)\big) \\
  & + \sup_{p\in Q_\psi, \ell\in \mathcal{A}}\sum_{n=1}^N\lambda^{n}(\tilde{\psi},\tilde{I},p^{n},\ell)\mathbb{E}\Big[\epsilon^{n}\ell^{n}(S+\frac{\tilde{\psi}^{n}}{2}+p^{n}\delta^n)+\phi\big(t,\tilde{q}-\ell^{n}\epsilon^{n},S,\tilde{\psi},\tilde{I}\big) - \phi(t,\tilde{q},S,\tilde{\psi},\tilde{I})\Big], 
\end{align*}
\vspace{-2mm}
for all $(t,S)\in B= \Big((\hat{t}-r,\hat{t}+r) \cap [0,T)\Big)\times \Big(\hat{S}-r,\hat{S}+r\Big)$ for a given $r\in (0,T-\hat{t})$. We can assume without loss of generality that $B$ contains the sequences $(t_i,S_i)_i$ and, by taking $\eta$ arbitrarily small
\begin{align*}
  \phi(t,\tilde{q},S,\tilde{\psi},\tilde{I})+\eta \leq v_\star(t,\tilde{q},S,\tilde{\psi},\tilde{I}) \leq v(t,\tilde{q},S,\tilde{\psi},\tilde{I})
\end{align*}
on the boundary of $B$, denoted by $\partial_p B$. Without loss of generality we can also assume that
\begin{align*}
  \phi(t,q,S,\psi,I)+\eta \leq v_\star(t,q,S,\psi,I) \leq v(t,q,S,\psi,I),
\end{align*}
for $(t,q,S,\psi,I)\in \tilde{B}$ where
\begin{align*}
  \tilde{B}=\Big\{ & (t,q,S,\psi,I): (t,S)\in B, q\in\{\tilde{q}-\min_n \epsilon^{n},\tilde{q},\tilde{q}+\min_n \epsilon^{n}\}, \\
  & \psi \in \prod_{n=1}^N \{\tilde{\psi}^n-\delta^n, \tilde{\psi}^n,\tilde{\psi}^n+ \delta^n\}, I\in \prod_{n=1}^N \{\tilde{I}_{-1}^n, \tilde{I}^n,\tilde{I}_{+1}^n\}, (q,\psi,I)\neq (\tilde{q},\tilde{\psi},\tilde{I})\Big\}. 
\end{align*}
We introduce the set 
\begin{align*}
  B_{\mathcal{D}}=\big\{(t,\tilde q,S,\tilde \psi,\tilde I): (t,S)\in B \Big\},
\end{align*}
and denote by $\tau_i$ the first exit time of $(t,q_t,S_t,\psi_t,I_t)_{t\geq t_i}$ from $B_{\mathcal{D}}$, with $q_{t_i}=\tilde{q},S_{t_i}=\hat{S},\psi_{t_i}=\tilde{\psi},I_{t_i}=\tilde{I}$, and the processes are controlled by the optimal controls $(\ell^{n},p^{n})_{n\in\{1,\dots,N\}}\in\mathcal{A}\times Q_\psi$. By Ito's formula, we get
\begin{align*}
 \phi(\tau_i,q_{\tau_i},&S_{\tau_i},\psi_{\tau_i},I_{\tau_i}) =  \phi(t_i,q_{t_i},S_{t_i},\psi_{t_i},I_{t_i}) + \int_{t_i}^{\tau_i} \partial_t\phi(s,q_s,S_s,\psi_s,I_s) + \mu\partial_S \phi \\
  & + \frac{1}{2}\sigma^2 \partial_{SS}\phi +\sum_{\mathbf{k}\in \mathcal{K}} r_{(\psi_s,I_s),(\mathbf{k}^\psi,\mathbf{k}^I)}\big(\phi(s,q_s,S_s,\mathbf{k}^\psi,\mathbf{k}^I)-\phi(s,q_s, S_s, \psi_s, I_s)\big) \\
  & + \sum_{n=1}^N\lambda^{n}(\psi_s,I_s,p_s^{n},\ell_s)\mathbb{E}\Big[\phi\big(s,q_s-\ell_s^{n}\epsilon_s^{n},S_s,\psi_s,I_s\big) - \phi(s,q_s,S_s,\psi_s,I_s)\Big] ds + M(\tau_i,t_i),
\end{align*}
where $M(\tau_i,t_i)$ is a martingale. This can be rewritten as
\begin{align*}
  \phi(\tau_i,q_{\tau_i},&S_{\tau_i},\psi_{\tau_i},I_{\tau_i}) =  \phi(t_i,q_{t_i},S_{t_i},\psi_{t_i},I_{t_i}) + \int_{t_i}^{\tau_i} \partial_t\phi(s,q_s,S_s,\psi_s,I_s) + \mu\partial_S \phi - g(q_s-q^\star(s)) \\
  & + \frac{1}{2}\sigma^2 \partial_{SS}\phi + \sum_{\mathbf{k}\in \mathcal{K}} r_{(\psi_s,I_s),(\mathbf{k}^\psi,\mathbf{k}^I)}\big(\phi(s,q_s,S_s,\mathbf{k}^\psi,\mathbf{k}^I)-\phi(s,q_s, S_s, \psi_s, I_s)\big) \\
  & + \sum_{n=1}^N \! \lambda^{n}(\psi_s,I_s,p_s^{n},\ell_s)\mathbb{E}\Big[\epsilon_s^{n}\ell_s^{n}(S_s\!+\!\frac{\psi_s^{n}}{2}+p_s^{n}\delta^n)\!+\! \phi\big(s,q_s\!-\!\ell_s^{n}\epsilon_s^{n},S_s,\psi_s,I_s\big) \!- \!\phi(s,q_s,S_s,\psi_s,I_s)\Big] ds \\
  & + M(\tau_i,t_i) - \sum_{n=1}^N \int_{t_i}^{\tau_i} \lambda^{n}(\psi_s,I_s,p_s^{n},\ell_s) \mathbb{E}\big[\epsilon_s^{n}\ell_s^{n}(S_s+\frac{\psi_s^{n}}{2}+p_s^{n}\delta^n)\big] + g(q_s - q^\star(s)) ds . 
\end{align*}
We derive
\begin{align*}
   \phi(\tau_i,q_{\tau_i},S_{\tau_i},\psi_{\tau_i},I_{\tau_i}) \geq & \phi(t_i,q_{t_i},S_{t_i},\psi_{t_i},I_{t_i}) \\
  & + M(\tau_i,t_i) - \sum_{n=1}^N \int_{t_i}^{\tau_i} \lambda^{n}(\psi_s,I_s,p_s^n,\ell_s) \mathbb{E}\big[\epsilon_s^{n}\ell_s^{n}(S_s+\frac{\psi_s^{n}}{2}+p_s^{n}\delta)\big] + g(q_s - q^\star(s)) ds.
\end{align*}
As the martingale term vanishes with the expectation, we get
\begin{align*}
   \phi(t_i,q_{t_i},S_{t_i},\psi_{t_i},I_{t_i}) \leq & \mathbb{E}\Big[\phi(\tau_i,q_{\tau_i},S_{\tau_i},\psi_{\tau_i},I_{\tau_i}) \\
  & + \sum_{n=1}^N \int_{t_i}^{\tau_i} \lambda^{n}(\psi_s,I_s,p_s^{n},\ell_s) \mathbb{E}\big[\epsilon_s^{n}\ell_s^{n}(S_s+\frac{\psi_s^{n}}{2}+p_s^{n}\delta)\big] - g(q_s - q^\star(s)) ds\Big].
\end{align*}
and thus
\begin{align*}
  \phi(t_i,q_{t_i},S_{t_i},\psi_{t_i},I_{t_i}) \leq & -\eta + \mathbb{E}\Big[v(\tau_i,q_{\tau_i},S_{\tau_i},\psi_{\tau_i},I_{\tau_i}) \\
  & + \sum_{n=1}^N \int_{t_i}^{\tau_i} \lambda^{n}(\psi_s,I_s,p_s^{n},\ell_s) \mathbb{E}\big[\epsilon_s^{n}\ell_s^{n}(S_s+\frac{\psi_s^{n}}{2}+p_s^{n}\delta)\big] - g(q_s - q^\star(s)) ds\Big].  
\end{align*}
For $i$ sufficiently large, we deduce 
\begin{align*}
  v(t_i,\tilde{q},S_{t_i},\tilde{\psi},\tilde{I}) \leq & -\frac{\eta}{2} + \mathbb{E}\Big[v(\tau_i,q_{\tau_i},S_{\tau_i},\psi_{\tau_i},I_{\tau_i}) \\
  & + \sum_{n=1}^N \int_{t_i}^{\tau_i} \lambda^{n}(\psi_s,I_s,p_s^{n},\ell_s) \mathbb{E}\big[\epsilon_s^{n}\ell_s^{n}(S_s+\frac{\psi_s^{n}}{2}+p_s^{n}\delta)\big] - g(q_s - q^\star(s)) ds\Big],  
\end{align*}
which contradicts the dynamic programming principle. In conclusion, we necessarily have
\begin{align*}
  0\geq & \partial_t v(\hat{t},\tilde{q},\hat{S},\tilde{\psi},\tilde{I}) \!-\! g(q - q_t^{\star}) \!+ \!\mu\partial_S v \!+\! \frac{1}{2}\sigma^2 \partial_{SS}v \!+\! \sum_{\mathbf{k}\in \mathcal{K}}\! r_{(\tilde{\psi},\tilde{I}),(\mathbf{k}^\psi,\mathbf{k}^I)}\big(v(\hat{t},\tilde q,\hat S,\mathbf{k}^\psi,\mathbf{k}^I)\!-\!v(\hat t,\tilde q, \hat S, \tilde \psi, \tilde I)\big) \\
  & + \sup_{p\in Q_\psi, \ell\in \mathcal{A}}\sum_{n=1}^N \lambda^{n}(\tilde{\psi},\tilde{I},p^{n},\ell)\mathbb{E}\Big[\epsilon^{n}\ell^{n}(\hat{S}+\frac{\tilde{\psi}^{n}}{2}+p^{n}\delta^n)+v\big(\hat{t},\tilde{q}-\ell^{n}\epsilon^{n},\hat{S},\tilde{\psi},\tilde{I}\big) - v(\hat{t},\tilde{q},\hat{S},\tilde{\psi},\tilde{I})\Big]. 
\end{align*}
The second part of the HJBQVI being straightforward, we prove that $v$ is a viscosity supersolution of the HJBQVI on $[0,T)\times \mathbb{R}\times \mathcal{D}$. The proof for the subsolution is identical, except that we need to prove
\begin{align*}
  \sum_{n=1}^N \sup_{m^n\in [0,\overline{m}]}m^n(S-\frac{\psi^n}{2}) + v\big(t,q-m^n,S,\psi,I\big) - v(t,q,S,\psi,I) \geq 0,
\end{align*}
which is direct by choosing the constant controls $\overline{m}^n = 0$ for all $n\in\{1,\dots,N\}$. \\

For the proof of the uniqueness, we recall the definition of subjet and superjet.

\begin{definition}
Let $v:[0,T)\times\mathbb{R}\times\mathcal{D}\rightarrow \mathbb{R}$ be l.s.c (resp u.s.c) with respect to $(\hat{t},\hat{S})$. For $(\hat{t},\tilde{q},\hat{S},\tilde{\psi},\tilde{I})\in [0,T)\times\mathbb{R}\times\mathcal{D}$ we say that $(y,p,A)\in\mathbb{R}^3$ is in the subjet $\mathcal{P}^- v(\hat{t},\tilde{q},\hat{S},\tilde{\psi},\tilde{I})$ (resp. the superjet $\mathcal{P}^+v(\hat{t},\tilde{q},\hat{S},\tilde{\psi},\tilde{I})$ if 
\begin{align*}
  v(t,\tilde{q},S,\tilde{\psi},\tilde{I})\geq v(\hat{t},\tilde{q},\hat{S},\tilde{\psi},\tilde{I}) + y(t-\hat{t}) + p(S-\hat{S}) + \frac{1}{2}A(S-\hat{S})^2 + o(|t-\hat{t}| + |S-\hat{S}|^2),
\end{align*}
(resp. $v(t,\tilde{q},S,\tilde{\psi},\tilde{I})\geq v(\hat{t},\tilde{q},\hat{S},\tilde{\psi},\tilde{I}) + y(t-\hat{t}) + p(S-\hat{S}) + \frac{1}{2}A(S-\hat{S})^2 + o(|t-\hat{t}| + |S-\hat{S}|^2)$), for all $(t,S)$ such that $(t,\tilde q,S,\tilde \psi, \tilde I)\in [0,T)\times\mathbb{R}\times\mathcal{D}$. We also define $\mathcal{\overline{P}}^- (\hat{t},\tilde{q},\hat{S},\tilde{\psi},\tilde{I})$ as the set of points $(y,p,A)\in\mathbb{R}^3$ such that there exists a sequence $(t_I,\tilde{q},S_i,\tilde{\psi},\tilde{I},y_i,p_i,A_i)$ satisfying 
\begin{align*}
  (t_i,\tilde{q},S_i,\tilde{\psi},\tilde{I},y_i,p_i,A_i)\rightarrow_{i\rightarrow +\infty} (\hat t,\tilde{q},\hat S,\tilde{\psi},\tilde{I},y,p,A).
\end{align*}
The set $\mathcal{\overline{P}}^+ (\hat{t},\tilde{q},\hat{S},\tilde{\psi},\tilde{I})$ is defined similarly. 
\end{definition}
We now introduce an analogous of the Ishii's lemma, whose proof can be found in \cite{bouchard2007introduction}. 
\begin{lemma}\label{viscolemma}
A l.s.c (resp u.s.c) function $v$ is a supersolution (resp. subsolution) of the HJBQVI on $[0,T)\times\mathbb{R}\times \mathcal{D}$ if and only if for all $(\hat{t},\tilde{q},\hat{S},\tilde{\psi},\tilde{I})\in [0,T)\times\mathbb{R}\times\mathcal{D}$, and all $(\hat{y},\hat{p},\hat{A})\in \mathcal{\overline{P}}^-(\hat{t},\tilde{q},\hat{S},\tilde{\psi},\tilde{I})$ (resp. $\mathcal{\overline{P}}^+(\hat{t},\tilde{q},\hat{S},\tilde{\psi},\tilde{I})$), we have
\begin{align*}
  0\leq & \min \Bigg\{-\hat y + g(\tilde q - q^\star(\hat{t})) - \mu\hat p - \frac{1}{2}\sigma^2 \hat A  - \sum_{\mathbf{k}\in \mathcal{K}} r_{(\tilde{\psi},\tilde{I}),(\mathbf{k}^\psi,\mathbf{k}^I)}\big(v(\hat{t},\tilde q,\hat S,\mathbf{k}^\psi,\mathbf{k}^I)-v(\hat t,\tilde q, \hat S, \tilde \psi, \tilde I)\big) \\
  & - \sup_{p\in Q_\psi, \ell\in \mathcal{A}}\sum_{n=1}^N \lambda^{n}(\tilde \psi,\tilde I,p^{n},\ell)\mathbb{E}\Big[\epsilon^{n}\ell^{n}(\hat S+\frac{\tilde \psi^{n}}{2}+p^{n}\delta^n)+v(\hat t,\tilde q-\ell^{n}\epsilon^{n},\hat S,\tilde \psi,\tilde I) \\
  & - v(\hat t,\tilde q,\hat S,\tilde \psi,\tilde I)\Big] ; \quad \sum_{n=1}^N v(\hat t,\tilde q,\hat S,\tilde \psi,\tilde I)-\sup_{m^n\in [0,\overline{m}]}m^n(\hat S-\frac{\tilde \psi^n}{2}) + v\big(\hat t,\tilde q-m^n,\hat S,\tilde \psi,\tilde I\big)
  \Bigg\},  
\end{align*}
(resp. $\leq 0)$.
\end{lemma}
We now prove the following comparison principle: 
\begin{proposition}
Let $u$ (resp. $v$) be a l.s.c supersolution (resp. u.s.c subsolution) with polynomial growth of the HJBQVI on $[0,T)\times\mathbb{R}\times \mathcal{D}$. If $u\geq v$ on $\{T\}\times\mathbb{R}\times \mathcal{D}$, then $u\geq v$ on $[0,T)\times\mathbb{R}\times \mathcal{D}$.
\end{proposition}
\begin{proof}
For $\rho>0$ we introduce the following change of variables:
\begin{align*}
  \tilde{u}(t,q,S,\psi,I)=e^{\rho t}u(t,q,S,\psi,I), \quad \tilde{v}(t,q,S,\psi,I)=e^{\rho t}v(t,q,S,\psi,I).
\end{align*}
Then, $\tilde u$ and $\tilde v$ are respectively supersolution and subsolution of the following equation:
\begin{align*}
  0= & \min \Bigg\{-\partial_t w(t,q,S,\psi,I) + \rho w(t,q,S,\psi,I) + g(q-q_t^{\star}) - \mu\partial_S w - \frac{1}{2}\sigma^2 \partial_{SS}w \\
  & - \sum_{\mathbf{k}\in \mathcal{K}} r_{(\psi,I),(\mathbf{k}^\psi,\mathbf{k}^I)}\big(w(t, q, S,\mathbf{k}^\psi,\mathbf{k}^I)-w( t, q, S, \psi, I)\big) \\
  & - \sup_{p\in Q_\psi, \ell\in \mathcal{A}}\sum_{n=1}^N\lambda^{n}(\psi,I,p^{n},\ell)\mathbb{E}\Big[\epsilon^{n}\ell^{n}e^{\rho t}(S+\frac{\psi^{n}}{2}+p^{n}\delta^n)+w\big(t,q-\ell^{n}\epsilon^{n},S,\psi,I\big) \\
  & - w(t,q,S,\psi,I)\Big] ; \quad \sum_{n=1}^N w(t,q,S,\psi,I)-\sup_{m^n\in [0,\overline{m}]}m^n e^{\rho t}(S-\frac{\psi^n}{2}) + w\big(t,q-m^n,S,\psi,I\big)
  \Bigg\},  
\end{align*}
on $[0,T)\times \mathbb{R} \times \mathcal{D}$, with $\tilde{u}\geq \tilde{v}$ on $\{T\}\times\mathbb{R}\times\mathcal{D}$. In order to prove the proposition, we only have to show that $\tilde{u}\geq\tilde{v}$ on $[0,T)\times\mathbb{R}\times\mathcal{D}$. Assume that the minimum is given by the first term and that $\sup_{[0,T)\times\mathbb{R}\times\mathcal{D}}\tilde{v}-\tilde{u} > 0$. We fix $p\in \mathbb{N}^\star$ such that 
\begin{align*}
  \lim_{\|S\|_2\rightarrow +\infty} \sup_{\substack{t\in [0,T]\\ (q,\psi,I)\in \mathcal{D}}} \frac{|\tilde{u}(t,q,S,\psi,I)| + |\tilde{v}(t,q,S,\psi,I)|}{1+ \|S\|_2^{2p}}= 0. 
\end{align*}
Then, there exists $(\hat{t},\tilde{q},\hat{S},\tilde{\psi},\tilde{I})\in [0,T]\times\mathbb{R}\times\mathcal{D}$ such that 
\begin{align*}
  0 &< \tilde{v}(\hat t,\tilde q,\hat S,\tilde\psi,\tilde I) - \tilde{u}(\hat t,\tilde q,\hat S,\tilde\psi,\tilde I) - \phi(\hat t, \tilde{q}, \hat{S}, \hat{S}, \tilde{\psi},\tilde{I}) \\
  & = \max_{(t,q,S,\psi,I)} \tilde{v}( t, q, S, \psi, I) - \tilde{u}( t, q, S, \psi, I) - \phi( t, q, S, S, \psi,I),
\end{align*}
where $\epsilon>0$ is small enough and 
\begin{align*}
  \phi( t, S, R) = \epsilon\exp(-\tilde{\kappa} t)\big( 1 + \|S\|_2^{2p} + \|R\|_2^{2p} \big), \quad \tilde{\kappa}>0.
\end{align*}
Since $\tilde{u}\geq \tilde{v}$ on $\{T\}\times\mathbb{R}\times\mathcal{D}$, we directly have $\hat{t}<T$. \\

For all $i\in \mathbb{N}$, we can find a sequence $(t_i,S_i,R_i)$ such that
\begin{align*}
0 & <\tilde{v}(t_i,\tilde{q}, S_i, \tilde{\psi},\tilde{I}) - \tilde{u}(t_i,\tilde{q}, R_i, \tilde{\psi},\tilde{I}) - \phi(t_i,S_i,R_i) - i|S_i - R_i|^2 - \left( |t_i-\hat{t}|^2 + |S_i - \hat{S}|^4 \right)\\
& = \underset{(t,S,R) }{\max} \tilde{v}(t,\tilde{q}, S, \tilde{\psi},\tilde{I}) - \tilde{u}(t,\tilde{q}, R, \tilde{\psi},\tilde{I}) - \phi(t,S,R) - i|S - R|^2 - \left( |t-\hat{t}|^2 + |S - \hat{S}|^4 \right).
\end{align*}
Then we have:
\[(t_i,S_i,R_i) \underset{i\rightarrow +\infty}{\longrightarrow} (\hat{t},\hat{S},\hat{S})\]
up to a subsequence, and
\begin{align*}
&\tilde{v}(t_i,\tilde{q}, S_i, \tilde{\psi},\tilde{I}) - \tilde{u}(t_i,\tilde{q}, R_i, \tilde{\psi},\tilde{I}) - \phi(t_i,S_i,R_i) - i|S_i - R_i|^2 - \left( |t_i-\hat{t}|^2 + |S_i - \hat{S}|^4 \right)\\
& \underset{n\rightarrow +\infty}{\longrightarrow } \tilde{v}(\hat t,\tilde{q}, \hat S, \tilde{\psi},\tilde{I}) - \tilde{u}(\hat t,\tilde{q}, \hat S, \tilde{\psi},\tilde{I}) - \phi(\hat t,\hat S,\hat S)
\end{align*}
Let us then denote for $i\in \mathbb{N}^*$ $$\varphi_i(t,S,R)= \phi(t,S,R) + i|S-R|^2 + |t-\hat{t}|^2 + |S-\hat{S}|^4 \quad \forall \ (t,S,R) \in [0,T]\times \mathbb{R}^2.$$
Then Ishii's Lemma (see \cite{barles2008second, crandall1992user}) guarantees that for all $\eta>0,$ we can find $(y^1_i,p^1_i,A^1_i) \!\in\! \bar{\mathcal{P}}^+ \tilde{v}(t_i,\tilde{q},S_i,\tilde{\psi},\tilde{I})$ and $(y^2_i,p^2_i,A^2_i) \in \bar{\mathcal{P}}^- \tilde{u}(t_i,\tilde{q},R_i,\tilde{\psi},\tilde{I})$ such that:
\[y^1_i - y^2_i = \partial_t \varphi_i(t_i,S_i,R_i), \quad (p^1_i,p^2_i) = \left(\partial_S \varphi_i, -\partial_R \varphi_i \right)(t_i,S_i,R_i),\]
and
\[\begin{pmatrix} A^1_i & 0\\ 0 & -A^2_i \end{pmatrix} \leq H_{SR} \varphi_i (t_i,S_i,R_i) + \eta \left( H_{SR} \varphi_n (t_i,S_i,R_i) \right)^2, \]
where $H_{SR}\varphi_i (t_i,.,.)$ denotes the Hessian matrix of $\varphi_i(t_i,.,.).$ Applying Lemma \ref{viscolemma}, we get
\begin{align*}
\rho& \left( \tilde{v}(t_i,\tilde{q}, S_i, \tilde{\psi},\tilde{I}) - \tilde{u}(t_i,\tilde{q}, R_i, \tilde{\psi},\tilde{I}) \right) \leq  y^1_i - y^2_i + \frac 12 \sigma^2 (A^1_i - A^2_i) + \mu (p^1_i - p^2_i) \\
&\quad + \sum_{\mathbf{k}\in \mathcal{K}} r_{(\tilde \psi,\tilde I),(\mathbf{k}^\psi,\mathbf{k}^I)}\big(\tilde{v}(t_i, \tilde q, S_i,\mathbf{k}^\psi,\mathbf{k}^I)-\tilde{v}( t_i, \tilde q, S_i, \tilde \psi, \tilde I)\big) \\
&\quad + \sup_{p\in Q_\psi, \ell\in \mathcal{A}}\sum_{n=1}^N \lambda^{n}(\tilde \psi,\tilde I,p^n,\ell)\mathbb{E}\Big[\epsilon^{n}\ell^{n}e^{\rho t_i}(S_i+\frac{\tilde{\psi}^{n}}{2}+p^{n}\delta^n)+\tilde{v}\big(t_i,\tilde{q}-\ell^{n}\epsilon^{n},S_i,\tilde \psi,\tilde I\big)- \tilde{v}(t_i,\tilde q,S_i,\tilde \psi,\tilde I)\Big] \\
&\quad - \sum_{\mathbf{k}\in \mathcal{K}} r_{(\tilde \psi,\tilde I),(\mathbf{k}^\psi,\mathbf{k}^I)}\big(\tilde{u}(t_i, \tilde q, R_i,\mathbf{k}^\psi,\mathbf{k}^I)-\tilde{u}( t_i, \tilde q, R_i, \tilde \psi, \tilde I)\big)\\
&\quad -\sup_{p\in Q_\psi, \ell\in \mathcal{A}} \sum_{n=1}^N \lambda^{n}(\tilde \psi,\tilde I,p^{n},\ell)\mathbb{E}\Big[\epsilon^{n}\ell^{n}e^{\rho t_i}(R_i+\frac{\tilde{\psi}^{n}}{2}+p^{n}\delta^n)+\tilde{u}\big(t_i,\tilde{q}-\ell^{n}\epsilon^{n},R_i,\tilde \psi,\tilde I\big)- \tilde{u}(t_i,\tilde q,R_i,\tilde \psi,\tilde I)\Big].
\end{align*}
Moreover, we have $$ H_{SR} \varphi_i (t_i,S_i,R_i) = \begin{pmatrix} \partial^2_{SS} \phi(t_i,S_i,R_i) + 2i +12 (S_i-\hat{S})^2 & \partial^2_{SR} \phi(t_i,S_i,R_i)-2i\\ \partial^2_{SR} \phi(t_i,S_i,R_i)-2i & \partial^2_{SR}\phi(t_i,S_i,R_i) + 2i \end{pmatrix},$$
and 
\begin{align*}
  & \partial_S \varphi_i (t_i,S_i,R_i) = \partial_S \phi(t_i,S_i,R_i) +2i|S_i-R_i| +  4 |S_i - \hat{S}|^3, \\
  & \partial_R \varphi_i (t_i,S_i,R_i) = \partial_R \phi(t_i,S_i,R_i) -2i|S_i-R_i| ,
\end{align*}
so from what precedes we can write
\begin{align*}
\rho & \left( \tilde{v}(t_i,\tilde{q},S_i,\tilde{\psi},\tilde{I}) - \tilde{u}(t_i,\tilde{q},R_i,\tilde{\psi},\tilde{I}) \right) \leq  \partial_t \phi(t_i,S_i,R_i) + 2(t_i-\hat{t}) + \mu\big(\partial_S \phi(t_i,S_i,R_i) + \partial_R \phi(t_i,S_i,R_i) \\
& +4(S_i-\hat{S})^3\big)\\
& + \frac 12 \sigma^2 \left(\partial^2_{SS}\phi(t_i,S_i,R_i) + \partial^2_{RR}\phi(t_i,S_i,R_i) + 2\partial^2_{SR}\phi(t_i,S_i,R_i) + 12(S_i-\hat{S})\right) + \eta C_i\\
& + \sum_{\mathbf{k}\in \mathcal{K}} r_{(\tilde \psi,\tilde I),(\mathbf{k}^\psi,\mathbf{k}^I)}\big(\tilde{v}(t_i, \tilde q, S_i,\mathbf{k}^\psi,\mathbf{k}^I)-\tilde{v}( t_i, \tilde q, S_i, \tilde \psi, \tilde I)\big) \\
& + \sup_{p\in Q_\psi, \ell\in \mathcal{A}}\sum_{n=1}^N\lambda^{n}(\tilde \psi,\tilde I,p^{n},\ell)\mathbb{E}\Big[\epsilon^{n}\ell^{n}e^{\rho t_i}(S_i+\frac{\tilde{\psi}^{n}}{2}+p^{n}\delta^n)+\tilde{v}\big(t_i,\tilde{q}-\ell^{n}\epsilon^{n},S_i,\tilde \psi,\tilde I\big)- \tilde{v}(t_i,\tilde q,S_i,\tilde \psi,\tilde I)\Big] \\
& - \sum_{\mathbf{k}\in \mathcal{K}} r_{(\tilde \psi,\tilde I),(\mathbf{k}^\psi,\mathbf{k}^I)}\big(\tilde{u}(t_i, \tilde q, R_i,\mathbf{k}^\psi,\mathbf{k}^I)-\tilde{u}( t_i, \tilde q, R_i, \tilde \psi, \tilde I)\big) \\
& - \sup_{p\in Q_\psi, \ell\in \mathcal{A}}\sum_{n=1}^N \lambda^{n}(\tilde \psi,\tilde I,p^{n},\ell)\mathbb{E}\Big[\epsilon^{n}\ell^{n}e^{\rho t_i}(R_i+\frac{\tilde{\psi}^{n}}{2}+p^{n}\delta^n)+\tilde{u}\big(t_i,\tilde{q}-\ell^{n}\epsilon^{n},R_i,\tilde \psi,\tilde I\big)- \tilde{u}(t_i,\tilde q,R_i,\tilde \psi,\tilde I)\Big],
\end{align*}
where $C_i$ does not depend on $\eta.$ As $\tilde{v}$ is u.s.c., $\tilde{u}$ is l.s.c. and $(t_i,S_i,R_i)_i$ is convergent, when $\eta \to 0$ it is clear, that, when $i\rightarrow +\infty$, for a certain constant $M$ we get
\begin{align*}
\rho \big( \tilde{v}(\hat{t},\tilde{q},\hat{S},\tilde{\psi},\tilde{I}) -& \tilde{u}(\hat{t},\tilde{q},\hat{S},\tilde{\psi},\tilde{I}) \big) \leq  \partial_t \phi(\hat{t},\hat{S},\hat{S}) + \mu\big(\partial_S \phi(\hat{t},\hat{S},\hat{S}) + \partial_R \phi(\hat{t},\hat{S},\hat{S}) \big)\\
& + \frac 12 \sigma^2 \left(\partial^2_{SS}\phi(\hat{t},\hat{S},\hat{S}) + \partial^2_{RR}\phi(\hat{t},\hat{S},\hat{S}) + 2\partial^2_{SR}\phi(\hat{t},\hat{S},\hat{S}) \right) + M .
\end{align*}
For $\tilde \kappa>0$ large enough, the right-hand side is strictly negative, and as $\rho>0$ we get $$\tilde{v}(\hat{t},\tilde{q}, \hat{S}, \tilde{\psi}, \tilde{I}) - \tilde{u}(\hat{t},\tilde{q}, \hat{S}, \tilde{\psi}, \tilde{I})<0,$$ which yields to a contradiction. The proof for the other part of the HJBQVI is direct. \\

With the two above propositions, it is easy to conclude the proof of the theorem. Indeed, as $v_\star$ is a supersolution such that $v_\star \geq v$ on $\{T\} \times \mathbb{R} \times \mathcal{D},$ and $v^\star$ is a subsolution such that $v^\star \leq v$ on $\{T\} \times \mathbb{R} \times \mathcal{D},$ we can apply the maximum principle to get $v_\star \geq v^\star$ on $[0,T] \times \mathbb{R} \times \mathcal{D}.$ But by definition of $v_\star$ and $v^\star,$ we must have $v_\star \leq v \leq v^\star$ on $[0,T] \times \mathbb{R} \times \mathcal{D},$ which proves that we have $v_\star = v = v^\star$ and $v$ is continuous. The maximum principle implies that if two continuous viscosity solutions of the HJBQVI satisfy the same terminal condition, they are equal on $[0,T] \times \mathbb{R} \times \mathcal{D},$ hence the uniqueness.
\end{proof}

\pagebreak

\section{Application to OTC market making}
\subsection{Framework}
The model we present in this article is designed for trading in cross-listed stocks in limit order books. However, it can be adapted straightforwardly to handle the problem of an OTC market maker, who often deals with a large number of assets driven by a few factors. We borrow here the factorial method market making framework of~\cite{bergault2019size} (we are also going to keep their notation only for this section). We consider a market maker who is in charge of providing bid and ask quotes on $d$ assets, whose dynamics are
\begin{align*}
  dS_t^i = \mu^i dt + \sigma^i dW_t^i, \quad i\in \{1,\dots,d\},
\end{align*}
where $\mu^i$ is the drift of the $i$-th asset, $\sigma^i$ is its volatility and $(W_t^1,\dots,W_t^d)$ is a $d$-dimensional Brownian motion. We consider a non-singular variance-covariance matrix $\Sigma=(\rho^{i,j}\sigma^i \sigma^j)_{i,j\in\{1,\dots,d\}}$ for the vector of assets $(S_t^1,\dots,S_t^d)$. The market maker sets bid and ask prices on every asset:
\begin{align*}
  S^{i,b}(t,z) = S_t^i -\delta^{i,b}(t,z), \quad S^{i,a}(t,z) = S_t^i +\delta^{i,a}(t,z), \quad z\in \mathbb{R},
\end{align*}
where $\delta=(\delta^{i,a},\delta^{i,b})_{i\in\{1,\dots,d\}}$ are the (predictable and uniformly lower bounded) bid and ask quotes around the mid-price of each asset. The volume of transactions on the bid and ask sides are modeled by marked point processes $N^{i,b}(dt,dz),N^{i,a}(dt,dz)$ of intensity $\nu_t^{i,b}(dz),\nu_t^{i,a}(dz)$ defined by
\begin{align*}
  \nu_t^{i,j}(dt,dz) = \Lambda^{i,j}\big(\delta^{i,j}(t,z)\big)\eta^{i,j}(dz), \quad i\in\{1,\dots,d\}, 
\end{align*}
where $\Lambda^{i,j}$ is a sufficiently regular function (exponential, logistic, SU Johnson etc.) modeling the probability to trade on the asset $i$, on the side $j$ for a given spread $\delta^{i,j}(t,z)$ and a size $z$. The functions $\eta^{i,j}(dz)$ are probability densities over $\mathbb{R}_+$ modeling the distribution of a trade size. The market maker manages his inventory process $q_t = (q_t^1,\dots,q_t^d)$ of dynamics given by
\begin{align*}
  dq_t^i = \int_{\mathbb{R}_+} z N^{i,b}(dt,dz) - \int_{\mathbb{R}_+} z N^{i,a}(dt,dz), \quad i\in \{1,\dots,d\}.
\end{align*}
The market maker manages his cash process given at time $t$ by
\begin{align*}
  dX_t = \sum_{i=1}^d \int_{\mathbb{R}_+} z S^{i,a}(t,z) N^{i,a}(dt,dz) - \int_{\mathbb{R}_+} z S^{i,b}(t,z) N^{i,b}(dt,dz).
\end{align*}
Its optimization problem is defined as
\begin{align*}
  \sup_{\delta}\mathbb{E}\Big[X_T + \sum_{i=1}^d q_T^i S_T^i - \int_0^T \phi(q_t) dt \Big],
\end{align*}
where $\phi$ is a running penalty preventing from too large positions and $\sum_{i=1}^d q_T^i S_T^i$ is the marked-to-market value of the market maker's portfolio at time $t$. The corresponding HJB equation is given by
\begin{align*}
  0 = & \partial_t v(t,q) + \sum_{i=1}^d q^i \mu^i - \phi(q) + \sum_{i=1}^d \int_{\mathbb{R}_+} z H^{i,b}\big(\frac{v(t,q)-v(t,q+z e^i)}{z}\big)\eta^{i,b}(dz) \\
  & + \sum_{i=1}^d \int_{\mathbb{R}_+} z H^{i,a}\big(\frac{v(t,q)-v(t,q-z e^i)}{z}\big)\eta^{i,a}(dz),
\end{align*}
with terminal condition $v(T,q)=0$, $H^{i,j}(p)=\sup_{\delta} \Lambda^{i,j}(\delta) (\delta - p)$, and $(e^1,\dots,e^d)$ is the canonical basis of $\mathbb{R}^d$. 

\subsection{Bayesian update for OTC market makers}\label{sec_Bayesian_OTC}

Usually, the functions $\Lambda^{i,j}$ are of the form
\begin{align*}
  \Lambda^{i,j}\big(\delta^{i,j}(t,z)\big) = \lambda^{i,j}_{\text{RFQ}}f\big(\delta^{i,j}(t,z)\big),
\end{align*}
where 
$\lambda^{i,j}_{\text{RFQ}}$ is the constant intensity of arrival of requests for quote, and $f\big(\delta^{i,j}(t,z)\big)$ gives the probability that a request will result in a transaction given the quote $\delta$ proposed by the market maker. The estimation of the quantity $\lambda^{i,j}_{\text{RFQ}}$ is of particular importance for the market maker so that he can adjust his quotes depending on his view on the number of request for a certain asset and a certain side. In the same spirit as in Section \ref{section_learn_intensities}, we assume the following prior distribution: 
\begin{align*}
  \lambda^{i,j}_{\text{RFQ}} \sim \Gamma(\alpha^{i,j},\beta^{i,j}), \quad (\alpha^{i,j},\beta^{i,j})>0.
\end{align*}
For an asset $i\in\{1,\dots,d\}$ on the side $j\in \{a,b\}$, this corresponds to an average intensity of $\frac{\alpha^{i,j}}{\beta^{i,j}}$, with variance equal to $\frac{\alpha^{i,j}}{(\beta^{i,j})^2}$. If the market maker is confident in his estimation of the intensity $\lambda^{i,j}_{\text{RFQ}}$, he can choose a large $\beta^{i,j}$ so that the variance of his Bayesian estimator is small. Given all the information accumulated up to time $t$, its best estimation of the quantity $\lambda^{i, j}_{\text{RFQ}}$, is given by
\begin{align}\label{lambda_rfq_update}
  \mathbb{E}\big[\lambda^{i,j}_{\text{RFQ}}|N(t,dz)\big]=\frac{\alpha^{i,j}+ \int_{\mathbb{R}_+} N(t,dz)}{\beta^{i,j} + \int_{\mathbb{R}_+}\int_0^t f(\delta^{i,j}(s,z))ds\;\eta^{i,j}(dz)}. 
\end{align}
By the law of large numbers, when the market maker has accumulated a sufficiently large number of observations, his best estimation of $\lambda^{i,j}_{\text{RFQ}}$ converges to the ``real'' intensity of the market. As time passes, the prior parameters $(\alpha^{i,j},\beta^{i,j})$ of the market maker are less important as the estimation will rely mostly on the observations. \\

Another important parameter of the model is the size of transactions, which impacts the quotes of the market maker as well as his inventory risk. In~\cite{bergault2019size}, the authors choose in their numerical experiments a $\Gamma(a^{i,j},b^{i,j})$ distribution for $\eta^{i,j}$. The trader can choose between Bayesian updates (revise only $a^{i,j}$, only $b^{i,j}$, or both), depending on his confidence on parameters' estimation. If he is confident with respect to the shape parameter $a^{i,j}$, that is he knows approximately the average size of a request but not the standard deviation, he sets $b^{i,j}\sim \Gamma(a^{i,j}_0,b^{i,j}_0)$. Given $n$ observations of size $z^1, \dots, z^n$, the best Bayesian estimate of $b^{i,j}$ (the scale parameter of the size of the request) is
\begin{align}\label{bayesian_b_gamma_OTC}
  \mathbb{E}[b^{i,j}|(z^1, \dots, z^n)] = \frac{a_0^{i,j} + n a^{i, j}}{b_0^{i, j} + \sum_{i=1}^n z^i}. 
\end{align}
The use of different prior distribution to take into account the uncertainty on the shape parameter $a^{i, j}$ (if $b^{i, j}$ is known) or on both $(a^{i, j}, b^{i, j})$ can be done in the same way. \\

Another sensitive parameter, especially for the multi-asset market making, is the variance-covariance matrix $\Sigma$. This quantity is usually estimated on a long run, but parameters are subject to a brutal change. For example, let us assume that the market maker is in charge of $d$ assets on $2$ different sectors (for instance, technology and aerospace). Following the factorial approach, the market making problem's dimension will be reduced from $d$ to $3$. The three factors mainly correspond to the three highest eigenvalues of the variance-covariance matrix $\Sigma$, and will drive the quotes of the market maker. However, in case of a sectorial tail event, for example the bankruptcy of one of the companies of the tech sector, it is likely that all the correlations between the assets of this sector will rise to one. This will impact the eigenvalue related to the technology sector, and change the quotes of the market maker as he has to avoid long inventory positions on assets whose values are decreasing. To design adaptive market making strategy based on Bayesian update of the correlation matrix and the drift of the assets, we define the Normal-Inverse-Wishart prior on $(\mathbf{\mu},\Sigma)\sim \text{NIW}(\mu_0, \kappa_0,\nu_0,\mathbf{\psi})$, where $(\mu_0, \kappa_0,\nu_0,\mathbf{\psi})\in \mathbb{R}^d \times \mathbb{R}_+^\star \times (d-1,+\infty)\times \mathcal{M}_d(\mathbb{R})$. This distribution is built as follows:
\begin{align*}
  \mathbf{\mu} | (\mathbf{\mu_0},\kappa_0,\Sigma) \sim \mathcal{N}\big(\mathbf{\mu_0}, \frac{1}{\kappa_0}\Sigma\big), \quad \Sigma| (\psi,\nu_0) \sim \mathcal{W}^{-1}(\psi,\nu_0), \text{ then } (\mathbf{\mu},\Sigma)\sim \text{NIW}(\mu_0, \kappa_0,\nu_0,\mathbf{\psi}),
\end{align*}
where $\mathcal{W}^{-1}$ is the standard inverse Wishart distribution. In other words, the drift vector $\mathbf{\mu}$ of the assets follows a multivariate Gaussian distribution whereas the variance-covariance matrix $\Sigma$ follows a standard inverse Wishart distribution. At time $t$, if we note $S_t=(S_t^1,\dots,S_t^d)$ the prices observed up to time $t$, the Bayesian update of $(\mathbf{\mu},\Sigma)$ is
\begin{align*}
  (\mu,\Sigma|S_t-S_0) \sim \text{NIW}\Big( &\frac{\kappa_0 \mathbf{\mu_0}+(S_t-S_0)}{\kappa_0 +t}, \kappa_0+t,\nu_0+t , \\
  & \psi + (S_t - \frac{S_t}{t})(S_t - \frac{S_t}{t})^{\mathbf{T}} + \frac{\kappa_0 t}{\kappa_0 + t}(\mathbf{\mu_0} - \frac{S_t}{t})(\mathbf{\mu_0} - \frac{S_t}{t})^{\mathbf{T}} \Big). 
\end{align*}
Following the law of large numbers, as $t\rightarrow +\infty$ we have a larger number of information and we converge toward the drift and variance-covariance of the market maker's portfolio. Therefore, the market maker can recompute his factors derived from the updated variance-covariance matrix and adjust his quotes. \\

This extension deserves several remarks. First, the problems encountered by an OTC market maker are quite different from a high-mid frequency trader in an order book. The model is more parsimonious, especially for the intensity functions. Therefore, the convergence toward the ``true'' market parameters will be faster than in order book model. The objective of the Bayesian update on the quantities $\lambda^{i,j}_{\text{RFQ}}$ is to determine the average behavior or the counterparts of the market maker. If he observes a large number of requests on the ask (resp. bid) side of the asset $i$, the Bayesian update \eqref{lambda_rfq_update} enables the market maker to adjust his quotes to set a higher ask (resp. bid) price for this asset. If the market maker observes a higher discrepancy than expected for the transaction sizes, the Bayesian update~\eqref{bayesian_b_gamma_OTC} helps to adjust his quotes. Finally, the Bayesian learning on the drift and covariance of the assets enables to update the factors from which the market maker chooses his quotes. 

\endgroup

\pagebreak
\bibliographystyle{abbrv}
\bibliography{biblio.bib}

\end{document}